\documentclass[final,onefignum,onetabnum]{siamart220329}
\usepackage{cite}
\usepackage{xspace}  
\usepackage{mathtools}
\usepackage{amssymb}
\usepackage{tcolorbox}
\usepackage{comment}
\usepackage{bbm}
\usepackage{enumitem} % simple package for making tighter lists. 
\usepackage{wrapfig}
\usepackage{mathdots}
\usepackage{thmtools} 
\usepackage{thm-restate}
\usepackage{times}
\usepackage[normalem]{ulem} % use \sout
\usepackage{url}
\usepackage{footmisc}

\renewcommand{\texttt}[1]{\xspace{\normalfont{\oldtexttt{#1}}}\xspace}
\newtheorem{fact}{Fact}
\renewcommand{\paragraph}[1]{\vspace{3pt}{\it #1}} 

\newcommand{\cint}{c_{int}\xspace}
\newcommand{\defn}[1]{\textbf{\emph{\boldmath #1}}}
\newcommand{\passive}{P\xspace}

\newcommand{\jams}{\mathcal{J}\xspace}

\newcommand{\clow}{C_{\rm low}\xspace}
\newcommand{\chigh}{C_{\rm  high}\xspace}

\newcommand{\whp}{w.h.p.\xspace}
\newcommand{\polylog}{\operatorname{polylog}}
\newcommand{\poly}{\text{poly}}
\newcommand{\OurAlg}{\textsc{Low-Sensing Backoff}\xspace} % placeholder for algorithmn name

\renewcommand{\epsilon}{\varepsilon}

\newcommand{\secref}[1]         {Section~\ref{sec:#1}}

\newcommand{\thmref}[1]         {Theorem~\ref{thm:#1}}
\newcommand{\thmreftwo}[2]      {Theorems~\ref{thm:#1} and~\ref{thm:#2}}

\newcommand{\thmlabel}[1]   {\label{thm:#1}}

\newcommand{\lemlabel}[1]   {\label{lem:#1}}
\newcommand{\lemref}[1]         {Lemma~\ref{lem:#1}}
\newcommand{\lemreftwo}[2]{Lemmas~\ref{lem:#1} and~\ref{lem:#2}}
\newcommand{\lemrefthree}[3]    {Lemmas~\ref{lem:#1}, \ref{lem:#2} and~\ref{lem:#3}}

\newcommand{\seclabel}[1]{\label{sec:#1}}

\newcommand{\corref}[1]         {Corollary~\ref{cor:#1}}

\newcommand{\wmin}{w_{\rm min}}
\newcommand{\wmax}{w_{\max}(t)\xspace}

\newcommand{\win}{d_1}
\newcommand{\lose}{d_2}
\newcommand{\Con}{C(t)\xspace} 
\newcommand{\first}{N\xspace}
\newcommand{\second}{H\xspace}
\newcommand{\third}{L\xspace}
\newcommand{\psuccess}{p_{\rm suc}\xspace}
\newcommand{\pempty}{p_{\rm emp}\xspace}
\newcommand{\pcollision}{p_{\rm noi}\xspace}
\newcommand{\psuccessgood}{p_{{\rm suc}|{\rm good}}\xspace}
\newcommand{\pemptygood}{p_{{\rm emp}|{\rm good}}\xspace}
\newcommand{\pemptylow}{p_{{\rm emp}|{\rm low}}\xspace}
\newcommand{\pcollisionhigh}{p_{{\rm noi}|{\rm high}}\xspace}
\newcommand{\E}{\mathbb{E}}
\newcommand{\expect}[1]         {{\rm E}\left[ #1 \right]}

\newtheorem{observation}[theorem]{Observation}

\newif\ifcomments
\commentstrue
\commentsfalse

\newcommand{\john}[1]{\ifcomments {\noindent \scriptsize  \textcolor{blue} {John: {#1}}} \fi{}}
\newcommand{\jeremy}[1]{\ifcomments {\noindent \scriptsize  \textcolor{blue} {Jeremy: { #1}}} \fi{}}

\newcommand{\michael}[1]{\ifcomments {\noindent \scriptsize  \textcolor{purple} {Michael: {#1}}} \fi{}}
\newcommand{\maxwell}[1]{\ifcomments {\noindent \scriptsize  \textcolor{magenta} {Maxwell: {#1}}} \fi{}}
\newcommand{\mab}[1]{\michael{#1}}  

\begin{tcbverbatimwrite}{tmp_\jobname_header.tex}
\title{Fully Energy-Efficient Randomized Backoff: Slow Feedback Loops Yield Fast Contention Resolution\thanks{A preliminary version of this work appeared in the proceedings of the  43rd ACM Symposium on Principles of Distributed Computing (PODC'24) \cite{bender2024fully}. \funding{This research was supported in part by the National Science Foundation (NSF) grants CCF-1918989, CCF-2106759, CCF-2144410, CNS-2210300, 
CCF-2247577, % MAB: HI grant
CCF-2106827, % MAB: Flatland
and by Singapore MOE-T2EP20122-0014.}}}

\author{Michael A. Bender\thanks{Department of Computer Science, Stony Brook University, Stony Brook, NY, USA  (\email{bender@cs.stonybrook.edu}); corresponding author.}
\and Jeremy T. Fineman\thanks{Department of Computer Science, Georgetown University, Washington, DC, USA   (\email{jfineman@cs.georgetown.edu}).} 
\and Seth Gilbert\thanks{School of Computing, National University of Singapore, Singapore (\email{seth.gilbert@comp.nus.edu.sg}).} 
\and\hspace{10cm}John Kuszmaul\thanks{Electrical Engineering \&
Computer Science Department, Massachusetts Institute of Technology, Cambridge, MA, USA, and Google, Cambridge, MA, USA (\email{john.kuszmaul@gmail.com}).} 
\and Maxwell Young\thanks{Department of Computer Science and Engineering, Mississippi State University, Mississippi State, MS, USA (\email{myoung@cse.msstate.edu}).} 
}

\headers{Fully Energy-Efficient Randomized Backoff}{M. A. Bender, J. T. Fineman, S. Gilbert, J. Kuszmaul, and M. Young}
\end{tcbverbatimwrite}
\input{tmp_\jobname_header.tex}

\begin{document}

%\thispagestyle{empty} % will suppress page only on cover letter, I have checked the other pages
%\begin{center}{\large\bf Cover Letter}\end{center}

%\medskip\medskip

%\noindent{\vspace{1cm}Dear Faith, \hfill July 9, 2024\\}

%\vspace{-30pt}

%\noindent{Sincerely,}
%\vspace{-0pt}

%\noindent---~Michael, Jeremy, Seth, John, and Max

%\clearpage
\setcounter{page}{1}

\maketitle

%%%%%%%%%%%%%%%%%%%%%%%%%%%%%%%%%%%%%%%%%%%%%%%%%%%%%%

% This handles issues with references going outside of the margins [shakes fist at SICOMP format]. However, it also results in occasional weird spacing issues, which probably need to be handled case by case.
\setlength\emergencystretch{\hsize}

\begin{abstract}
Contention resolution addresses the problem of coordinating access to a shared communication channel. Time is discretized into synchronized  slots, and a packet transmission can be made in any slot. A packet is successfully sent if no other packet is also transmitted during that slot. If two or more packets are sent in the same slot, then these packets collide and fail.  Listening on the channel during a slot provides ternary feedback, indicating whether that slot had (0)~silence, (1)~a successful transmission, or (2+)~noise. 
No other feedback or exchange of information is available to packets.
Packets are (adversarially) injected into the system over time. 
A packet departs the system once it is successfully sent. The goal is to send all packets while optimizing throughput, which is roughly the fraction of successful slots.  

Most prior contention resolution algorithms with constant throughput require 
a short feedback loop, in the sense that a packet’s sending probability in slot $t + 1$ is fully determined by its internal state at slot $t$ and the channel feedback at slot $t$. 
This paper answers the question of whether these  short feedback loops are necessary; that is, how often must listening and updating occur in order to achieve constant throughput?
We can restate this question in terms of energy efficiency:  given that both  listening and sending consume significant energy,
is it possible to have a contention-resolution algorithm with ternary feedback that is efficient for both operations?

A shared channel can also suffer random or adversarial noise, 
which causes any listener to hear noise, even when no packets are actually sent. Such noise arises due to hardware/software failures or malicious interference (all modeled as ``jamming''), which can have a ruinous effect on the throughput and energy efficiency. % guarantees of a protocol. 
How does  noise affect our goal of long feedback loops/energy efficiency? 

Tying these questions together, we ask: \emph{what does a contention-resolution algorithm have to sacrifice to reduce channel accesses?} Must we give up on constant throughput? What about robustness to noise? Here, we show that we need not concede anything by presenting an algorithm with the following guarantees. Suppose there are $N$ packets arriving over time and $\jams$ jammed slots, where the input is determined by an adaptive adversary. With high probability in $N+\jams$, our algorithm guarantees $\Theta(1)$ throughput and $\polylog(N+\jams)$ channel accesses (sends or listens) per packet. We also have analogous guarantees when the input stream is infinite---we prove implicit throughput bounds of $\Omega(1)$ for all time slots $t$, and this translates to $\Theta(1)$ guaranteed throughput for any slot $t$ where the implicit throughput is sufficiently small in $\Theta(1)$. As a special case, these throughput results give rise to adversarial-queuing theory guarantees.
\end{abstract}

\maketitle

\section{Introduction}\label{sec:intro}

Since the 1970s, randomized 
\defn{backoff}  protocols %~\cite{Goldberg-notes-2000,HastadLeRo96,MetcalfeBo76,AbramsonKu73,802.11-standard},
such as binary exponential backoff~\cite{MetcalfeBo76}, have been used for managing contention on a shared communication channel.
Originally used in the ALOHA system~\cite{ALOHAnet} and Ethernet~\cite{MetcalfeBo76}, randomized backoff plays an important role in a wide range of applications, including  WiFi \cite{802.11-standard}, wireless sensor networks \cite{khanafer2013survey},  transactional memory~\cite{scherer2005advanced,HerlihyMo93}, and congestion control~\cite{wierman2003unified}. The salient feature of the communication
channel is that it supports only one message transmission at a time: if more than one message is sent simultaneously, there is a \defn{collision} resulting in indecipherable noise~\cite{Goldberg-notes-2000,kurose:computer,xiao:performance,DBLP:journals/jacm/GoldbergMPS00,MetcalfeBo76}. 

% Removed references
This \defn{contention-resolution problem} is formalized as follows. %~\cite{hastad:analysis,goldberg2004bound,kelly1987number,aldous:ultimate,GoodmanGrMaMa88,mosely1985class,RaghavanU98,chlebus2009maximum,BenderGiKu22,zhou2019singletons,BenderFiGi19,FinemanNW16:multiple,fineman:contention,DeMarcoSt17,BenderFiGi16,FernandezAntaMoMu13,BenderFaHe05,deMarco:fast,DEMARCO20171,ChangJiPe19,BenderKoKu20,BenderKoPe16}. 
There are $N$ packets  
arriving over time, and each packet needs to succeed (be the only packet sending) on the channel.\footnote{For ease of exposition, we slightly abuse terminology and have the packets themselves taking action (e.g., sending themselves on the channel, listening on the channel), rather than introducing ``agents''/``devices''/``senders'', where each one  arrives on the scene with a packet to transmit.} 
Time is divided into synchronized \defn{slots}, each of which is sized to fit a single packet. To succeed, the packet requires exclusive access to the channel %\cite{ChlamtacKu85,Goldberg-notes-2000,xiao:performance,DBLP:journals/jacm/GoldbergMPS00,hastad:analysis,MetcalfeBo76,AbramsonKu73}; 
that is, the packet must be the only one transmitted during that slot.
Otherwise, if two or more packets are transmitted in the same slot, the result is a collision, where none of the transmitted packets succeed. 
A packet departs the system once it succeeds. There is no {\it a priori} coordinator or central authority; packet transmissions are scheduled in a distributed manner. The objective is to have all packets succeed while optimizing the \defn{throughput} of the channel, which is roughly the fraction of successful slots.

A popular contention-resolution protocol is 
binary exponential backoff~\cite{MetcalfeBo76} (see \cite{kelly:decentralized,kelly1987number,rosenkrantz:some,aldous:ultimate,GoodmanGrMaMa88,GoodmanGrMaMa88,Al-Ammal2000,al2001improved,song:stability}). 
Informally, under binary exponential backoff, a packet that has been in the system for $t$ slots is sent with probability $\Theta(1/t)$.

\paragraph{Feedback loops: short versus long (versus no feedback).} 
An elegant, but ultimately problematic, feature of classical binary exponential backoff  is that it is \defn{oblivious}---{packets  remaining in the system do not use channel feedback to adjust their behavior:} 
a packet with age $t$ sends with probability $\Theta(1/t)$ until it succeeds, regardless of channel history.
The unfortunate result is that 
with adversarial packet arrivals, %\footnote{There is also much work analyzing the \emph{stochastic} arrival rates under which binary exponential backoff is stable or unstable~\cite{HastadLeRo87,aldous:ultimate,GoodmanGrMaMa88}; see  also~\cite{MetcalfeBo76,Goldberg-notes-2000,hastad:analysis,goldberg2004bound,shoch1980measured,abramson1977throughput,kelly1987number,aldous:ultimate,GoodmanGrMaMa88,mosely1985class,RaghavanU98,RaghavanU95,GOLDBERG1999232}.} 
binary exponential backoff  supports only a subconstant throughput---specifically, $O(1/\ln{N})$; in fact, even for the batch case where all $N$ packets arrive at the same time, the throughput of binary exponential backoff is only $O(1/\ln N)$~\cite{BenderFaHe05}.

In contrast, contention-resolution protocols can exploit frequent channel feedback to achieve $\Theta(1)$ throughput under adversarial arrivals~\cite{BenderKoKu20,
BenderFiGi19,
BenderKoPe16,
ChangJiPe19,
awerbuch:jamming}. These protocols are not oblivious---packets \defn{listen} on the channel and adjust their sending probabilities up or down based on this feedback.

% removed refs richa:jamming2,richa:jamming3,richa:jamming4,
The primary model for channel feedback is the \defn{ternary-feedback model}~\cite{BenderFiGi19,DBLP:conf/spaa/AgrawalBFGY20,gilbert:near,DBLP:journals/dc/KingPSY18,chlebus2009maximum,chlebus2007stability,ChlebusKoRo06,fineman:contention2,anantharamu2019packet,greenberg1985lower}. In this model, a packet can listen on the channel in each slot and learn whether that slot is
(0)~\defn{empty}, if no packets send,
(1)~\defn{successful}, if exactly one packet sends, or
(2+)~\defn{noisy}, if two or more packets send.
Based on this feedback, the packet can decide when to attempt to send next. 
Once a packet succeeds, it immediately departs the system.

Most constant-throughput algorithms with ternary-feedback (e.g., \cite{BenderKoKu20,
BenderFiGi19,
BenderKoPe16,
ChangJiPe19,DeMarcoSt17,awerbuch:jamming,richa:jamming2,richa:jamming4,richa:jamming3,ogierman:competitive}) listen on the channel in every slot (or every constant number of slots).  That is, these algorithms have a  short feedback loop: a packet's sending  probability in slot $t+1$ is fully determined by its internal state at slot $t$ and the channel feedback at slot $t$. For example, the algorithm by Chang, Jin, and Pettie~\cite{ChangJiPe19}, 
listens in every slot $t$ and 
multiplicatively updates the sending probability in slot  $t+1$ based on whether it heard silence, a successful transmission, or a collision in slot $t$.  
 
A key question is as follows: 
\emph{how frequently does the packet need to listen on the channel and update its behavior  in order to achieve constant throughput?}  
In every slot? 
In a vanishingly small fraction of slots? 
How short a feedback loop is necessary for good throughput?

As an analogy, one cannot navigate a ship without a control feedback loop: monitoring the surroundings and correcting course. One option is to  constantly monitor  and continuously update the heading to avoid obstacles. 
But the analogous question is whether one can still safely navigate with only a vanishingly small amount of course correction.
In contention resolution, the monitoring is via the channel sensing and the 
course being corrected is the transmission probabilities.

% Removed references 

\paragraph{Robustness to noise.} Finally, in addition to the factors discussed above, much of the recent work on contention resolution (see \cite{anantharamu2019packet,Chlebus:2016:SWM:2882263.2882514,DBLP:conf/spaa/AgrawalBFGY20,ChenJiZh21,jiang:robust,awerbuch:jamming,richa:jamming2,richa:jamming4,richa:jamming3,ogierman:competitive,BenderFiGi19,ChangJiPe19}) has sought to address an additional factor, which is that the real world is often noisy. 
Sometimes interference prevents transmissions in a slot and listeners hear noise even if nobody actually sends. Noisy channels arise due to hardware/software failures, co-located devices, or malicious jamming  \cite{chiwewe2015using,liu2017interference,pelechrinis:denial,mpitziopoulos:survey,xu:feasibility}. 
Regardless of the source of the noise in a slot, we may think of these slots as being ``jammed'' by an adversary.
Jamming has evolved from a mostly-theoretical risk into a credible threat to systems over the past decade, with several publicized examples~\cite{nicholl,FCC:humphreys,marriot,priest}.

After a long line of work, there are now many contention-resolution protocols that achieve constant throughput in the presence of noise (e.g.~\cite{awerbuch:jamming,richa:jamming2,richa:jamming4,richa:jamming3,ogierman:competitive,BenderFiGi19,ChangJiPe19}); however, again,  these protocols listen to the channel in every slot and update their sending probabilities accordingly.  Our goal is to eliminate short feedback loops not just in the classical model, but also with jamming. 

\paragraph{Minimizing channel accesses = energy efficiency.}
Up until now, we have discussed whether one can minimize listening and thus avoid providing immediate feedback for a contention-resolution algorithm. Another way of viewing this problem is through the lens of energy efficiency. Each channel access---whether for sending or listening---consumes energy. Most work on contention resolution is \defn{sending efficient}, but is not \defn{listening efficient} (e.g., ~\cite{awerbuch:jamming,richa:jamming2,richa:jamming3,richa:jamming4,BenderFiGi19,BenderKoKu20,DeMarcoSt17,jiang:robust,de2022contention}).   
That is, most protocols optimize how frequently a packet sends, but allow a packet to listen  in every slot ``for free''. 

% Removed references here
In fact, \emph{both} sending and listening are expensive operations (e.g., \cite{polastre:telos,feeney2001investigating, yang2017narrowband,lauridsen2018empirical}), which should be minimized if devices are energy-constrained (e.g., battery-powered devices). Minimizing energy usage by having devices sleep as much as possible has been a long-standing and popular strategy to maximize the network lifetime (for example, the development of duty-cycle  protocols \cite{merlin2010duty,ye2006ultra,li2005energy}). 
Why optimize listening in contention resolution protocols if packets must receive messages (for other network applications)? Informally, there are two types of listening: listening to receive messages, and listening to execute a contention resolution protocol (which is what enables sending messages).  To receive messages, packets do not need to listen in every time slot; there exist methods for optimizing this first type of listening. Although this is beyond the scope of this paper, as one example, in many wireless settings, a central base station can monitor the channel constantly and facilitate message exchange~\cite{feng2012survey}.

We call a protocol \defn{fully energy-efficient} if it is both sending efficient and listening efficient. 
By definition, such a protocol cannot have short feedback loops, since it can access the channel only rarely.

\paragraph{Past work: Minimizing listening by allowing for explicit synchronization.} 
Currently, the only known path to full energy efficiency is via \emph{explicit synchronization}. This means that the model is extended so that packets can send synchronization messages to each other whenever they broadcast \cite{BenderKoPe16,DeMarcoSt17,marco:time}. In~\cite{BenderKoPe16}, these synchronization messages have size $\Theta(\ln N)$ bits each, which means that an arbitrary polynomial amount of communication can be expressed in a slot. In~\cite{DeMarcoSt17,marco:time}, the synchronization messages are smaller ($O(1)$ bits), but some packets are permitted to stick around as ``Good Samaritans'' in order to serve as long-term coordinators (that send many messages via many broadcasts). 
We highlight that the model used in \cite{DeMarcoSt17, marco:time} is directly comparable to our model. The language used in those works frames the problem as devices attempting to successfully send packets on a shared channel. In that language, the devices are allowed to remain in the system issuing additional coordination messages after successfully sending their packet.

Using $\Theta(\ln N)$-bit synchronization messages, Bender, Kopelowitz, Pettie, and Young~\cite{BenderKoPe16} give an algorithm with  $\Theta(1)$ throughput and expected $O(\ln (\ln^*N))$ channel accesses per packet. Using Good-Samaritan packets, De Marco and Stachowiak~\cite{DeMarcoSt17} and De Marco, Kowalski, and Stachowiak~\cite{marco:time} provide a constant-throughput algorithm that is \emph{sending} efficient ($O(\ln N)$ transmissions per packet) and conjecture that their techniques can be extended to achieve fully energy efficiency ($O(\polylog N)$ channel accesses per packet), with high probability and even without collision detection. 

In all of these cases, even with the help of explicit synchronization, it remains open whether one can achieve such results in the presence of adversarial noise. Indeed, noise has the potential to be extra-problematic for energy-efficient algorithms since these algorithms listen to the channel less frequently and can potentially be thrown off by a small amount of well-placed noise. 

\paragraph{This paper.} We show that it is indeed possible to achieve constant throughput with full energy efficiency  while being robust to adversarial noise. Moreover, these results hold in the standard ternary-feedback model, without requiring the addition of any sort of explicit synchronization.

Our algorithm belongs to a natural family of multiplicative-weight-update algorithms (e.g., \cite{ChangJiPe19,awerbuch:jamming,richa:jamming2,richa:jamming3,ogierman:competitive,richa:jamming4}). 
When a packet hears silence, it multiplicatively increases both its listening and sending probabilities. Conversely, when a packet hears noise, it multiplicatively decreases these probabilities. 
There are no control messages,  no Good Samaritan packets, and no leaders elected.  

What makes our algorithm/analysis interesting is that we are able to support a multiplicative-weight-update framework while having each packet `cover its eyes' almost all of the time. This is in stark contrast to prior constant-throughput algorithms, which adjust sending probabilities in every slot. Because each packet listens to so few slots, the different packets that are in the system at the same time may end up with very different perspectives on the world from each other. In order to analyze the `herd behavior' of the packets in this potentially chaotic setting, substantially new techniques end up being required. These techniques are also what allow us to handle the additional chaos that adversarial jamming adds to the system.

\subsection{Model}\label{sec:model}

A \emph{finite} or \emph{infinite} stream of indistinguishable packets arrives over time; the number of arrivals is unknown to the algorithm. Each packet must be sent on the \defn{multiple-access channel}. Time is divided into synchronized slots, each of which is sufficiently large to send a single packet.
An \defn{adversary} (specified below) controls how many packets are injected into the system in each slot. When a packet succeeds, it departs the system. There is no universal numbering scheme for the slots; that is, there is no global clock from which a packet could infer the system lifetime or slot parity. Additionally, the packets do not receive any additional information about how many packets have arrived or will arrive. Instead, packets only receive information through the ternary feedback model.

% Removed references
We now describe the \defn{ternary feedback model}. %~\cite{richa:jamming2,richa:jamming3,richa:jamming4,BenderFiGi19,DBLP:conf/spaa/AgrawalBFGY20,,BenderKoPe18,gilbert:near,DBLP:journals/dc/KingPSY18,chlebus2009maximum,chlebus2007stability,ChlebusKoRo06,fineman:contention2,anantharamu2019packet,greenberg1985lower,willard:loglog}. Initially, we will define the model without jamming. 
In each time slot, each packet in the system can take one of three actions: (i)~sleep, (ii)~send, or  (iii)~listen to the channel.
Packets that take actions~(ii) or~(iii) are said to \defn{access the channel}.
If no packets choose to send during a  slot, then that  (non-jammed) slot is \defn{empty}/\defn{silent}.
%; if exactly one packet sends, then that  (non-jammed) slot is \defn{full}   and  \defn{successful}; if two or more packets send, then that slot is \defn{full}   and \defn{noisy}. 
A packet that listens during a slot (action~iii)  learns whether 
the slot was (0)~empty, (1)~successful, or (2+)~noisy. 
A packet that sleeps during a slot (action~i) learns nothing about the state of the slot.
A packet that sends (action~ii), either \defn{succeeds} and leaves the system, or \defn{collides}  and remains in the system.\footnote{For ease of presentation, in our algorithms, we say that a packet can listen and send simultaneously, but any packet that is sending actually does not need to listen to determine the state of the channel. 
If the packet is still in the system after sending in slot $t$, then slot $t$ was  noisy.}

We now add jamming to the picture; an adversary determines which slots are jammed. 
To jam a particular slot, the adversary broadcasts \defn{noise} into that slot.
All jammed slots are thus \defn{full}  and \defn{noisy}.
A packet that listens in a jammed slot hears that the slot was noisy, but does not know whether that noise came from jamming or was merely a collision between two or more packets.
A packet that sends during a jammed slot \defn{collides} and thus remains in the system.

An \defn{adversary} determines, for each slot $t$, how many packets to inject in slot $t$ and whether to jam in that slot. This paper considers an \defn{adaptive adversary}, which bases its decision on the  entire state of the system so far, i.e., up to the end of slot $t-1$, but not the outcomes of future coin tosses. Thus, if at slot $t$, a packet $p$ decides whether to send based on a coin flip, the adaptive adversary does \emph{not} get to see that coin flip until after slot $t$.  

\paragraph{A basic metric: (overall) throughput.} 
The main objective of contention resolution is to optimize  \defn{throughput}, defined next. A slot is \defn{active} if at least one packet is in the system during that slot; \defn{inactive} slots can be  ignored in our analysis. Similarly, a packet is \defn{active} if has been injected into the system and has not yet departed. Without loss of generality, assume throughout that the first slot is active.  Without jamming, the \defn{throughput at time $t$} is defined as $T_t/S_t$, where \defn{$T_t$} is the number of successful transmissions during slots $1,2,\ldots,t$, and \defn{$S_t$} is the number of active slots in the same time interval.\footnote{By assumption that the first slot is active, we have $S_t \geq 1$ and hence throughput is always well defined.} On a finite input, the \defn{(overall) throughput} is defined with respect to the final active slot $t$; at this point, the overall throughput is $N/S$, where $N=T_t$ is the total number of packets and $S=S_t$ is the total number of active slots.  Overall throughput is not well-defined on an infinite execution.

A deficiency of the throughput metric (when defined na\"\i{}vely) is that even if an algorithm guarantees $\Theta(1)$ overall throughput, it is not possible to achieve $\Theta(1)$ throughput uniformly across time. For example, if there is a burst of $N$ packets at time $0$, and $N$ is unknown to the algorithm, then the throughput will be $0$ for a superconstant number of slots (e.g., see~\cite{willard:loglog,ChangKoPe19}), and this is \emph{provably unavoidable} regardless of the backoff strategy being used. Thus (overall) throughput is only meaningful at the end of the execution, or at points in time where there are no packets in the system.

\paragraph{A stronger metric: implicit throughput \cite{BenderKoKu20}.}
In order to support a meaningful notion of throughput, even at intermediate points in time, Bender, Kopelowitz, Kuszmaul, and Pettie~\cite{BenderKoKu20} propose a refined definition that they term ``implicit throughput''~\cite{BenderKoKu20}. The \defn{implicit throughput} at time $t$ is defined as $N_t/S_t$, where $N_t$ is the total number of packets that arrive at or before time~$t$ and $S_t$ is the total number of active slots so far.\footnote{Notice that while $N_t$ depends on the adversary, the number of active slots depends on the  algorithm, whose goal it is to make slots inactive by completing packets.}

One perspective on implicit throughput is that it is an \emph{analytical tool}. Indeed, whenever we reach a point in time where overall throughput is meaningful (i.e., there are no packets left in the system), the two metrics become provably equal. This includes both at the end of any finite execution or during quiet periods of infinite executions. 
Another perspective on implicit throughput is that it is a stronger metric that offers a meaningful guarantee even at intermediate points in time: what constant implicit throughput means is that the number of active slots used so far should never be asymptotically larger than the number of packets that have arrived. 

% adding this?
Since we ignore slots where no packets are in the system, the implicit throughput of our algorithm (see Section \ref{sec:our-algorithm}) is guaranteed (w.h.p.) to be $\Omega(1)$. Moreover, for every slot $t$ where the implicit throughput up to that slot is sufficiently small in $\Theta(1)$, the standard notion of throughput is also $\Theta(1)$.
%Note that a consequence of our implicit throughput bounds is Corollary 1.5, in which we provide stability results in the adversarial-queuing arrival setting.

Throughout the rest of the paper, we focus exclusively on implicit throughput. We should emphasize, however, that this only makes our results stronger---the results also imply the standard constant-throughput guarantees that one would normally strive for.

\paragraph{Extending to adversarial jamming.} We next extend the definitions of throughput and implicit throughput for the case of jamming following~\cite{BenderFiGi19}. 
An algorithm wastes a slot if that slot has silence or a collision, and throughput measures the fraction of slots that the algorithm could have used but instead wasted. Let $\jams_t$ denote the number of jammed slots through slot $t$. Then, the throughput of an execution ending at time $t$ is defined to be $(T_t + \jams_t)/S_t$, and the implicit throughput at slot $t$ is defined to be $(N_t + \jams_t)/S_t$.

\paragraph{Some useful properties of implicit throughput, and applications to adversarial queuing theory.}
We conclude the section by summarizing several useful properties of implicit throughput.

\begin{observation}[\hspace{.05em}\cite{BenderKoKu20}]\label{obs:implicit-one}
  Consider any inactive slot $t$, i.e., where there are no active packets in the system. Then the implicit throughput and throughput are the same at slot $t$.
\end{observation}

\begin{observation}[\hspace{.05em}\cite{BenderKoKu20}]\label{obs:implicit-two}
    Let $\delta$ be any lower bound on the implicit throughput of an algorithm; that is, suppose the algorithm achieves implicit throughput of at least $\delta$ at all times.  Let $N_t$ and $S_t$ denote the total number of packet arrivals and active slots, respectively, at or before time $t$. Then $S_t \leq N_t/\delta$.  Consequently:
    
    \begin{itemize}[leftmargin=13pt,noitemsep,nosep]
    \item {\bf \em Overall throughput.} Suppose that there are a total of $N\geq 1$ packet arrivals.  Then the total number of active slots is at most $N/\delta$, and hence the overall throughput is at least $\delta$. 

    \item {\bf \em Backlog reprieve.} 
    If $N_t < \delta t$, then there exists an inactive slot $t' \leq t$. Thus, all packets that arrived before $t'$ have completed before~$t'$, and hence the throughput at time $t'$ is at least $\delta$.
    %\jeremy{I don't like that this is stated to only imply one occurrence}
    \end{itemize}
\end{observation}
 
Finally, we observe that there are several natural settings in which implicit-throughput guarantees directly imply strong guarantees on packet backlog, even for infinite input sequences. Suppose, in particular, that packets arrive according to the \defn{adversarial queuing theory} arrival model (described below), which parameterizes the ``burstiness'' of packet-arrival in infinite streams. In adversarial queuing 
theory~\cite{BenderFaHe05,ChlebusKoRo12,ChlebusKoRo06,BorodinKlRaSu01}, the adversary is restricted from injecting too many packets/jammed slots over a set of consecutive slots of length $S$, where $S$ is a parameter of the model that we refer to as the \defn{granularity}. For granularity $S$, the number of packet arrivals plus the number of jammed slots is limited to $\lambda S$, where the \defn{arrival rate} $\lambda$ is a constant less than one. 
On the other hand, how the packet arrivals are distributed within each $S$-sized window is adversarial: 
no restrictions are placed on how the (at most) $\lambda S$ packets and jammed slots are distributed. 
By showing that the implicit throughput of all active slots is $\Omega(1)$, we obtain as a corollary a strong bound of $O(S)$ on the number of packets backlogged in the system at all times, as long as $\lambda$ is a sufficiently small constant.

\subsection{Main results}\label{sec:main-results}

We present an algorithm that with high-probability guarantees $\Omega(1)$ implicit throughput in all active slots and full energy efficiency. 
Thus, we  resolve two open questions in contention resolution: 
First, we show that full-energy efficiency is feasible, even with only ternary feedback.  Second, we show that these guarantees are achievable, even in the presence of adversarial jamming.

Our algorithm relies on the natural multiplicative-weight approach to backoff---with a careful choice of probabilities and update rules.
In contrast to some previous approaches, we do not rely on packet batching (to turn the online problem into a series of batch problems), 
leader election, busy tones, population estimation, dividing the channel into simulated subchannels (e.g., odd and even slots) for packet coordination, 
or other approaches seen in many modern algorithms; alas, these seem problematic in our setting. Given the  simplicity of our algorithm, the technical innovation lies in choosing the update parameters, analyzing the underlying combinatorial process, and proving that it is fast, robust, and fully energy-efficient.

Our algorithm \OurAlg guarantees the main theorems below.\footnote{An event occurs \defn{with high probability (w.h.p.) in $x$} if for any fixed constant $c \geq 1$, the probability of the event is at least $1-x^{-c}$.} 

\begin{theorem}[\bf Implicit throughput, infinite packet streams]\label{thm:infinite_arrival_throughput} At the $t$-th active slot, the implicit throughput is $\Omega(1)$ w.h.p.\ in~$t$.
\end{theorem}

\begin{corollary}[\bf Throughput, finite packet streams]
Consider an input stream of  $N$ packets with $\mathcal{J}$ jammed slots. The throughput for the  execution is $\Theta(1)$ w.h.p.\ in $N + \mathcal{J}$.
\label{cor:finite_arrival_throughput}
\end{corollary}
 
\begin{corollary}[\bf Bounded backlog for adversarial-queuing arrivals]\label{cor:backlog-intro} Consider adversarial-queuing-theory arrivals with a sufficiently small constant arrival-rate $\lambda$ and granularity $S$; i.e.,  in any interval of length $S$, the total number of packet arrivals and jammed slots is at most $\lambda S$. Then, for any given slot $t$, the number of packets currently in the system is at most $O(S)$ w.h.p.\ in $S$.  
\end{corollary}

\begin{restatable}[\bf Energy, finite executions]{theorem}{thmfiniteaccesses}
Consider a finite execution with $N$ total packet arrivals and $\mathcal{J}$ total jammed slots against an adaptive adversary. 
Any given packet accesses the channel $O(\polylog$ $(N + \mathcal{J}))$ times w.h.p.\ in $N+\mathcal{J}$.
\label{thm:adaptive_energy_finite}
\end{restatable}

\begin{restatable}[\bf Energy, adversarial queuing]{theorem}{thminfiniteaccesses}
Consider a (finite or infinite) packet stream with adversarial-queuing arrivals with granularity $S$ and where the arrival rate is a sufficiently small constant. Then
any given packet  accesses the channel $O(\polylog(S))$ times w.h.p.\ in $S$ against an adaptive adversary.
\label{thm:adaptive_energy}
\end{restatable}
 
\begin{theorem}[\bf Energy, infinite executions]
    Consider an infinite packet stream, and let $N_t$ and $\jams_t$ denote the number of arrivals and jammed slots, respectively, up until time~$t$ against an adaptive adversary. Then any given packet accesses the channel $O(\polylog(N_t+\jams_t))$ times before time $t$, w.h.p. in $N_t+\jams_t$.
    \label{thm:infinite_energy}
\end{theorem}

The above statements are simplified slightly for ease of presentation, while the formal statements are presented later in Section~\ref{sec:analysis}. Specifically, 
\thmref{infinite_arrival_throughput} is proved in \secref{bettinggame} where it appears as Corollary \ref{cor:implicit-throughput}, and Corollary~\ref{cor:backlog-intro} corresponds to Corollary~\ref{cor:backlog}.
The remaining theorems are proved in \secref{energy}. In particular, 
\thmref{adaptive_energy_finite} corresponds to \thmref{finite_adaptive_accesses}, and 
\thmref{adaptive_energy} corresponds to  \thmref{queuing_adversarial_accesses}.  Finally, \thmref{infinite_energy} is included in \thmref{infinite_accesses}.  

\subsection{Extensions to Reactive Adversary}

A \defn{reactive adversary}~\cite{richa:jamming3,wilhelm2011short,DBLP:journals/dc/KingPSY18} has an instantaneous reaction time; that is, this adversary listens to the channel and can decide whether to jam and/or inject new packet(s) in slot $t$ based on what it hears in slot $t$ itself.  In contrast, the standard adaptive adversary would not know whether any packet chooses to send in slot $t$ until slot $t+1$. This allows a reactive adversary to cheaply prevent any particular packet $p$ from  succeeding by jamming only those slots where $p$ makes transmission attempts. Thus, against a reactive adversary, the total number of channel accesses required is at least linear in the amount of jamming. For  exponential backoff, the situation is more dire: for any $T$ a reactive adversary can also drive the throughput down to $O(1/T)$ by jamming a single packet a mere $\Theta(\ln T)$ times.  

Note that reactivity and adaptivity are somewhat orthogonal. Reactivity addresses how quickly the adversary can react to the detectable channel state---importantly, only sending is revealed, since it is detectable. In contrast, an adaptive adversary knows all of the internal state and random choices of packets up to the previous slot, and in particular this adversary also knows if packets choose to listen. 

In addition to our main results on purely adaptive adversaries, we also address an adversary that is {\it both} adaptive and reactive. It turns out that the reactive adversary does not impact our implicit-throughput bounds for our algorithm (our analysis applies whether or not the adversary can see the channel activity at the current time). Reactivity thus only impacts the number of channel accesses. Roughly speaking, the theorem states that the reactive adversary has nontrivial impact on the worst-case number of channel accesses and thus energy, which is to be expected, but it does not have significant impact on the average.

\begin{theorem}[\bf Energy, reactive adversary] 
    The following  apply to a reactive and adaptive adversary.
    \begin{enumerate}[leftmargin=15pt,noitemsep]
    \item {\bf Finite streams.}  Consider a finite execution with $N$ total packet arrivals and $\mathcal{J}$ total jammed slots. Any given packet accesses the channel $O( (\mathcal{J}+1) \polylog(N))$ times w.h.p.\ in $N+\mathcal{J}$.
    Moreover, the average number of channel accesses is only $O((\jams/N + 1) \polylog(N + \mathcal{J}))$ times w.h.p.\ in $N+\mathcal{J}$.
    
    \item {\bf Adversarial queuing.}
    Consider a (finite or infinite) packet stream with  adversarial-queuing arrivals with granularity $S$ and where the arrival rate is a sufficiently small constant. Then any given packet accesses the channel at most $O(S)$ times, w.h.p.\ in $S$.  In addition, the average number of accesses per slot is $O(\polylog(S))$, w.h.p.\ in~$S$. 

    \item {\bf Infinite executions.}
    Consider an infinite packet stream, and let $N_t$ and $\jams_t$ denote the number of arrivals and jammed slots, respectively, up until time~$t$. Then any given packet accesses the channel $O( (\jams_t +1 )\polylog(N_t + \jams_t))$ times before time $t$, w.h.p. in $N_t + \jams_t$. 
       Moreover, the average number of channel accesses is $O((\jams_t/N_t + 1)\polylog(N_t+\jams_t))$.
    \end{enumerate}
    \label{thm:reactive_energy_all}
\end{theorem}

The items of \thmref{reactive_energy_all} are each proved separately as Theorems~\ref{thm:finite_reactive_accesses}, \ref{thm:queuing_reactive_accesses}, and~\ref{thm:infinite_accesses} in \secref{energy}.

%%%%%%%%%%%%%%%%%%%%%%%%%%%%%%%%%%%%%%%%%%%%%%%
%%%%%%%%%%%%%%%%%%%%%%%%%%%%%%%%%%%%%%%%%%%%%%%
%%%%%%%%%%%%%%%%%%%%%%%%%%%%%%%%%%%%%%%%%%%%%%%

\section{Related Work}
\label{sec:related}

Contention resolution has been an active area of research for several decades. Here, we discuss related results from the literature. We highlight that a preliminary version of this work has appeared previously \cite{bender2024fully}. 

\paragraph{Robustness.} Early work on jamming resistance by Awerbuch, Richa, and Scheideler \cite{awerbuch:jamming} considers the case of $N$ devices, each of which always has a packet to ready to send; that is, rather than changing over time, the number of packets contending for channel access is always $N$. For any interval of at least $T$ consecutive slots, at most a $(1-\lambda)$-fraction may be jammed, where the constant $\lambda>0$ is unknown. In this setting, Awerbuch, Richa, and Scheideler\cite{awerbuch:jamming} showed that constant throughput is possible with a polylogarithmic number of transmissions per packet. Several subsequent results have built on this result, addressing multi-hop~\cite{richa:jamming2,richa:competitive-j} and co-located networks~\cite{richa:jamming4}, reactive jamming~\cite{richa:jamming3,richa:efficient-j}, and aspects of signal propagation~\cite{ogierman:competitive}. However, this series of results relies on devices knowing rough bounds on $N$ and $T$; in particular, devices must set a parameter whose value is $O(1/(\ln\ln(N) + \ln(T)))$.  

A similar restriction on jamming is explored by  Anantharamu et al.~\cite{anantharamu2019packet,anantharamu2011medium}, where $\rho t+b$ slots in any window of consecutive $t$ slots can be jammed, for $0<\rho \leq 1$, and a non-negative integer, $b$. Probabilistic jamming has also been considered by Chlebus et al.~\cite{Chlebus:2016:SWM:2882263.2882514} in a multi-channel setting, and by Agrawal et al.~\cite{DBLP:conf/spaa/AgrawalBFGY20} where packets have delivery deadlines.
 
When the amount of jamming can be unbounded, Bender et al.~\cite{BenderFiGi19,BenderFiGi16} showed expected constant throughput is possible over non-jammed slots using a $\polylog(N)$ number of sending attempts per packet.  Chang, Jin, and Pettie \cite{ChangJiPe19} give an elegant algorithm that achieves the same asymptotic result, but  drastically improves the expected throughput, rivaling the $1/e$ result achieved by the well-known slotted ALOHA~\cite{binder:aloha}, despite the more-challenging adversarial setting. 

There is a growing body of work~\cite{ChenJiZh21,jiang:robust,DeMarcoSt17,BenderKoKu20} 
on how to design contention resolution protocols where packets listening on the channel during a given slot learn whether or not that slot was successful, but not whether the slot was empty or noisy; that is, packets lack the ability to detect collisions. As mentioned earlier, the results by De Marco and Stachowiak~\cite{DeMarcoSt17},  De Marco, Kowalski, and Stachowiak~\cite{marco:time}, and Bender et al.~\cite{BenderKoKu20} address this setting.  Dealing with jamming in the absence of collision detection is challenging. Here, results by Chen, Jiang, and Zheng \cite{ChenJiZh21} and Jiang and Zheng~\cite{jiang:robust} show how jamming can still be tolerated, although throughput is $\Theta(1/\ln N)$ and packets are allowed to listen on the channel for free.

There is a large body of work on applied security considerations with respect to jamming. The impact of these attacks has been evaluated under a variety of wireless settings  and adversarial strategies~\cite{bayraktaroglu:performance,xu:feasibility,lin:link,aschenbruck:simulative}, and several defenses have been proposed. For example, spread-spectrum techniques involve devices coordinating their communication over different portions of the spectrum in order to evade collisions~\cite{spread:foiling,navda:using,deejam:wood}.  Another defense is to have devices map the areas impacted by jamming and then routing traffic around those areas~\cite{wood2003jam,han2018hierarchical}.  Specialized antenna technology can be used to filter out a jamming signal, which is a process referred to as null steering~\cite{zheng2018design,mcmilin2015gps}. 

\paragraph{Energy efficiency.} As discussed in Section~\ref{sec:intro}, there are a handful of results on contention resolution that achieve constant throughput and address full energy efficiency. To reiterate, Bender et al.~\cite{BenderKoPe18} achieve expected $O(\ln (\ln^*N))$ channel accesses per packet, which Chang et al.~\cite{ChangKoPe19} prove is tight. 
De Marco and Stachowiak~\cite{DeMarcoSt17} and De Marco, Kowalski, and Stachowiak~\cite{marco:time} give sending-efficient algorithms and conjecture that it is possible to extend the results in their papers so that each packet listens to the channel a polylogarithmic number of times (see Section 7 Discussion and open problems in \cite{marco:time}).
%$O(\ln N)$ times, with high probability, even without collision detection. 
However, these results deviate significantly from the ternary model (requiring control messages and, in the case of~\cite{DeMarcoSt17,marco:time}, Good Samaritan packets).

We also note that full energy efficiency has drawn attention outside the domain of contention resolution. Many of these results address wireless networks where, in order to extend the network's operational lifetime,  low-power devices may spend significant time in an energy-efficient sleep state, waking up rarely to access the channel (e.g.,~\cite{wang2019connectivity,deng2005scheduling,li2005energy,ye2006ultra,ullah2009study,merlin2010duty}). Several algorithmic results address energy efficiency in the context of traditional distributed computing problems(e.g., \cite{BarenboimM21,chatterjee2020sleeping,ghaffari2023distributed,dufoulon2023distributed,chang2018energy,king:sleeping,dani2021wake}).

\paragraph{Dynamic packet arrivals.} Much of the initial work on contention resolution addresses arrival times that are dictated by a stochastic process~\cite{MetcalfeBo76,Goldberg-notes-2000,hastad:analysis,goldberg2004bound,shoch1980measured,abramson1977throughput,kelly1987number,aldous:ultimate,GoodmanGrMaMa88,mosely1985class,RaghavanU98,RaghavanU95,GOLDBERG1999232}. The intricacies of whether the contention-resolution algorithm knows local or global history, can adapt its actions over time, can listen on the channel, knows the number of packets in the system, etc. have been explored in depth~\cite{tsybakov1979ergodicity,ferguson1975control,fayolle1977stability,rosenkrantz1983instability,kelly:decentralized,paterson1995contention,DBLP:journals/jacm/GoldbergMPS00,tsybakov1978free-russian,tsybakov1978free,hayes:adaptive,massey1981collision,gallager,tsybakov1980random,verdu1986computation,tsybakov1987upper}. A survey of much of the literature in this setting is provided by Chlebus~\cite{chlebus2001randomized}.

In contrast to stochastic arrivals, packets may arrive at times scheduled by an adversary. The special case where all packets arrive simultaneously has been explored, and $O(n)$ makespan is possible~\cite{Gereb-GrausT92,GreenbergL85,FernandezAntaMoMu13,BenderFaHe05,BenderFiGi06}. More generally, when packets are injected over time in a worst-case fashion,  models from adversarial queuing theory~\cite{BorodinKlRa96,cruz-one,andrews2001universal} have been used to derive stability results~\cite{chlebus2009maximum,chlebus2007stability,anantharamu2010deterministic,aldawsari:adversarial,anantharamu2019packet,anantharamu2011medium,anantharamu2015broadcasting,hradovich2020contention,bienkowski:distributed,garncarek2018local,garncarek2019stable}. Several results explore deterministic approaches for contention resolution (e.g., \cite{deMarco:fast,greenberg1985lower,DBLP:conf/icdcs/MarcoKS19,DEMARCO20171,de2013contention,anantharamu2009adversarial,anantharamu2017adversarial,Kowalski05}, although such approaches often have bounds that are asymptotically inferior to their corresponding randomized solutions.

\paragraph{ Other computation on a shared channel.} A closely-related problem to contention resolution is the wake-up problem, which addresses how long it takes until a  {\it single} successful transmission \cite{jurdzinski2015cost,jurdzinski2005probabilistic,jurdzinski2002probabilistic,jurdzinski2002probabilistic,farach2007initializing,farach2015initializing,Chlebus:2016:SWM:2882263.2882514,ChrobakGK07,ChlebusGKR05,DBLP:journals/dam/MarcoPS07,DBLP:conf/ictcs/MarcoPS05,willard:loglog,fineman:contention,fineman:contention2}; this is in contrast to having all packets succeed. 

More generally, several fundamental distributed computing challenges have been studied involving communication on a multiple access channel, and often even in the presence of adversarial jamming. These include   broadcast~\cite{koo2,gilbert:near,gilbert:making,chen:broadcasting}, leader election~\cite{gilbert:malicious,gilbert2009malicious,DBLP:journals/talg/AwerbuchRSSZ14}, gossip~\cite{gilbert:interference,dolev:gossiping}, node discovery~\cite{meier:speed}, and simple point-to-point communication~\cite{king:conflict,DBLP:journals/dc/KingPSY18}, packet scheduling~\cite{DBLP:conf/waoa/JurdzinskiKL14,DBLP:journals/tcs/AntaGKZ17}.

%%%%%%%%%%%%%%%%%%%%%%%%%%%%%%%%%%%%%%%%%%%%%%%
%%%%%%%%%%%%%%%%%%%%%%%%%%%%%%%%%%%%%%%%%%%%%%%
%%%%%%%%%%%%%%%%%%%%%%%%%%%%%%%%%%%%%%%%%%%%%%%

\section{\OurAlg Algorithm}\label{sec:our-algorithm}

This section presents the \OurAlg algorithm; see Figure~\ref{fig:OurAlg}. For ease of presentation, we describe our algorithm as listening whenever it sends. However, the packet need not actually do both; observe that any packet that is sending  does not need to listen to determine the state of the channel, since if the packet is still in the system after sending in slot $t$, then slot $t$ was noisy.

\begin{figure}[t!]
\begin{tcolorbox}[standard jigsaw, opacityback=0]

\noindent{}\hspace{-3pt}{\bf \OurAlg for packet \boldmath $u$ }\medskip

\noindent{\hspace{-2pt}\it{}Key Variables:} 
\begin{itemize}[leftmargin=19pt,itemsep=1pt]
\item $w_u(t)$: window size of $u$ in slot $t$.\\
 If  $u$ is injected at time slot $t$, then  $w_u(t)=\wmin$.
  \item $c$: a sufficiently large positive constant.
\end{itemize}
\smallskip

\noindent In every slot $t$, packet $u$ executes the following four steps with probability {\small $\displaystyle \frac{c\ln^3 \left(w_u(t)\right)}{w_u(t)}$}:\vspace{-2pt}
\begin{itemize}[leftmargin=19pt,itemsep=0pt,topsep=8pt,after=\vspace{3pt}]
\item \textbf{Listen}
\vspace{10pt}
\item \textbf{Send} with probability 
$\displaystyle \frac{1}{c\ln^3(w_u(t))}$
\vspace{10pt}
\item  If $u$ heard a silent slot, then 
$w_u(t+1) \leftarrow \max{\left\{ \mbox{\small $\displaystyle \frac{w_u(t)}{1+1/(c\ln (w_u(t)))}$}, \, \wmin \right\}}$ 
\vspace{10pt}
\item  If $u$ heard a noisy slot, then $w_u(t+1)\leftarrow  w_u(t)\cdot\left( \mbox{\small $\displaystyle 1+\frac{1}{c\ln(w_u(t))}$} \right)$
\end{itemize}

\end{tcolorbox}
\vspace{-8pt}\caption{\OurAlg algorithm.}  
\label{fig:OurAlg}\vspace{-0pt}
\end{figure}

The probabilities for sending and listening in \OurAlg are determined by a single parameter, which we call \defn{packet $u$'s window size}.
Let $c$ be a sufficiently large positive constant. Let $w_u(t)$ denote packet $u$'s window size at time slot $t$.
When $u$ is injected into the system, its window size is set to the minimum allowed value: {\boldmath $\wmin$}, a sufficiently large positive constant (certainly satisfying $\wmin > 2$ and $\wmin/\ln^3 \wmin \geq c$, the latter inequality ensuring that that $c \ln^3 w_u(t) / w_u(t)$ is at most $1$). 
 The sending and listening rules are as follows. First, packet $u$ listens with probability  $c\ln^3 \left(w_u(t)\right)/w_u(t)$. Then, conditioned on listening, $u$ sends with probability $1/(c\ln^3(w_u(t)))$.

A packet $u$ only has the option to change its window size when it accesses the channel.
Specifically, if at time $t$, packet $u$ listens to the channel and learns that the slot $t$ is busy, then the window size increases (or {\it backs off}) by a \defn{backoff factor} of {\small $\displaystyle 1+1/(c\ln (w_u(t)))$};  that is,
$w_u(t+1) \leftarrow w_u(t)( 1+ 1/(c\ln (w_u(t)))  )$.
Similarly, if at time $t$, packet $u$ accesses the channel and learns that the slot $t$ is empty, then the window size shrinks (or {\it backs on}) by a \defn{backon factor} of 
$1+ 1/c\ln (w_u(t))$, or until it gets back down to $\wmin$, that is,
$w_u(t+1) \leftarrow \max\left\{ w_u(t)/(1+1/(c\ln w_u(t))), \, \wmin \right\}$. 

%%%%%%%%%%%%%%%%%%%%%%%%%%%%%%%%%%%%%%%%%%%%%%%%%%

\section{Technical Overview}

This section gives a technical overview.
In \secref{contention-intro},  we introduce the notion of contention.
In  \secref{potential}, we introduce our potential function $\Phi(t)$. \secref{overview-analysis} gives the main structure of our analysis in terms of intervals. 
Finally, \secref{proof-organization} provides a synopsis of the main analytical results achieved and how they are deployed to make progress towards our main results (\secref{main-results}). 
That is, we describe the main point of each of the technical sections; namely, Sections~\ref{sec:preliminaries}--\ref{sec:energy}.

\subsection{Contention}
\label{sec:contention-intro}
 
For any slot $t$, we define the \defn{contention $\Con$} $= \sum_{u}  1 / w_u(t)$ to be the sum of the sending probabilities in that slot, i.e., the expected number of packets that attempt to send during that slot. 
We say contention is \defn{high} when $\Con > \chigh$, where {\boldmath{$\chigh$}}$>1$ is some fixed positive constant. Conversely, we say that contention is \defn{low} when $\Con < \clow$, where we define {\boldmath{$\clow$}} to be some fixed positive constant such that $\clow \leq 1/\wmin$. Otherwise, if contention is in $[\clow, \chigh]$, then we say that contention is \defn{good}. We also refer to these three windows of contention as \defn{contention regimes}.

%%%%%%%%%%%%%%%%%%%%%%%%%%%%%%%%%%%%%%%%%%%%%%%%%%

\subsection{Our potential function} \label{sec:potential}

Throughout the execution of \OurAlg, we  maintain a potential function $\Phi(t)$ that captures the state of the system at time $t$ and measures the progress toward delivering all packets. When a slot $t$ is inactive, $\Phi(t)=0$. We will see that packet arrivals increase $\Phi(t)$ by $\Theta(1)$ per newly arrived packet, that packets succeeding decrease $\Phi(t)$ by $\Theta(1)$ per packet, that a jammed slot increases $\Phi(t)$ by $O(1)$,  and that on average each slot decreases the potential by $\Theta(1)$, ignoring newly arrived packets. 

For any slot {\boldmath{$t$}}, {\boldmath{$N(t)$}} is the number of packets in the system, 
{\boldmath{$w_u(t)$}} is $u$'s window size, {\boldmath{$\wmax$}} is the largest window size over all packets, and {\boldmath{$\alpha_1$}}, {\boldmath{$\alpha_2$}}, and {\boldmath{$\alpha_3$}} are positive constants. % to be specified later. 
Our potential function consists of three terms. Implicitly, the third term is 0 if there are no packets in the system:
$${\Phi(t)} = \alpha_1 N(t)\, +\, \alpha_2 \sum_u \frac{1}{\ln (w_u(t))}\, +\, \alpha_3 \frac{\wmax}{\ln^2(\wmax)} \,.  $$

\noindent  We abbreviate $\Phi(t)$ as:
$$\Phi(t) = \alpha_1 \first(t) + \alpha_2 \second(t) + \alpha_3 \third(t),$$
 where $\alpha_1$, $\alpha_2$, and $\alpha_3$ may be set so that $\Phi(t)$ will decrease as time progresses for all values of contention $\Con$. The notation {\boldmath{$H(t)$}} and {\boldmath{$L(t)$}} is used to highlight that these terms capture the impact on $\Phi$ from high contention and low contention, respectively. 

\paragraph{Why these terms?} There are three main features of the state of the system that are captured by the potential function: the number of packets, the contention, and the size of the windows.  (Note that these are not independent, as larger windows correspond to lower contention.)  Roughly speaking, when there are many active packets, potential should be high, and when there are no packets, the potential should be 0.  The $\first(t)$ term captures this idea directly by counting the number of packets.  

The $\second(t)$ term is chosen so that the expected change to $\second(t)$ in a slot is proportional to the contention. When the contention is high (and the slot is most likely to have a collision), in expectation $\second(t)$ decreases proportional to the contention (due to the update rule on noisy slots). On the other hand, when the contention is low (and the slot is most likely to have silence), $\second(t)$ increases proportional to the contention. 
Overall, this is good news: when contention is high, $\second(t)$ is likely to decrease by a lot.  When contention is low, there is a small expected increase in $\second(t)$, but that increase is counterbalanced by the (small) reduction in $\first(t)$ in expectation due to packet successes.  Choosing $\alpha_1 > \alpha_2$ makes the net effect a decrease. 

Finally, the $\third(t)$ term allows us to cope with 
the situation that the contention is low but  some packets in the system have large windows (e.g., there is a single packet with a very large window).  As it is likely to take a long time for the packet to succeed, the potential should be high.  $\third(t)$ is roughly the expected time for a packet with window $\wmax$ to decrease its window size to a constant if all slots are silent. 
The analysis then needs to show that any increases to $\third(t)$ are counterbalanced by decreases in the other terms, ensured by $\alpha_1 > \alpha_2 > \alpha_3$. 

%%%%%%%%%%%%%%%%%%%%%%%%%%%%%%%%%%%%%%%%%%%%%%%%%%

\paragraph{Challenge with $\third(t)$.}
The $\first(t)$ and $\second(t)$ terms are well-behaved in the sense that they change on a {\it per-slot basis}, while the $\third(t)$ term cannot. To see why, consider the case that several packets with window size $\wmax$ remain.   The $\third(t)$ term only decreases after \emph{all} of those packets have chosen to listen and observed silence.  On any slot that multiple such packets remain, it is extremely unlikely that all of the packets choose to listen.  Thus, $\third(t)$ does not decrease by a constant in expectation.  Instead, we need a coarser granularity to understand the behavior of $\third(t)$.
 
%%%%%%%%%%%%%%%%%%%%%%%%%%%%%%%%%%%%%%%%%%%%%%%%%%

\subsection{Analyzing intervals}\label{sec:overview-analysis}

Our analysis divides the execution into disjoint \defn{intervals} of time. The first interval starts at the first slot with an active packet. An interval starting at time $t$ has size $\tau = (1/\cint) \max\{\third(t), \sqrt{\first(t)}\}$, where $\cint$ is a constant.  If any active packets remain, the next interval starts immediately after the previous interval ends.  (Otherwise, an interval begins the next time there is an active packet.)

A key technical theorem is the following.  Let $\mathcal{A}$ and $\mathcal{J}$ denote the number of arrivals and jammed slots, respectively, in the size-$\tau$ interval.  For $\mathcal{A}=\mathcal{J}=0$, the lemma states that the potential decreases by $\Omega(\tau)$ across the interval, with high probability in $\tau$, meaning a decrease of $\Omega(1)$ per slot. For general $\mathcal{A},\mathcal{J}$, the potential decreases by $\Omega(\tau) - O(\mathcal{A}+\mathcal{J})$. 

\begin{restatable*}[\textbf{\boldmath Decrease in $\Phi(t)$ over  interval $\mathcal{I}$ w.h.p.\ in $|\mathcal{I}|$}]{theorem}{phimain}\thmlabel{phi-main}
Consider an interval $\mathcal{I}$ starting at $t$, and ending at $t'$, of length $|\mathcal{I}|= \tau = (1/\cint) \cdot \max\big\{\frac{\wmax}{\ln^2( \wmax)}, \first(t)^{1/2}\big\}$. Let $\mathcal{A}$ and $\mathcal{J}$ be the number of packet arrivals and jammed slots in $\mathcal{I}$. With high probability in $\tau$,  $\Phi$ decreases over $\mathcal{I}$ by at least $\Omega(\tau) - O(\mathcal{A} + \mathcal{J})$.  That is,
\[ 
\Pr\left[(\Phi(t') - \Phi(t)) 
\,\geq\,  \Theta(\mathcal{A}+\mathcal{J}) - \Omega(\tau)\right] \,\leq 
\, (1/\tau)^{\Theta(1)} \ .
\]
\end{restatable*}

Our proof of \thmref{phi-main} is broken into several lemmas according to the level of contention.   Specifically, we have separate cases for high contention, good contention, and low contention. In each of the cases, absent arrivals and jamming, we argue that there is a net decrease in potential, with high probability, but the contributing term is different in each case. The interplay between $\first(t)$ and $\second(t)$ is tight enough that we analyze the net effect on the sum of these terms together, but we analyze $\third(t)$ separately. A more detailed summary is provided next in \secref{proof-organization}. 

A significant complication is that (1) the probability stated in \thmref{phi-main} depends on the size of the interval, and (2) the interval sizes are determined adaptively by actions of the adversary.  To analyze the full process, we model an execution as a specific biased random walk that we set up as a betting game (\secref{bettinggame}).  The bounds provided by the betting game translate into high-probability bounds with respect to the total number of packets. 

Throughout the paper, standard Chernoff bounds sometimes cannot be used for two reasons. First the adversary can adaptively influence the length of an interval. Moreover within each interval, the adversary can influence which slots are high, low, and good contention. To be able to analyze these slots separately, we must instead apply a generalization of Azuma's inequality (see \thmreftwo{adversarial_azumas_upper}{adversarial_azumas_lower} in Section~\ref{sec:preliminaries}, taken
from~\cite{KuszmaulQi21})  
that gives us Chernoff-like bounds but with adaptively chosen probability distributions. 

Finally, good upper bounds for $\Phi(t)$ enable us to characterize the (implicit and standard) throughput and energy consumption of \OurAlg in all its variety of settings (finite versus infinite executions, arbitrary infinite versus  infinite with adversarial-queuing arrivals, adaptive adversaries that are reactive versus non-reactive).
The most direct application of $\Phi(t)$ is to bound implicit throughput. 
$\Phi(t)$ also gives us an upper bound on the maximum window size  $\wmax$, specifically, $\wmax=O\big(\Phi(t)\ln^2(\Phi(t))\big)$, which we use to prove energy bounds in \secref{energy}. 

%%%%%%%%%%%%%%%%%%%%%%%%%%%%%%%%%%%%%%%%%%%%%%%%%%

\subsection{Proof organization}
\label{sec:proof-organization}

The main analysis in this paper, including all of the proofs, appears in \secref{analysis}. This section summarizes the proof structure, highlighting the key lemma statements.

%%%%%%%%%%%%%%%%%%%%%%%%%%%%%%%%%%%%%%%%%%%%%%%%%%

\paragraph{Overview of \secref{preliminaries}: Preliminaries.} This section lists several well-known  inequalities that are used throughout our analysis. We review  bounds on the probability that a slot is noisy, empty, or contains a successful transmission as a function of contention (\lemrefthree{p-suc}{p-emp}{p-noi}). The lemmas in this section allow us to immediately obtain constant bounds on the probabilities of empty slots, successful slots, and noisy slots in different contention regimes.  %\mab{What are the lemma names to cite here?}

Theorems~\ref{thm:adversarial_azumas_upper} and~\ref{thm:adversarial_azumas_lower} give upper and lower bounds for the sum of random variables, where the distribution of each subsequent random variable is determined by an adaptive adversary. This adversarial, multiplicative version of Azuma's inequality is a powerful tool from~\cite{KuszmaulQi21} that allows us to analyze the performance of our algorithm in situations where a simpler Chernoff-bound-style argument does not appear to work, given the adaptive nature of our adversary.
 
%%%%%%%%%%%%%%%%%%%%%%%%%%%%%%%%%%%%%%%%%%%%%%%%%%

\paragraph{Overview of \secref{second-analysis}:~$\first(t) + \second(t)$ over single slots and intervals,  when contention 
is low, high, and good.} This section addresses the behavior of $\first(t)$ and $\second(t)$. When contention is high, we expect to see a decrease in $\second(t)$, which should be large enough that its reduction outweighs any increase from $\third(t)$, and thus $\Phi(t)$ decreases. 

\lemref{B_term_increase_inverse_log_cubed} 
shows how much $\second(t)$ changes as a result of a \emph{specific} packet listening during a slot~$t$---that is, how much   $\second(t)$ increases when the slot is  silent and decreases when the slot is noisy.
\lemref{AB_low_mid} 
analyzes the change to $\first(t) + \second(t)$  due to the low and good contention slots in an arbitrary interval.
We highlight that the adaptive adversary exerts some control over which slots have low and good contention, since it can inject packets and/or jam in slot $t+1$ based on the packets' random choices in slot~$t$.
In particular, let $|\mathcal{G}|$ denote the number of good slots in the interval, then  \lemref{AB_low_mid} shows the following. 
Over \emph{good-contention} slots, $\first(t)+\second(t)$ decreases by $\Omega(|\mathcal{G}|)$, minus the number of packet injections, jammed slots, and a polylog term in the length of the interval, \whp in the interval length. 
\lemref{AB_low_mid} also  shows that over the \emph{low-contention slots}, $\first(t)+\second(t)$ increases by at most the number of packet injections, jammed slots, and a polylog term in the length of the interval, again \whp in the interval length.  

\begin{restatable*}[\bf Increase/decrease in {\boldmath{$\second(t)$}} due to a silent/noisy slot]{myLemma}{Bterminc}\lemlabel{B_term_increase_inverse_log_cubed}
When packet $u$ listens to a silent slot $t$, $\second(t)$ increases by $\Theta\big(\frac{1}{c\ln^3 w_u}\big)$ due to packet $u$. 
When a packet $u$ listens to a noisy slot, $\second(t)$ decreases by $\Theta\big(\frac{1}{c\ln^3w_u}\big)$ due to packet $u$.  
\end{restatable*}

\begin{restatable*}[\bf Net delta in contribution of {\boldmath $\first(t)$} and {\boldmath $\second(t)$} to potential over low and good contention slots]{myLemma}{ABlowmid}\lemlabel{AB_low_mid}
Let $\mathcal{I}$ be an arbitrary interval 
starting at time $t$, and ending at $t'$, with length $|\mathcal{I}|=\tau$. 
Let $\mathcal{L}$ be the set of slots in $\mathcal{I}$ during which $\Con < \clow$. 
Let $\mathcal{G}$ be the set of slots in $\mathcal{I}$ during which $\Con \geq \clow$ and $\Con \leq \chigh$. 
Let $\mathcal{A}_\mathcal{L}$ be the number of packet arrivals in time slots in $\mathcal{L}$.
Let $\mathcal{A}_\mathcal{G}$ be the number of packet arrivals in time slots in $\mathcal{G}$.
Let $\mathcal{J}_\mathcal{L}$ be the number of jammed slots in $\mathcal{L}$.
Let $\mathcal{J}_\mathcal{G}$ be the number of jammed slots in $\mathcal{G}$.
Define:
\begin{itemize}
    \item the $\defn{net delta}$ over $\mathcal{L}$ to be the sum of the changes in $\first(t)$ and $\second(t)$ during the slots in $\mathcal{L}$, i.e., \[\sum_{t' \in \mathcal{L}}\Big(\alpha_1\big(\first(t'+1) - \first(t')\big) + \alpha_2\big(\second(t'+1) - \second(t')\big)\Big).\]
\item the $\defn{net delta}$ over $\mathcal{G}$ to be the sum of the changes in $\first(t)$ and $\second(t)$ during the slots in $\mathcal{G}$, i.e., 
\[\sum_{t' \in \mathcal{G}}\Big(\alpha_1\big(\first(t'+1) - \first(t')\big) + \alpha_2\big(\second(t'+1) - \second(t')\big)\Big).\]
\end{itemize}

\noindent Then, for proper choices of $\alpha_1$ and $\alpha_2$:
\begin{itemize}[noitemsep]
\item The net delta over $\mathcal{L}$ is at most $O(\ln^2 \tau) + \alpha_1( \mathcal{A}_\mathcal{L} + \mathcal{J}_\mathcal{L})$
w.h.p.\ in~$\tau$.
\item The net delta over $\mathcal{G}$ is at most $O(\ln^2 \tau) + \alpha_1 (\mathcal{A}_\mathcal{G} + \mathcal{J}_\mathcal{G}) - \Omega(|\mathcal{G}|)$
w.h.p.\ in~$\tau$.
\end{itemize}
\end{restatable*}

\lemref{AB_high} provides a symmetric high-probability bound on $\first(t) + \second(t)$ over the \emph{high-contention} slots in an arbitrary interval. For the high-probability bound, in this case, we have that $\first(t) + \second(t)$ will decrease by $\Omega(|\mathcal{H}|)$, where $|\mathcal{H}|$ is the number of high-contention slots, up to the usual additional terms of jamming, packet injections, and a polylog term in terms of the interval length.

% this lemma is showing up wrongly numbered
\begin{restatable*}[\bf Net delta in contribution of {\boldmath $\first(t)$} and {\boldmath $\second(t)$} to potential over high  contention slots]{myLemma}{ABhigh}
\label{lem:AB_high}
Let $\mathcal{I}$ be an arbitrary interval starting at time $t$, and ending at $t'$, with length $|\mathcal{I}|=\tau$.
Let $\mathcal{H}$
be the set of slots in $\mathcal{I}$ during which $\Con > \chigh$. 
Let $\mathcal{A}_{\mathcal{H}}$ be the number of packet arrivals in time slots in $\mathcal{H}$. %, and let $\mathcal{J}_{\mathcal{H}}$ be the number of jammed slots in $\mathcal{H}$,

Define the $\defn{net delta}$ over $\mathcal{H}$ to be the sum of the changes in $\first$ and $\second$ during the slots in $\mathcal{H}$, i.e., \[\sum_{t' \in \mathcal{H}}\Big(\alpha_1\big(\first(t'+1) - \first(t')\big) + \alpha_2\big(\second(t'+1) - \second(t')\big)\Big).\]

\noindent Then the net delta over $\mathcal{H}$ is at most $O(\ln ^3 \tau ) + \alpha_1 \mathcal{A}_\mathcal{H} - \Omega(|\mathcal{H}|)$ w.h.p.\ in~$\tau$.  
\end{restatable*}

It is worth noting that the proofs of 
\lemreftwo{AB_low_mid}{AB_high} are  technically involved. One of the reasons for this is that these lemmas contain our main applications of Theorems \ref{thm:adversarial_azumas_upper} and \ref{thm:adversarial_azumas_lower}. 
This is necessary because the potential-function terms behave very differently in the three contention regimes---and because  the  adaptive adversary has the ability to change the contention in a slot on the fly.

\lemref{AB_not_high_and_decrease} then collates Lemmas~\ref{lem:AB_low_mid} and~\ref{lem:AB_high} to show that over an arbitrary interval of length $\tau$, it is either the case that almost all of the slots are low contention slots, or the first two terms decrease by $\Omega(\tau)$ with high probability (again, up to terms for packet insertions and jamming). 
\lemref{AB_not_high_and_decrease}  considers  all slots, rather than only those of a particular contention regime. This lemma is the only one from this subsection that will be used later in the analysis, but the earlier lemmas in the subsection are necessary to build up to it.

\begin{restatable*}[\textbf{Unless most slots have low contention,  \boldmath $\alpha_1\first(t)+\alpha_2\second(t)$ decreases}]{myLemma}{ABnothigh}
Let $\mathcal{I}$ be an arbitrary interval of length $\tau > \Omega(1)$ with $\mathcal{A}$ packet arrivals and $\mathcal{J}$ jammed slots. With high probability in $\tau$, at least one of the following two conditions holds:

\begin{itemize}[noitemsep]
    \item Less than $1/10$ of slots satisfy $\Con \geq \clow$.
    \item $\alpha_1\first(t) + \alpha_2\second(t)$ decreases by $\Omega(\tau) - O(\mathcal{A} + \mathcal{J})$ over $\mathcal{I}$.
\end{itemize}

\noindent Additionally, $\alpha_1\first(t) + \alpha_2\second(t)$ increases by at most $O(\ln^3 \tau + \mathcal{A} + \mathcal{J})$ w.h.p.\ in $\tau$.

\label{lem:AB_not_high_and_decrease}
\end{restatable*}

%%%%%%%%%%%%%%%%%%%%%%%%%%%%%%%%%%%%%%%%%%%%%%%%%%

\paragraph{Overview of \secref{third-analysis}: Amortized behavior of \boldmath $\third(t)$.}  This section analyzes $\third(t)$'s behavior over intervals of length   
$|\mathcal{I}|= \tau = (1/\cint) \cdot \max\{\frac{\wmax}{\ln^2( \wmax)}, \first(t)^{1/2}\}$.
The two main things that we want to show are that $\third(t)$ does not increase by much, regardless of the contention regime, and that when there are many low-contention slots, $\third(t)$ exhibits a substantial decrease. 

\lemref{window-change} argues that a packet with large-enough window size is unlikely to have its window change by much during the interval. This lemma is instrumental when considering packets across an interval (notably in the proof of \lemref{C-term}) as it means that their probability of listening also does not change by much. 

\lemref{tail-bound-third} provides one of the main results of the section: a tail bound, and hence also a high probability bound, on how much $\third$ increases over the interval regardless of contention.  The proofs for both \lemreftwo{window-change}{tail-bound-third} amount to arguing that an individual packet is unlikely to listen to the channel too many times, which means that its window size also cannot change by very much.  Because we are pessimistically counting the number of listens, the actual state of the channel does not appear in the proofs, and thus the number of jammed slots is irrelevant.  

\begin{restatable*}[\textbf{Bounds on the factor that a large window can grow/shrink}]{myLemma}{Wchange}\label{lem:window-change}
    Consider any packet during an interval $\mathcal{I}$ with $\tau = |\mathcal{I}|$.  
    Let $Z$ satisfy $Z / \ln^2(Z) = \tau$.  And let $W\geq \wmin$ be the initial size of the packet's window. Let $W^-$ be the smallest window size the packet has while still active in the interval, and let $W^+$ be the biggest window size the packet achieves during the interval. Then for large enough choice of constants $\wmin$ and $c$ and any constant parameter $\gamma>0$ and $k\geq 2$:\\
    If $W = \Theta(kZ)$, then 
    $$\Pr\left[W^+ \geq e^{\gamma} W\,\mbox{ or }\, W^- < W/e^{\gamma}\right] \leq 1/\tau^{\Theta(c \gamma \ln (\gamma k))}
    \,.
    $$
    If $W = O(Z)$, then
    \[ \Pr[W^+ \geq \Omega(e^{\gamma} Z)] \leq 1/\gamma^{\Theta(c \gamma \ln(\gamma ))}.\]
    \john{Add to proof last sentence.}
\end{restatable*}

\begin{restatable*}[\textbf{Tail bound on increase in \boldmath $\third(t)$}]{myLemma}{Ltail}
\label{lem:tail-bound-third}
Consider an interval $\mathcal{I}$ with length $\tau = |\mathcal{I}|$ starting from time~$t$ and ending at time $t' = t+\tau$, where $\tau = (1/\cint)\max\left\{L(t), \sqrt{\first(t)}\right\}$.  Let $\mathcal{A}$ be the number of arrivals during the interval.
Then for large-enough constant $c$ in the algorithm and any $k \geq 2$:
\[
    \Pr\left[\third(t') \geq \Theta(\mathcal{A} + k\tau)\right] \leq 2^{-\Theta(c(\ln \tau \cdot \ln k + \ln^2 k))}
    \,.
\]
\end{restatable*}

\lemref{C-term} is the other main result of the section.  This lemma says that as long as most slots have low contention, then $\third$ decreases by $\Omega(\tau)$, minus the number of packet arrivals and jammed slots.   The proof focuses on packets with large windows, i.e., window closes to $\wmax$. The main idea of the proof is to give a high-probability lower bound on the number of times each such packet listens and hears silence as well as an upper bound on how many times the packet listens and hears noise. As long as the former is larger by a constant factor, the packet is likely to decreases its window size by a constant factor.  Taking a union bound across packets is enough to conclude that all packets with large windows have their window sizes decrease, with high probability. 

\begin{restatable*}[\textbf{Mostly low contention implies decrease in
\boldmath $\third(t)$}]{myLemma}{Cterm}
\label{lem:C-term}
Consider an interval $\mathcal{I}$ starting at $t$ of length $\tau$, where $\tau = (1/\cint)$ $\max\left\{L(t), \sqrt{\first(t)}\right\}$. 
Let $t_1 = t + \tau$, and let $\mathcal{A}$ and $\mathcal{J}$ denote the number of packet arrivals and jammed slots, respectively, over $\mathcal{I}$. Then, with high probability
in $\tau$, either 
\begin{itemize}
    \item $\third(t_1) \leq \third(t) / d + O(\mathcal{A})$,
    % \mab{I think we do not need to include jammed slots in the additive term}
    where $d>1$ is a constant, 
        or
    \item At least a $1/10$-fraction of the slots $t'$ in the interval $\mathcal{I}$ are either jammed or have contention $C(t') \geq \clow$. 
    \end{itemize}  
    Incorporating the fact that $\tau \geq \third(t)/\cint$, 
   it follows that if at least a $9/10$ fraction of slots in the interval have contention at most $\clow$, then $(L(t_1)-L(t)) \leq O(\mathcal{A}+\mathcal{J})-\Omega(\tau)$.  
\end{restatable*}

%%%%%%%%%%%%%%%%%%%%%%%%%%%%%%%%%%%%%%%%%%%%%%%%%%

\paragraph{Overview of \secref{combining-analysis}: Combining the analyses of  $\first(t)$, $\second(t)$, and $\third(t)$, to analyze $\Phi(t)$.}
This section combines all three terms of the potential function to characterize the overall behavior of $\Phi(t)$.  The key tools established in this section are 
\thmref{phi-main} (stated previously in Section~\ref{sec:overview-analysis}) and \thmref{tail-bound}, which allow us to argue that the potential will decrease (most of the time) and that, when this fails to occur, the amount by which it increases is bounded. Specifically, consider a size-$\tau$ interval with $\mathcal{A}$ packet arrivals and $\mathcal{J}$ jammed slots.  \thmref{phi-main} shows that $\Phi(t)$  decreases by $\Omega(\tau) - O(\mathcal{A} + \mathcal{J})$ w.h.p.\ in $\tau$.  \thmref{tail-bound} establishes tail bounds, proving that even when the high-probability bound of \thmref{phi-main} fails, the probability that $\Phi(t)$ increases by
more than $k\tau^2 + O(\mathcal{A} + \mathcal{J})$ is  less than $ \frac{1}{\poly(\tau)}\cdot {(1/2)}^{\Theta(\ln^2k)}$. 

\thmreftwo{phi-main}{tail-bound} are the tools needed to fit our betting game, discussed next and in \secref{bettinggame}, and thereby argue that the potential is likely to decrease sufficiently across multiple intervals.

\begin{restatable*}[\textbf{\boldmath Tail bound on increase in $\Phi(t)$ over interval $\mathcal{I}$}]{theorem}{tailbound}
% [\textbf{Tail bound on increase in \boldmath $\Phi(t)$ over  interval ]
 \thmlabel{tail-bound}
Consider an interval $\mathcal{I}$ of length $|\mathcal{I}|=\tau$ starting at time $t$ and ending at time $t'$, where $\tau = (1/\cint)\max\left\{L(t), \sqrt{\first(t)}\right\}$. Let $\mathcal{A}$ be the number of packet arrivals in $\mathcal{I}$. Then 
the probability that $\Phi$ increases by at least $\Theta(\mathcal{A}) + \Theta(k\tau^2)$ 
is at most $2^{-\Theta(c(\ln \tau \cdot \ln k + \ln^2 k))} \leq  (1/\tau^{\Theta(c)})\cdot 2^{-\Theta(\ln^2 k)}$, where $c$ is the constant parameter of the algorithm, and $k\geq 2$ is a constant.  That is, 
\[ \Pr\left[(\Phi(t')-\Phi(t)) \geq  \Theta(\mathcal{A}) + \Theta(k\tau^2)\right] \leq \left(\frac{1}{\tau^{\Theta(1)}}\right) \cdot 2^{-\Theta(\ln^2 k)} \ .
\]
\end{restatable*}

%%%%%%%%%%%%%%%%%%%%%%%%%%%%%%%%%%%%%%%%%%%%%%%%%%
\paragraph{Overview of \secref{bettinggame}: Using $\Phi(t)$ to prove throughput via a betting-game argument.} The analysis so far establishes progress guarantees over sufficiently large intervals in the form of \thmreftwo{phi-main}{tail-bound}. 
Here, we show how to apply these theorems to give upper bounds on the potential over the execution with high probability in the total number of packets and jammed slots. 

Since the adversary is adaptive, we have to be careful in
combining bounds across intervals. 
The adversary can use the results of earlier intervals in choosing new arrivals and jamming, which affects the size of later intervals.  
To reason about this process, we reframe it in a setting that resembles a random walk, which we describe below in a \defn{betting game}. Our analysis of this game then allows us to analyze the implicit throughput (recall Section~\ref{sec:model}).

\paragraph{The Betting Game.} We first summarize the betting game and then later relate it to the backoff process. 
The adversary corresponds to a \defn{bettor} who makes a series of bets.  Each bet has a size equal to the duration $\tau$. The bettor also has some amount of money, which is initially $0$ dollars. When the bettor loses a bet, the bettor loses some money, and when the bettor wins, the bettor wins some money. (The amounts won or lost are specified below as a function of the size of the bet.)  Additionally, at any time, the bettor may choose to receive a \defn{passive income}.  The passive income is added to the bettor's wealth. The total amount of passive income taken, however, means that the bettor must play the game longer.  The game begins when the bettor first takes some passive income, and the game does not end until either the bettor goes broke or the bettor has resolved bets totaling $\Omega(\passive)$ size, where $\passive$ is the passive income received, whichever comes first.  The bettor's goal is to complete the game without going broke. Importantly, although the bettor can always choose to take more passive income, doing so increases the total play time. 

We set the details of the betting game to mirror the backoff process. Each bet corresponds to an interval. Passive income during a bet corresponds to the number of arrivals and jammed slots during the interval.  Money corresponds to potential. 

The bettor loses a size-$\tau$ bet with probability at least $1-\frac{1}{\poly(\tau)}$.
If the bettor  loses the size-$\tau$ bet, it loses $\Theta(\tau)$ dollars.
This loss corresponds to the high-probability event (in $\tau$) of \thmref{phi-main}. 
The bettor wins a size-$\tau$ bet  with probability $O(1/\poly(\tau))$.
If the bettor wins the bet, it gets $\Theta(\tau^2)$ dollars, plus $Y$ \defn{bonus dollars}, where $Y$ is a random variable such that $\Pr[Y \geq k\tau^2] \leq  \frac{1}{\poly(\tau)}\cdot 2^{-\Theta(\ln^2k)}$; these winnings correspond to tail bound of \thmref{tail-bound}.
(Of course, during each bet, the bettor can  also gain passive income for arrivals and jammed slots.)

We allow the bettor the power to choose arbitrary bet sizes 
(subject to a minimum  interval 
size, which itself is determined by $\wmin$), and the bettor is even allowed to place bets whose loss would cause the bettor to end with negative money.  (In the actual backoff process, the interval sizes are dictated by the current state of the system, and not entirely under the control of the adversary.) 

The rules of betting game are set in favor of the bettor, such that when the bettor wins, $\Phi(t)$ increases more slowly than the bettor's wealth increases, and when the bettor loses, $\Phi(t)$ decreases at least as fast as the bettor's wealth decreases. Therefore, this betting game stochastically dominates the potential function. 

\emph{The takeaway is that at any point $t$, the bettor's wealth is an upper bound on $\Phi(t)$.} 
Because $\Phi(t)$ is an upper bound on the number of packets in the system, the bettor going broke corresponds to all packets succeeding. We thus obtain good implicit throughput, because there must either be many jammed slots or packet arrivals, or there must be many packets succeeding, leading to inactive slots.  

\paragraph{Upper bounding the bettor's maximum wealth/potential and showing $\Omega(1)$ implicit throughput.} In \lemref{passive}, we provide a high-probability upper bound on the bettor's maximum wealth and the amount of time until it goes broke, which corresponds to there being no packets in the system. \medskip

%\begin{restatable}({\bf The bettor loses the betting game})
%{lemma}{BGpassive} 
%Suppose the bettor receives $\passive$ dollars of passive income.   Then with high probability in $\passive$, the bettor never has more than $O(\passive)$ dollars across the execution.  Moreover, the bettor goes broke within $O(\passive)$ active slots, with high probability in $\passive$. 
%\label{lem:passive}
%\end{restatable}

{\textsc{Lemma}}~\ref{lem:passive} {\it ({\bf The bettor loses the betting game}). Suppose the bettor receives $\passive$ dollars of passive income.   Then with high probability in $\passive$, the bettor never has more than $O(\passive)$ dollars across the execution.  Moreover, the bettor goes broke within $O(\passive)$ active slots, with high probability in $\passive$. }
\medskip

We briefly explain here how \lemref{passive} implies implicit throughput.  Consider a time horizon $t$, and suppose that the bettor has received $\passive = t/c$ dollars from passive income, for constant $c$ matching the \mbox{big-$O$} of the lemma. Then from \lemref{passive}, with high probability in $t/c$, the bettor goes broke within $c\cdot t/c = t$ time; that is, there are no active packets at time $t$.  We thus obtain the $\Omega(1)$ throughput result of \thmref{infinite_arrival_throughput}. 

%%%%%%%%%%%%%%%%%%%%%%%%%%%%%%%%%%%%%%%%%%%%%%%%%%

\paragraph{Overview of \secref{energy}: Channel access/energy bounds.}  
In this section, we establish energy bounds. Two of the theorem statements, namely \thmreftwo{finite_adaptive_accesses}{infinite_accesses}; the rest appear in \secref{energy}. 
Theorems~\ref{thm:finite_adaptive_accesses}--\ref{thm:infinite_accesses}
are proved via properties of our potential function. (Several additional useful lemmas about the potential, not highlighted above, do appear in \secref{bettinggame}).
So far, we have primarily motivated
$\Phi(t)$ as a tool for proving throughput bounds, but $\Phi(t)$  also enables channel-access bounds.

\thmref{finite_adaptive_accesses} gives energy bounds in the finite case against an adaptive adversary. Specifically, if the stream has $N$ packets and $\jams$ jammed slots, then w.h.p.\
each packet accesses the channel at most $\polylog(N+\jams)$ times.
The proof structure is as follows:
Our upper bound on $\Phi(t)$ immediately gives an upper bound on a packet's maximum window size: $\wmax = O(\poly(\Phi(t)) = O(\poly(N+\jams))$. 
Thus, if a packet accesses the channel too many times, then many of these accesses must have been listening during silent slots, so that the packet window can get smaller.
However, by the structure of \OurAlg, whenever a packet first chooses to listen, there is at least a $1/\polylog(N+\jams)$ probability that it also sends. 
Thus, after  $\polylog(N+\jams)$ channel accesses when all other packets are silent, with high probability that packet has succeeded. %has been transmitted. 

\begin{restatable*}[{\bf Energy bound for finite case against adaptive adversary}]{theorem}{finiteenergy}\label{thm:finite_adaptive_accesses}
Consider an input stream with $N$ packets and $\jams$ jammed slots. Assume that the adversary is adaptive but not reactive. Then w.h.p.\ in $N+\jams$, every packet accesses the channel at most $O(\ln^4(N+\jams))$ times.
\end{restatable*}

The corresponding proof illustrates one subtle design choice of \OurAlg,  which leads to an easier energy analysis. Specifically, a given packet's sending and listening probabilities are correlated: if a packet sends, then it has already decided to listen (but, of course, a packet can listen without deciding to send). 
We conclude by observing that, with an adaptive adversary, all packets have good channel-access bounds.

\thmref{finite_reactive_accesses} gives an analogous result for an adversary that is both adaptive and reactive.
By the very nature of a reactive adversary, there is no possibility of good per-packet bounds on channel accesses. 
(For example, a reactive adversary could target a specific packet
and reactively jam whenever it sees this packet try to send.)
However, interestingly, the amortized channel-access bounds are still good. This is because the reactive adversary only learns about sending on the channel and can react instantaneously; it does not learn whether a packet is listening in the current slot. 
Thus, a targeted packet can still reduce its window (as the other packets do) and it will succeed unless the adversary does significant jamming. 
For example, consider the special case where the targeted packet is the only packet remaining. Then, unless the adversary (which does not sense when a packet will listen) jams a large number of slots, this packet will correctly back on and succeed.

\thmreftwo{queuing_adversarial_accesses}{queuing_reactive_accesses}  generalize 
\thmreftwo{finite_adaptive_accesses}{finite_reactive_accesses} 
to the  adversarial-queuing setting with granularity $S$ and sufficiently small arrival rate $\lambda$.
The main tool is \lemref{queuing_nearby_inactive}, which allows us to transform the adversarial queuing case into finite instances that are not very large. 
\thmref{infinite_accesses} applies to infinite streams with arbitrary arrivals.

\begin{restatable*}[{\bf Channel access bounds for infinite case against adaptive and reactive adversaries}]{theorem}{infiniteenergy}\label{thm:infinite_accesses}
    Suppose that up until time $t$ there have been $N_t$ packet arrivals and $\jams_t$ jammed slots.  
    \begin{itemize}[noitemsep, leftmargin=10pt]
        \item  Consider an adaptive adversary that is not reactive. Then w.h.p.\ in $\jams_t+N_t$, each packet  makes $O(\ln^4(\jams_t+N_t))$ channel accesses before time $t$. 
        \item Consider and adaptive adversary that is reactive.  Then w.h.p.\ in $\jams_t+N_t$, a particular packet accesses the channel at most $O( (\jams_t+1)\ln^3(N_t+\jams_t) + \ln^4(N_t+\jams_t))$ times. Moreover, the average number of channel accesses is $O((\jams_t/N_t + 1)\ln^4(N_t+\jams_t))$.
    \end{itemize}   
\end{restatable*}

%%%%%%%%%%%%%%%%%%%%%%%%%%%%%%%%%%%%%%%%%%%

\section{Analysis}\label{sec:analysis}

Our analysis is presented in the same order as discussed in our overview of the proof organization in Section~\ref{sec:proof-organization}.

\subsection{Preliminaries}\label{sec:preliminaries}

We make use of the following inequalities, which appear throughout the previous literature.

\begin{fact}\label{fact:taylor}
The following inequalities hold. 
\begin{enumerate}[label={(\alph*)},leftmargin=20pt]
%\item For any $0\leq x<1$, $1 - x \geq e^{-x/(1-x)}$.\label{fact-a}
\item For any $0\leq x<1$, $1 - x \geq e^{-x/(1-x)}$.\label{fact-a}
\item For any $0\leq x<1$, $\ln(1+x) = x - \sum_{j=2}^{\infty} (-x)^{j}/j$.\label{fact-b} 
\item For any $0\leq x<1/2$,  $1+2x \geq e^{x}$.\label{fact-c} 
\item For any $x$, $1 - x \leq e^{-x}$.\label{fact-d}
\end{enumerate}
\end{fact}

For any fixed slot $t$, let {\boldmath $\psuccess(t)$} denote the probability that some message is successful in slot $t$. We state the following result relating $\Con$ and  $\psuccess(t)$; this has appeared in similar forms in~\cite{richa:jamming4,BenderFiGi19,DBLP:conf/spaa/AgrawalBFGY20,BenderFiGi16,ChangJiPe19}.

\begin{restatable}[\textbf{Probability of success as a function of contention}]{myLemma}{sucProb}\label{lem:p-suc}
For any unjammed slot $t$, where $w_u(t) \geq 2$ for all packets, the following holds:
$$ \frac{ \Con}{e^{2\Con}} \leq \psuccess(t) \leq \frac{2\Con}{e^{\Con}}.$$
\end{restatable}
\begin{proof}
To obtain the lower bound, note that the probability that some packet succeeds is at least:
 $$  \sum_{\mbox{\tiny all packets $u$}} \left( p_u(t) \prod_{\mbox{\tiny all packets $v$}}(1-p_v(t))\right)  
  \geq \frac{ C(t)}{e^{2C(t)}}$$  
  \noindent where the inequality follows from the left-hand side of Fact~\ref{fact:taylor}\ref{fact-a} and that $p_v\leq 1/2$ for all $v$, given that each packet always has a window size at least two. 
  
To obtain the upper bound, we have:
 \begin{eqnarray*} 
  && \sum_{\mbox{\tiny all packets $u$}} \left(p_u(t) \prod_{\mbox{\tiny all packets $v\not=u$}}(1-p_v(t))\right)\\
  &=& \sum_{\mbox{\tiny all packets $u$}} \left(p_u(t) \frac{(1-p_u(t))}{(1-p_u(t))} \prod_{\mbox{\tiny all packets $v\not=u$}}(1-p_v(t))\right)\\
  &\leq& e^{-C(t)} \left( \sum_{\mbox{\tiny all packets $u$}} \frac{p_u(t)}{1-p_u(t)}\right) \\
  & \leq & e^{-C(t)} 2C(t), 
  \end{eqnarray*}
  \noindent where the second-last line follows from applying the right-hand side Lemma~\ref{fact:taylor}, and the last line follows since $p_v\leq 1/2$ for all $v$.
\end{proof}

Note that any packet's window sizes is always at least $2$ under \OurAlg, so this lemma is always applicable for our algorithm.

Let \defn{$\pempty(t)$} denote the probability that slot $t$ is empty, and let 
\defn{$\pcollision(t)$}  denote the probability that slot $t$ is  noisy.  We have the following lower bounds on these probabilities.

\begin{restatable}[\textbf{Probability of an empty slot as a function of contention}]{myLemma}{empProb}\label{lem:p-emp}
%For any slot $t$, the probability that the slot is empty $\pempty(t) \geq 1-\Con$.
For any unjammed slot $t$, where $w_u(t) \geq 2$ for all packets, the following holds:
$$  e^{-2\Con} \leq \pempty(t) \leq e^{-\Con}.$$
\end{restatable}
\begin{proof}
Slot $t$ is empty when no packet sends in that slot, and this event occurs with probability:
\begin{align*}
\prod_{\mbox{\tiny all packets u}}(1-1/w_u(t)) & \geq e^{-2\sum_{\mbox{\tiny all packets u}} 1/w_u(t)} \\
& = e^{-2\Con}.
\end{align*}
\noindent where the first inequality holds by Fact~\ref{fact:taylor}\ref{fact-a} and because window sizes are always at least $2$. The upper bound is derived as follows:
\begin{align*}
\prod_{\mbox{\tiny all packets u}}(1-1/w_u(t)) & \leq e^{-\sum_{\mbox{\tiny all packets u}} 1/w_u(t)} \\
& = e^{-\Con}.
\end{align*}
where the first inequality holds by Fact~\ref{fact:taylor}\ref{fact-d}.
\end{proof}

\begin{restatable}[\textbf{Probability of a noisy slot as a function of contention}]
{lemma}{noiProb}\label{lem:p-noi}
For any unjammed slot $t$, where $w_u(t) \geq 2$ for all packets, $\pcollision(t) \geq 1 -  \frac{2\Con}{e^{\Con}} - \frac{1}{e^{\Con}}$.
\end{restatable}
\begin{proof}
The probability that slot $t$ contains noise is: 
\begin{align*}
1-\psuccess(t) - \pempty(t) &\geq 1 - \frac{2\Con}{e^{\Con}} -  \frac{1}{e^{\Con}}.
%& \geq 1 - \frac{3\Con}{e^{\Con}}.
\end{align*}
\end{proof}

Let \defn{$\psuccessgood$} be a lower bound on the probability of success when contention is good in a slot. Let \defn{$\pemptylow$} be a lower bound on the probability of an empty slot when contention is low in a given slot. Let  \defn{$\pcollisionhigh$} be a lower bound on the probability of a noisy slot when contention is high in slot $t$. Finally, let \defn{$\pemptygood$} be a lower bound on the probability of an empty slot given that contention is good (in which case we can establish such a lower bound using the fact that $\Con \leq \chigh$). A constant lower bound can be established for each event, so we set all of these quantities to $\Theta(1)$, which follows directly from Lemmas~\ref{lem:p-suc}, \ref{lem:p-emp}, and~\ref{lem:p-noi}. 

Throughout this paper, we often want to use Chernoff bounds to bound some random process over various intervals. The problem is that the probability distributions vary in each time slot, i.e., there are different probabilities of listening, sending, etc., depending on the contention in each slot as well as the window sizes of the packets that make up that contention. On top of that, we want to apply different analyses depending on whether contention is good, low, or high, and the number of slots within each contention category is something that is also determined by the adversary adaptively. To simplify enormously (and in some cases, enable) these Chernoff-style arguments, we use the following theorems from~\cite{KuszmaulQi21}, which generalize Chernoff bounds using a martingale style analysis, allowing the probability distribution in one time slot to depend on the outcomes of previous time slots.

\begin{theorem}[\textbf{Corollary~11 from~\cite{KuszmaulQi21}}]
Suppose that Alice constructs a sequence of random variables $X_1, \ldots, X_n$, with $X_i \in [0, c]$, $c > 0$, using the following iterative process: once the outcomes of $X_1, \ldots, X_{i-1}$ are determined, Alice then selects the probability distribution $\mathcal{D}_i$ from which $X_i$ will be drawn; $X_i$ is then drawn from distribution $\mathcal{D}_i$. Alice is an adaptive adversary in that she can adapt $\mathcal{D}_i$ to the outcomes of $X_1, \ldots X_{i-1}$. The only constraint on Alice is that $\sum_i \E[X_i \sim \mathcal{D}_i] \leq \mu$, that is, the sum of the means of the probability distributions $\mathcal{D}_1, \ldots , \mathcal{D}_n$ must be at most $\mu$.

Let $X = \sum_i X_i$. For any $\delta > 0$,
\[\Pr[X \geq (1 + \delta) \mu] \leq \exp{\left(- \frac{\delta^2 \mu }{(2 + \delta)c}\right)}.\]

Suppose the value of $\sum_i \E[X_i \sim D_i]$ is not fixed in advance, but is bounded below by $\mu' \geq 1$. Then, a similar result is achieved with a slightly weaker bound:
\begin{align*}
\Pr\left[X \geq 2(1 + \delta) \sum_i \E[X_i \sim D_i]\right] 
%&\leq 2 \log_2 \frac{c(2 + \delta)}{\delta^2} \exp\left( - \frac{\delta^2 \mu'}{(2+\delta)c}\right) \\
&\leq O\left(\exp{\left(- \frac{\delta^2 \mu' }{(2 + \delta)c}\right)}\right).
\end{align*}
\label{thm:adversarial_azumas_upper}
\end{theorem}
\begin{proof}
When $\mu$ is known in advance, this theorem is proved as Corollary~11 in~\cite{KuszmaulQi21}.  When $\mu$ is not known in advance, we can construct a sequence of such games for $\mu = 1, 2, 4, \ldots$. Each of these games uses the same sequence of distributions from Alice. If Alice's choices satisfy the $\mu$ for that specific game, then the above tail bound holds. Since the error probability decreases exponentially with $\mu$, we can take a union bound over all the games for which Alice's choices satisfy the requirement, with the error dominated by the smallest~$\mu \geq \mu'$. 
\end{proof}

\begin{theorem}[\textbf{Corollary~16 from \cite{KuszmaulQi21}}]
Suppose that Alice constructs a sequence of random variables $X_1, \ldots , X_n$, with $X_i \in [0, c], c > 0$, using the following iterative process. Once the outcomes of $X_1, \ldots, X_{i-1}$ are determined, Alice then selects the probability distribution $\mathcal{D}_i$ from which $X_i$ will be drawn; $X_i$ is then drawn from distribution $\mathcal{D}_i$. Alice is an adaptive adversary in that she can adapt $\mathcal{D}_i$ to the outcomes of $X_1, \ldots X_{i-1}$. The only constraint on Alice is that $\sum_i \E[X_i \sim \mathcal{D}_i] \geq \mu$, that is, the sum of the means of the probability distributions $\mathcal{D}_1, \ldots , \mathcal{D}_n$ must be at least~$\mu$.

Let $X = \sum_i X_i$. For any $\delta > 0$,
\[\Pr[X \leq (1 - \delta) \mu] \leq \exp{\left(- \frac{\delta^2 \mu }{2c}\right)}.\]

Suppose the value of $\sum_i \E[X_i \sim D_i]$ is not fixed in advance, but is bounded below by $\mu' \geq 1$. Then, the same result is achieved with a slightly weaker bound:

%If the value of $\mu$ is not fixed in advance (but is a bounded and well-defined result of Alice's choices, and at least $1$), then the same result is achieved with a slightly weaker bound:
\[\Pr\left[X \leq (1 - \delta) \sum_i \E[X_i \sim D_i] /2\right] \leq O\left(\exp{\left(- \frac{\delta^2 \mu' }{2c}\right)}\right).\]
\label{thm:adversarial_azumas_lower}
\end{theorem}
\begin{proof}
When $\mu$ is known in advance, this theorem is proved as Corollary~16 in~\cite{KuszmaulQi21}.
The rest of the proof follows identically to that of \thmref{adversarial_azumas_upper}.
\end{proof}

%%%%%%%%%%%%%%%%%%%%%%%%%%%%%%%%%%%%%%%%%%%%%

\subsection{ \texorpdfstring{\boldmath{$\first(t)+\second(t)$}}{} over single slots and intervals,  when contention is low, high, and good}\label{sec:second-analysis}
\hfill \\
\paragraph{Analyzing \boldmath $\first(t)$: over single slots and intervals,  when $\Con$ is good.}
We ease into  our analysis by starting with the first term of our potential function,  $\first(t)$, in the special case where contention is constant. Lemma~\ref{lem:A-term} analyses the expected decrease in $\first(t)$ on a per-slot basis when contention is good.

\begin{lemma}\label{lem:A-term}
Consider a  slot where contention is good, i.e., $\Con \in [\clow, \chigh]$, and the number of packet arrivals is $\mathcal{A}$. Then,  $\first(t)$ decreases by at least $\psuccessgood - \mathcal{A}=\Theta(1) - \mathcal{A}$ in expectation.
\end{lemma}
\begin{proof}
By Lemma~\ref{lem:p-suc}, $\psuccessgood(t) \geq \Con\,e^{-2\Con} \geq \clow\, e^{-2\chigh} =\Theta(1)$, which is thus the expected amount by which $\first(t)$ decreases. Packet arrivals increase $\first(t)$ by $1$ per packet arrival.
\end{proof}

\paragraph{How $\second(t)$ changes in expectation on a per-slot basis.} Observation~\ref{obs:busyB} shows that, in a noisy slot, $\second(t)$ decreases in expectation by $\Omega(\Con)$. Conversely, Observation~\ref{obs:silentB} shows that, in a
silent slot, $\second(t)$ increases in expectation by $O(\Con)$. These next two observations only offer intuition; however, for completeness, we provide their accompanying analysis.

\begin{restatable}{observation}{noisyexpectation}\label{obs:busyB}
Assume no packet arrivals and no jamming. Then, \[\E[\second(t+1)~|~t \mbox{~is noisy}, \second(t), \Con]  \leq \second(t) -\frac{\Con}{2c}.\]
\end{restatable}
\begin{proof}
$\second(t)$ is fixed (i.e., not a random variable) and each packet $u$ listens with probability $\frac{c\ln^3 (w_u(t))}{w_u(t)}$. If $u$ listens to slot $t$, then $w_u(t+1) \leftarrow w_u(t) \left(1+ \frac{1}{c\ln (w_u(t))}\right)$. Thus, we have:
\begin{align*}
\E[\second(t)-\second(t+1)] &=\sum \frac{c\ln^3 (w_u(t))}{w_u(t)}  \left( \frac{1}{c\ln (w_u(t))} - \frac{1}{c\ln (w_u(t+1))}\right)\\
&= \sum \frac{c\ln^3 (w_u(t))}{w_u(t)}  \left( \frac{1}{c\ln (w_u(t))} - \frac{1}{c\ln\left(w_u(t)\left(1+\frac{1}{c\ln (w_u(t))}\right)\right)} \right)
\end{align*}
\begin{align*}
~~~~~~~~~~~~~~~&\geq \sum \frac{c\ln^3 (w_u(t))}{w_u(t)} \left(  \frac{1}{c\ln (w_u(t))} - \frac{1}{c\ln (w_u(t))+\frac{1}{\ln (w_u(t))}}  \right)\\
&=\sum \frac{c\ln^3 (w_u(t))}{w_u(t)} \left( \frac{1}{c^2\ln^3 (w_u(t)) + c\ln (w_u(t))}\right)\\
& \geq \sum_u \frac{c\ln^3 (w_u(t))}{w_u(t)}  \left(\frac{1}{2c^2\ln^3 (w_u(t))}\right)\\
& = \frac{\Con}{2c}.
\end{align*}
\noindent The third line follows from Fact~\ref{fact:taylor}\ref{fact-b} using $x=\frac{1}{c\ln (w_u(t))}$ and noting that   $c\ln\left(1+\frac{1}{c\ln (w_u(t))}\right) \leq \frac{1}{\ln (w_u(t))}$. By linearity of expectation, we have $\E[\second(t+1) | \second(t)]  \leq \second(t) -\frac{\Con}{2c}$, as claimed. % The fifth line holds because $c\ln^2 w(t)\geq 1$. 
\end{proof}

\begin{restatable}{observation}{silentexpectation}
\label{obs:silentB}
Assume no packet arrivals and no jamming. Then, \[\E[\second(t+1)~|~t \mbox{~is silent}, \second(t), \Con ] \leq \second(t) + \frac{2\Con}{c}.\]
\end{restatable}
\begin{proof}
$\second(t)$ is fixed and each packet $u$ listens with probability $c\ln^3 (w_u(t))/w_u(t)$. If $u$ hears slot $t$, then $w_u(t+1) \leftarrow 
w_u(t) / \big(1+ \frac{1}{c\ln (w_u(t))}\big)$. 
Thus, we have:
\begin{align*}
\E[\second(t+1)-\second(t)] \hspace{-2pt} &=  \sum \frac{c\ln^3 (w_u(t))}{w_u(t)}  \left( \frac{1}{c\ln (w_u(t+1))} - \frac{1}{c\ln (w_u(t))}\right)\\
&= \sum \frac{c\ln^3 (w_u(t))}{w_u(t)}  \left( \frac{1}{c\ln\left(w_u(t)/\left(1+\frac{1}{c\ln w_u(t)}\right)\right)} - \frac{1}{c\ln w_u(t)}\right)\\
&=\hspace{-1pt} \sum \hspace{-2pt} \frac{c\ln^3 (w_u(t))}{w_u(t)} \hspace{-3pt}  \left(\hspace{-2pt} \frac{1}{c\ln (w_u(t)) \hspace{-2pt} - \hspace{-2pt} c\ln\left(1+\frac{1}{c\ln w_u(t)}\right)} - \frac{1}{c\ln (w_u(t))}\hspace{-2pt}\right)\\ % hspaces are fixing overfull hbox
&\leq \sum \frac{c\ln^3 (w_u(t))}{w_u(t)} \left(\frac{1}{c\ln (w_u(t))-\frac{1}{\ln (w_u(t))}} - \frac{1}{c\ln (w_u(t))} \right)\\
&=\sum \frac{c\ln^3 (w_u(t))}{w_u(t)} \left( \frac{1}{c^2\ln^3 (w_u(t)) - c\ln (w_u(t))}\right)\\
& \leq \sum_u \frac{c\ln^3 (w_u(t))}{w_u(t)}  \left(\frac{2}{c^2\ln^3 (w_u(t))}\right)\\
& = \frac{2\Con}{c}.
\end{align*}
\noindent The fourth line follows from  Fact~\ref{fact:taylor}\ref{fact-b} using $x=\frac{1}{c\ln (w_u(t))}$. By linearity of expectation, we have $\E[\second(t+1)]  \leq \second(t) + \frac{2\Con}{c}$, as claimed.
\end{proof}

\paragraph{How one packet listening in  slot $t$ changes $\second(t)$.}
Lemma~\ref{lem:B_term_increase_inverse_log_cubed} illustrates how $\second(t)$ increases or decreases due to a specific packet listening to an empty or noisy slot.  Going forward, we simplify our notation by omitting $t$ with respect to $w_u$, where $t$ is always made clear from the context.

\Bterminc

\begin{proof}
Packet $u$'s contribution to $\second(t)$ is $1 / \ln( w_u)$. After listening to a silent slot, $u$'s window size becomes $w_u / (1+ \frac{1}{c\ln w_u})$. Thus the change to its term in $\second(t)$ is
\begin{align*}
\frac{1}{\ln w_u - \ln(1 + \frac{1}{c \ln w_u})} - \frac{1}{\ln w_u} &= 
{\textstyle
\Theta\big(\ln \big(1 + \frac{1}{c \ln w_u} \big) / \ln^2 w_u \, \big)
}\\ 
&= \Theta(1 / (c\ln w_u) / \ln^2 w_u) = \Theta(1 / (c\ln^3w_u))\,.
\end{align*}
Similarly, if the slot is noisy, $u$'s window size becomes $w_u \cdot (1+1/c\ln w_u)$. Thus, the change to $u$'s contribution to $\second(t)$ is
\begin{align*}
    \frac{1}{\ln w_u + \ln (1+1/(c\ln w_u))} - \frac{1}{\ln w_u} &= -\Theta(\ln(1+1/(c\ln w_u) / \ln^2w_u))\\ &= -\Theta(1/(c \ln^3 w_u))\,.
\end{align*}

\end{proof}

\paragraph{Net behavior of $\first(t)+ \second(t)$ during low and good  contention slots.} 
We now bound the behavior of $\first(t)$ and $\second(t)$ during non high contention slots. Specifically, we show that these terms do not increase too much in low-contention slots and that they decrease in good-contention slots, with high probability.

In our arguments, we will speak of \defn{sending attempts} and \defn{listening attempts}, which refer to the probabilistic action of a packet trying to send or listen to the channel at a time slot. If the packet succeeds, then the sending attempt is successful. If the packet actually listens, then we say that the listening attempt is successful.

\ABlowmid
 
\begin{proof}
Since we are only considering slots in the low contention regime and good contention regime, $H(t)$ will generally not be decreasing significantly. Instead, our line of argument will be to show that $\first(t)$ will decrease over slots with good contention, overwhelming any increase from $\second(t)$. And in slots with low contention, $\first(t)$ will decrease sufficiently to counterbalance any increase from $\second(t)$ (up to an additive $O(\ln^2 \tau)$ term).

In this lemma, we are going to analyze the sending attempts and the listening attempts separately.  When a sending attempts are successful, they reduce $\first(t)$; when listen attempts are successful, they (may) increase $\second(t)$ (if the slot is empty).  We will show that in each case, the change in $\first(t)$ or $\second(t)$ is proportional to the sum of contention in the slots being considered, with high probability. By choosing $\alpha_1$ sufficiently large with respect to $\alpha_2$, we can ensure that reduction of $\first(t)$ due to successes dominates the increase of $\second(t)$ due to listening.

The adversary, being adaptive, has some control over which slots packets are listening/sending in, and which slots have non-high contention.  Throughout this proof, we think of the process as a game where in each slot the adversary adaptively dictates: (i)~what the contention is, and (ii)~what the window sizes are that lead to this contention. Thus, we allow the adversary some additional power.\footnote{The adversary can deterministically increase the contention in a particular slot by jamming or injecting new packets. But the adversary may also be able to encourage a decrease in subsequent slots by jamming: if the slot is jammed, then any listening packets will increase their window size and hence decrease contention.  Note that in allowing the adversary to exactly dictate the contention, we are giving it more power than it actually has.\label{fn:contention}}

\paragraph{Upper bound for listening: bounding the increase to $\second$.} First, we focus on listening attempts. In this proof, we are going to assume (the worst case) that every time a packet listens, it hears silence and thus decreases its window size, increasing the $\second(t)$ term. As we have seen in \lemref{B_term_increase_inverse_log_cubed}, if a packet with window size $w$ listens to a silent slot, then its  contribution to $\second(t)$ increases by $\Theta(1 / \ln^3 w)$ additively.

Let $X_1, X_2, \ldots $ be random variables associated with listening attempts in $\mathcal{I}$ during slots in either $\mathcal{L}$ or $\mathcal{G}$.  (The same analysis holds for both of these subsets of slots.) Each $X_j$ is associated with a particular packet at a particular time slot. For random variable $X_j$, if the associated listener has window size $w$ and the listen succeeds (i.e., the packet does actually listen), we define  $X_j = 1 / \ln^3 w$; otherwise, $X_j = 0$.  Notice that each $X_j$ is in the range $[0,1]$, and $X = \sum X_i$ upper bounds the total change, within constant factors, of $\second(t)$ during the slots being analyzed by \lemref{B_term_increase_inverse_log_cubed}.

%\john{We need to address reviewer issue about adversary influencing contention. Can influence more broadly than just increasing}
Let $\mu_x = \expect{X}$.  While the adversary can adaptively determine the contention in each slot, we can still express $\mu_x$ as a function of the selected contention.  Let $X^t$ be the sum of the $X_j$ associated with time slot $t$.  Since $\expect{X_j} = \Theta(1/w)$, if the corresponding packet has window size $w$, we know that $\expect{X^t} = \Theta(\Con)$.  Thus, if we are analyzing the slots in $\mathcal{L}$: $\mu_x = \Theta(\sum_{t \in \mathcal{L}} \Con)$ (i.e., the sum of the contention of low contention slots).
If we are analyzing the slots in $\mathcal{G}$: $\mu_x = \Theta(\sum_{t \in \mathcal{G}} \Con)$ (i.e., the sum of the contention of good contention slots).

We can now apply \thmref{adversarial_azumas_upper} to show that $X$ is close to $\mu_x$. Notice that the precise distribution for each $X_i$ is determined adaptively by the adversary, e.g., by affecting which slots are high or low contention. Thus, the distribution $D_i$ of $X_i$ is determined by the adversary and depends on the outcomes of $X_1, \ldots X_{i-1}$, and we have $\sum_{i} \E[X_i \sim D_i] = \mu_x$. Therefore, we can apply \thmref{adversarial_azumas_upper} to get:
\[\Pr[X \geq 2(1 + \delta) \mu_x] \leq O\left(\exp{\left(- \frac{\delta^2 \mu_x }{(2 + \delta)}\right)}\right).\]
If $\mu_x$ is $O(\ln{\tau})$, we set $\delta = \Theta(\ln{\tau})$, and conclude that with high probability in $\tau$, $X \leq O(\ln^2(\tau))$.  Otherwise, if $\mu_x$ is $\Omega(\ln(\tau))$, we can choose $\delta = O(1)$, and observe that with high probability in $\tau$, $X \leq O(\mu_x)$.  

We conclude, then, that the following two claims hold:
\begin{itemize}
\item The change in $\second(t)$ during $\mathcal{L}$ is at most $O(\sum_{t \in \mathcal{L}} \Con + \ln^2(\tau))$, with high probability in $\tau$.
\item The change in $\second(t)$ during $\mathcal{G}$ is at most $O(\sum_{t \in \mathcal{G}} \Con + \ln^2(\tau))$, with high probability in $\tau$.
\end{itemize}

\paragraph{Lower bound for sending: lower bounding the decrease to $\first$.} Next, we similarly consider sending attempts.  We will temporarily ignore jamming, and determine the number of successful sending attempts if there was no jamming.  Let $Y_1, \ldots Y_n$ be random variables associated with time slots in $\mathcal{L}$ or in $\mathcal{G}$.  (As before, we analyze these two regimes simultaneously.)  Define each $Y_j = 1$ if there is a success in the slot associated with $Y_j$, and $Y_j = 0$ otherwise.  Notice that each $Y_j$ is in the range $[0,1]$, and $Y = \sum Y_i$. Let $\mu_y = \sum_{i} \E[Y_i \sim D_i]$, where $D_i$ denotes the probability distribution (partially determined by the adversary) of a success for $Y_i$ given the outcomes of $Y_1, \ldots, Y_{i-1}$.

Again, we relate $\mu_y$ to the contention values in the relevant slots.  For a given slot $t$, the probability $\Pr(Y_j = 1) = \Theta(\Con / e^{2\Con}) = \Theta(\Con)$---this follows from \lemref{p-suc}, because we are considering only slots where $\Con \leq \chigh = O(1)$.  Thus, when analyzing $\mathcal{L}$: $\mu_y = \Theta(\sum_{t \in \mathcal{L}} \Con)$ (i.e., the sum of the contention of the low contention slots).  When analyzing $\mathcal{G}$: $\mu_y = \Theta(\sum_{t \in \mathcal{G}} \Con)$ (i.e., the sum of the contention of the good contention slots).

We can now apply \thmref{adversarial_azumas_lower} to show that $Y$ is close to its expectation $\mu_y$. As with $X_i$ above, the precise distribution for each $Y_i$ is determined adaptively by the adversary, e.g., by affecting which slots are high or low contention. Again, the distribution $D_i$ of $Y_i$ is determined by the adversary and depends on the outcomes of $Y_1, \ldots Y_{i-1}$, with $\sum_i \E[Y_i \sim D_i] = \mu_y$. Therefore, we can apply \thmref{adversarial_azumas_lower} to get: 
\[\Pr[Y \leq (1 - \delta) \mu_y/2] \leq O\left(\exp{\left(- \frac{\delta^2 \mu_y }{2}\right)}\right).\]
If $\mu_y$ is $O(\ln{\tau})$, then we trivially conclude that $Y \geq \mu_y - O(\ln(\tau))$.  Otherwise, if $\mu_y$ is $\Omega(\ln(\tau))$, we can choose $\delta = \Theta(1)$, and observe that with high probability in $\tau$, $Y \geq \Theta(\mu_y)$.  

So far, we have analyzed the number of successes, ignoring jamming.  In fact, up to $\mathcal{J}_\mathcal{L}$ or $\mathcal{J}_\mathcal{G}$ slots may be jammed, and each jammed message reduces the number of successes by at most one.  We conclude, then, that the following two claims hold:
\begin{itemize}[noitemsep]
\item The change in $\first(t)$ during $\mathcal{L}$ is at most $O(\ln(\tau)) - \Omega(\sum_{t \in \mathcal{L}} \Con) + \mathcal{A}_\mathcal{L} + \mathcal{J}_\mathcal{L}$, with high probability in $\tau$.
\item The change in $\first(t)$ during $\mathcal{G}$ is at most $O(\ln(\tau)) - \Omega(\sum_{t \in \mathcal{G}} \Con) + \mathcal{A}_\mathcal{G} + \mathcal{J}_\mathcal{G}$, with high probability in $\tau$.
\end{itemize}
Taking a union bound over the above four bulleted claims (for each of the two terms, there is a claim for each of the two contention regimes), all of them hold with high probability in $\tau$.

\paragraph{Combining sending and listening.} Combining our previous claims for sending and listen, the net delta over $\mathcal{L}$ is at most: 
$$\alpha_2 O\big(\sum_{t \in \mathcal{L}} \Con\big) - \alpha_1 \Omega\big(\sum_{t \in \mathcal{L}} \Con\big) + (\alpha_1 + \alpha_2) O(\ln^2(\tau)) + \alpha_1( \mathcal{A}_{\mathcal{L}} + \mathcal{J}_{\mathcal{L}}), $$
with high probability in $\tau$.  Analogously, for $\mathcal{G}$, we have shown that the net delta over $\mathcal{G}$ is: 
$$\alpha_2  O\big(\sum_{t \in \mathcal{G}} \Con\big) - \alpha_1 \Omega \big(\sum_{t \in \mathcal{G}} \Con\big) + (\alpha_1 + \alpha_2) O(\ln^2(\tau)) + \alpha_1( \mathcal{A}_{\mathcal{G}} + \mathcal{J}_{\mathcal{G}}), $$
with high probability in $\tau$.

Hence, by choosing $\alpha_1$ to be a constant sufficiently larger multiplicatively than $\alpha_2$, the decrease in $\first$ due to packet successes overwhelms the (possible) increase in $\second$ due to listening.
We conclude that (with high probability in $\tau$), the net delta over $\mathcal{L}$ is at most  $$(\alpha_1 + \alpha_2) O(\ln^2(\tau)) - O\big(\alpha_1 \sum_{t \in \mathcal{L}} \Con)\big) + \alpha_1 ( \mathcal{A}_{\mathcal{L}} + \mathcal{J}_{\mathcal{L}}), $$
implying that the net delta over $\mathcal{L}$ is at most $O(\ln^2 \tau) + \alpha_1( \mathcal{A}_\mathcal{L} + \mathcal{J}_{\mathcal{L}})$
w.h.p.\ in~$\tau$.

For $\mathcal{G}$, we observe something stronger: since $\Con \geq \clow = \Omega(1)$ in each slot, we know that $\sum_{t \in \mathcal{G}} \Con = \Omega(|\mathcal{G}|)$.  Thus, for $\mathcal{G}$, we conclude that (with high probability in $\tau$), the net delta over $\mathcal{G}$ is at most  $$(\alpha_1 + \alpha_2) O(\ln^2(\tau)) - \alpha_1 \Omega(|\mathcal{G}|) + \alpha_1 ( \mathcal{A}_{\mathcal{L}} + \mathcal{J}_{\mathcal{L}}), $$
implying that the net delta over $\mathcal{G}$ is at most $O(\ln^2 \tau) + \alpha_1 (\mathcal{A}_\mathcal{G} + \mathcal{J}_{\mathcal{G}}) - \Omega(|\mathcal{G}|)$
w.h.p.\ in~$\tau$.
\end{proof}

%\mab{Hi John, there's a different definition for $\mathcal{L}$ in this lemma, and the previous one. How problematic is this? If you feel it should be fixed, please do so.}

\paragraph{Cumulative behavior of $\first(t)+ \second(t)$ during high-contention slots.} 
The previous lemma considered the slots of the interval where the contention was not high; the next lemma considers the slots of the interval where the contention is high.  In this case, the net delta decreases due to decreases in $\second(t)$.

% ========================================
% RESTATABLE USED HERE
% Leaving the old lemma statement to make it easier for us to read while editing the proof
% =========================================
% \begin{lemma}
% \label{lem:AB_high}{\bf(Net delta in contribution of {\boldmath $\first(t)$} and {\boldmath $\second(t)$} to potential over high  contention slots).}
% Let $\mathcal{I}$ be an arbitrary interval starting at time $t$ with length $|I|=\tau$.
% Let $\mathcal{H}$
% be the set of slots in $\mathcal{I}$ during which $\Con > \chigh$. 
% Let $\mathcal{A}_{\mathcal{H}}$ be the number of packet arrivals in time slots in $\mathcal{H}$, and let $\mathcal{J}_{\mathcal{H}}$ be the number of jammed slots in $\mathcal{H}$,

% Define the $\defn{net delta}$ over $\mathcal{H}$ to be the sum of the changes in $\first$ and $\second$ during the slots in $\mathcal{H}$, i.e., \[\sum_{t' \in \mathcal{H}}\Big(\alpha_1\big(\first(t'+1) - \first(t')\big) + \alpha_2\big(\second(t'+1) - \second(t')\big)\Big).\]

% \noindent Then the net delta over $\mathcal{H}$ is at most $O(\ln ^3 \tau ) + \alpha_1 \mathcal{A}_\mathcal{H} - \Omega(|\mathcal{H}|)$ w.h.p.\ in~$\tau$.  
% \end{lemma}
\ABhigh
\begin{proof}

For our line of argument, we do not require $\first(t)$ to decrease (and indeed packets do not often succeed when contention is high), and so simply bound the increase to $\first(t)$ by $\mathcal{A}_\mathcal{H}$ due to packet arrivals. We focus for the rest of the proof on the change in $\second(t)$ over $\mathcal{H}$.

When the contention is high, a large constant fraction of slots (in expectation) in $\mathcal{H}$ are noisy, and in each noisy slot, $\second(t)$ decreases. Some of the slots (up to a constant fraction in expectation) will be empty, however, and any packets that listen in such slots will decrease their window size, \defn{increasing} $\second(t)$ (rather than decreasing it). We will show that the decrease is sufficiently more than the increase to achieve the desired result.

Since the adversary is adaptive, it has some control over the contention in each slot. 
Throughout this proof, %we think of this as a game where in each slot the adversary adaptively determines: 
we allow the adversary some additional power, namely the ability to adaptively dictate (i)~what the contention is, and (ii)~what the window sizes are that lead to this contention. (See footnote\footref{fn:contention} for a discussion of the adversary's influence over contention.) %(Thus we allow the adversary some additional power.) 
We analyze the change in $\second(t)$ over $O(\ln(\tau))$ contention classes: $\chigh, {\chigh+1}, \ldots, \Theta(\ln(\tau))$.  Fix $C$ to be one of those contention classes, and define $\mathcal{H'}$ to be the slots in $\mathcal{H}$ with contention in the range $(C, C+1)$. (The last class includes all slots with contention $\Theta(\ln(\tau))$ or larger.)  

We separately bound the number of listens in noisy slots (which decrease $\second(t)$) and in empty slots (which increase $\second(t))$.  Let $\mathcal{E}$ be the subset of slots in $\mathcal{H'}$ that are empty; let $\mathcal{F}$ be the subset of slots in $\mathcal{H'}$ that are noisy.

%(i.e., (F)ull).\maxwell{reviewer didn't like that "(F)ull"; need to address.}

By \lemref{B_term_increase_inverse_log_cubed}, if a packet with window size $w$ listens to a silent slot, then its  contribution to $\second(t)$ increases by $\Theta(1 / \ln^3 w)$ additively; if it listens to a noisy slot, then its contribution to $\second(t)$ decreases by $\Theta(1 / \ln^3 w)$.  Notice that jamming actually increases the likelihood that the slot is noisy (to a certainty), and so only helps to decrease $\second(t)$.

\paragraph{Analyzing noisy slots.} Let $X_1, X_2, \ldots $ be random variables associated with potential listens in $\mathcal{H}$ during slots in $\mathcal{F}$.  Each $X_j$ corresponds to a potential listen, i.e., a particular active packet (which listens probabilistically) at a particular slot. If that packet has window $w$ and decides to listen, then we define $X_j = 1/\ln^3(w)$; otherwise, $X_j = 0$.
Since the slots are (by definition) noisy, the contribution of that particular packet at that time slot to $\second(t)$ will decrease by $\Omega(X_j)$ by \lemref{B_term_increase_inverse_log_cubed}.

There remains one challenge with analyzing the $X_j$s: we have restricted our attention to noisy slots, and in noisy slots it is more likely that a given packet has listened (and then chosen to send).  Luckily, knowing that a slot is noisy only increases the likelihood of listening, so we can conclude that $\Pr(X_j \neq 0) \geq \Theta(\ln^3(w_u)/w_u)$.\footnote{In more detail, we can apply Bayes' Theorem, and we have $\Pr(\text{listen} | \text{noisy}) = \frac{\Pr(\text{noisy} | \text{listen}) \cdot \Pr(\text{listen})}{\Pr(\text{noisy})} \geq \Pr(\text{listen})$ as long as $\Pr(\text{noisy}|\text{listen}) \geq \Pr(\text{noisy})$. And given that packets only choose to send when it also listens, knowing that a packet listens increases the likelihood of noise.} Thus we define a new set of independent $\hat{X_j}$ random variables where  $\hat{X_j} = 1 / \ln^3(w)$ with probability $\Theta(\ln^3(w_u)/w_u)$ which are stochastically dominated by %(from below) 
the original $X_j$ random variables.

%\john{Continue review}

Notice that $\hat{X} = \sum \hat{X_j}$ determines the total change, within constant factors, of $\second(t)$ during the slots being analyzed.  Let $\mu_x = \expect{\hat{X}}$. Since $\expect{\hat{X_j}} = \Theta(1/w)$, if the packet has window size $w$, we know that $\mu_x = \Theta(\sum_{t \in \mathcal{F}} \Con)$ (i.e., the sum of the contention of slots in $\mathcal{F}$).

Since each $\hat{X_j}$ is in the range $[0,1]$ and they are independent, we can now apply \thmref{adversarial_azumas_lower} to show that $\hat{X}$ is close to $\mu_x$. Notice that the precise distribution for each $\hat{X_j}$ is determined adaptively by the adversary, e.g., by affecting which slots are high or low contention. However, the specific coin flips for each listening attempt are independent.

\thmref{adversarial_azumas_lower} shows that: \[\Pr[\hat{X} \leq (1 - \delta) \mu_x/2] \leq O\left(\exp{\left(- \frac{\delta^2 \mu_x }{2}\right)}\right).\]
If $\mu_x$ is  $O(\ln{\tau})$, then we trivially conclude that $\hat{X} \geq \mu_x - O(\ln(\tau))$.  Otherwise, if $\mu_x \geq \Omega(\ln(\tau))$, we can choose $\delta = \Theta(1)$, and observe that with high probability in $\tau$, $\hat{X}$ is $ \Theta(\mu_x)$. 

Since the $\hat{X_j}$ stochastically dominate the $X_j$, we conclude, then, that $\second(t)$ during $\mathcal{F}$ decreases by at least $\Omega(\sum_{t \in \mathcal{F}} C - \ln^2(\tau))$, with high probability in $\tau$.

\paragraph{Analyzing empty slots.} We can now apply the same analysis to the empty slots.  Let $X_1, X_2 \ldots$ be random variables associated with  potential listen, i.e., a particular active  packet at a particular time slot. If that packet has window $w$ and decides to listen, then we define $X_j = 1/\ln^3(w)$; otherwise, $X_j = 0$.  Since the slots are (by definition) empty, the contribution of that particular packet at that time slot to $\second(t)$  will increase by $\Omega(X_j)$ by \lemref{B_term_increase_inverse_log_cubed}.

Knowing that a slot is empty only decreases the likelihood of listening, so we can conclude that $\Pr(X_j \neq 0) \leq \Theta(\ln^3(w_u)/w_u)$.\footnote{The probability of an empty slot, given that a packet is listening, is less than the probability of an (unconditional) empty slot, and so Bayes' law says (via the same argument as in the prior footnote) that the probability of listening in an empty slot is less than the probability of (unconditional) listening.}  Thus we define a new set of independent $\hat{X_j}$ random variables where  $X_j = 1 / \ln^3 w$ with probability $\Theta(\ln^3(w_u)/w_u)$ which stochastically dominate (from above) the original $X_j$ random variables.

Notice that each $\hat{X_j}$ is in the range $[0,1]$, and $\hat{X} = \sum \hat{X_i}$ bounds the total change, within constant factors, of $\second(t)$ during the slots being analyzed.  Let $\mu_x = \expect{\hat{X}}$. Since $\expect{\hat{X_j}} = \Theta(1/w)$, if the packet has window size $w$, we know that $\mu_x = \Theta(\sum_{t \in \mathcal{E}} \Con)$ (i.e., the sum of the contention of slots in $\mathcal{E}$).

We can now apply \thmref{adversarial_azumas_upper} to show that $\hat{X}$ is close to $\mu_x$.  \thmref{adversarial_azumas_upper} shows that: \[\Pr[\hat{X} \geq 2(1 + \delta) \mu_x] \leq O\left(\exp{\left(- \frac{\delta^2 \mu_x }{(2 + \delta)}\right)}\right).\]
If $\mu_x < O(\ln{\tau})$, we set $\delta = \Theta(\ln{\tau})$, and conclude that with high probability in $\tau$, $\hat{X} \leq O(\ln^2(\tau))$.  Otherwise, if $\mu_x \geq \Omega(\ln(\tau))$, we can choose $\delta = O(1)$, and observe that with high probability in $\tau$, $\hat{X} \leq O(\mu_x)$. 

Since the $\hat{X_j}$ stochastically dominate the $X_j$, we conclude, then, that $\second(t)$ during $\mathcal{E}$ increases by at most $O(\sum_{t \in \mathcal{E}} (C+1) + \ln^2(\tau))$, with high probability in $\tau$, assuming that we are not considering the ``last'' contention class containing all the slots with contention $\Theta(\ln \tau)$ or higher (which we will handle separately).

\paragraph{Relating the noisy and empty slots.} We now need to combine our analysis of noisy and empty slots, and then sum across all the different contention classes.  The remaining observation is that a large constant fraction of the slots are noisy.  The probability that a slot $t$ is noisy is at least $(1 - e^{-\Theta(\Con}) \geq 3/4$ (for appropriate choice of $\chigh$), whereas the probability that a lost $t$ is empty is at most $1/4$. So most of the slots in each class will be noisy.  

There are now three cases to deal with: (i) $C < O(\ln(\tau))$ and the number of slots is small, i.e., $\mathcal{H'} < O(\ln(\tau))$; (ii)  $C < O(\ln(\tau))$ and the number of slots is not small, i.e., $\mathcal{H'} > \Omega(\ln(\tau))$; (iii) $C > \Omega(\ln(\tau))$.

In the first case, if $\mathcal{H'} < O(\ln(\tau))$, we will simply assume the worst case, i.e., all the slots are empty.  As we have shown above, this implies that with high probability in $\tau$, the term $\second(t)$ increases by at most $O(C \ln(\tau) + \ln^2(\tau)) = O(\ln^2(\tau))$ during these slots.  Since there are only $O(\ln(\tau))$ different contention classes, this will increase $\second(t)$ by at most $O(\ln^3(\tau))$ throughout the entirety of $\mathcal{H}$.

In the second case, if $\mathcal{H'} > \Omega(\ln(\tau))$: if there are $k$ slots in $\mathcal{H'}$, then we expect $(3/4)k$ of them to be noisy; we can again use \thmref{adversarial_azumas_lower} to show that at least $(2/3)k$ of them are noisy, and at most $(1/3)k$ of them are empty (e.g., choosing $\delta = 1/9$), with high probability in $\tau$.  Thus, with high probability in $\tau$, $\second(t)$ will decrease by at least $\Omega(k((2/3)C - (1/3)(C+1)) - \ln^2(\tau)) = \Omega(|\mathcal{H'}| - \ln^2(\tau))$.  (Notice that jammed slots only increase the number of slots that are noisy in the above analysis.)

In the third case, for the ``last'' contention class where $C > \Omega(\ln \tau)$, we observe that a slot is empty with probability at most $e^{-\Theta(\Con} \leq 1/\tau^{\Theta(1)}$.  Thus, with high probability (by a union bound) every slot in $\mathcal{H'}$ is noisy. Thus, by the analysis above, we know that with high probability in $\tau$, $\second(t)$ decreases by at least $\Omega(|\mathcal{H'} \ln(\tau)| - \ln^2(\tau))$.

Overall, there are at least $|\mathcal{H}| - O(\ln^2(\tau))$ slots in the second and third cases. So summing up over all the cases (and taking a union bound over everything), we conclude that with high probability in $\tau$, $\second(t)$ decreases by at least $\Omega(|\mathcal{H}| - \ln^3(\tau))$.
\end{proof}

\paragraph{Combined analysis of $\first(t)+\second(t)$ over an interval for all slots.}
The next lemma aggregates Lemmas~\ref{lem:AB_low_mid} and~\ref{lem:AB_high} to characterize the behavior of $\first(t) + \second(t)$ over all slots in an arbitrary interval. Specifically, \lemref{AB_not_high_and_decrease} shows that, for an arbitrary interval of length $\tau$, it is either the case that almost all of the slots are low contention slots, or $\first(t) + \second(t)$ decreases by $\Omega(\tau)$ with high probability, blunted by the number of packet injections and jammed slots. This lemma addresses {\it all} slots, rather than being limited to those of a particular contention regime. 

%These last two lemmas (\ref{lem:AB_not_high_and_decrease} and \ref{lem:tail-inequality-second}) are the only ones that will be used later in the analysis, but the earlier lemmas in the subsection are necessary to build up to them. 

% ========================================
% RESTATABLE USED HERE
% Leaving the old lemma statement to make it easier for us to read while editing the proof
% =========================================
% \begin{lemma}[\textbf{Unless most slots have low contention,  \boldmath $\alpha_1\first(t)+\alpha_2\second(t)$ decreases}]
% Let $\mathcal{I}$ be an arbitrary interval of length $\tau > \Omega(1)$ with $\mathcal{A}$ packet arrivals and $\mathcal{J}$ jammed slots. With high probability in $\tau$, at least one of the following two conditions holds:

% \begin{itemize}[noitemsep]
%     \item Less than $1/10$ of slots satisfy $\Con \geq \clow$.
%     \item $\alpha_1\first(t) + \alpha_2\second(t)$ decreases by $\Omega(\tau) - O(\mathcal{A} + \mathcal{J})$ over $\mathcal{I}$.
% \end{itemize}

% \noindent Additionally, $\alpha_1\first(t) + \alpha_2\second(t)$ increases by at most $O(\ln^3 \tau + \mathcal{A} + \mathcal{J})$ w.h.p.\ in $\tau$.

% \label{lem:AB_not_high_and_decrease}
% \end{lemma}
\ABnothigh
 %This lemma follows immediately from \lemreftwo{AB_low_mid}{AB_high}.
\begin{proof}
  Let $\mathcal{L}$ be the slots in $\mathcal{I}$ with contention $< \clow$; let $\mathcal{G}$ be the slots in $\mathcal{I}$ with contention $\geq \clow$ and $\leq \chigh$; let $\mathcal{H}$ be the slots in $\mathcal{I}$ with contention $> \chigh$.  
Combining the results from \lemreftwo{AB_low_mid}{AB_high}, we conclude that, with high probability in $\tau$, the net delta in $\alpha_1\first(t) + \alpha_2\second(t)$ is at most:
$$
\alpha_1( \mathcal{A} + \mathcal{J}) + O(\ln^2(\tau)) + O(\ln^3(\tau)) - \Omega(|\mathcal{G}|) - \Omega(|\mathcal{H}|)
$$
If $|\mathcal{G}| + |\mathcal{H}| > (1/10)\tau$ (and noting that $\tau > \ln^3(\tau)$), then the net delta in $\alpha_1\first(t) + \alpha_2\second(t)$ is at most $O(\mathcal{A} + \mathcal{J}) - \Omega(\tau)$.

For the case that $|\mathcal{G}| + |\mathcal{H}| \leq (1/10)\tau$ (and in fact in both cases), we can omit the subtracted terms, to give an upper bound of $O(\ln^3 \tau + \mathcal{A} + \mathcal{J})$ w.h.p.\ in $\tau$.
\end{proof}

% \begin{lemma}
%     Suppose that $\Con > \wmax^{10}$, and let $\mathcal{A}$ be the number of packet arrivals in slot $t$. Then with high probability in $\Con$, $\second(t)$ will decrease by $\Omega(\Con) - O(\mathcal{A})$ in slot $t$.
% \end{lemma}

\subsection{Amortized behavior of \texorpdfstring{\boldmath{$\third(t)$}}{}}\label{sec:third-analysis}

%For the analysis of $\third(t)$, the size of the intervals starts to come into play.  \jeremy{Should be referencing something earlier}. Recall that the size of an interval starting at time $t$ is chosen to be $\tau = (1/\cint) \cdot \max\big\{\frac{\wmax}{\ln^2( \wmax)}, \first(t)^{1/2}\big\}$, where $\cint \geq 1$ is a constant that controls the interval size. 
%\mab{(In following proofs, we will be able to  make $\cint$ as big as we want by settings $c$ and $\wmin$ in the algorithm to be as big as necessary.) }

%\mab{To get a full tail bound for the betting game:  Suppose that the number of packets that arrives during the interval is $\mathcal{A}$. There are two cases, depending on whether $\mathcal{A}$ is big or small, where small  means less than, say $\tau^{20}$ and big means greater than $\tau^{20}$. In the small case for $\mathcal{A}$, the same tail bound for the \lemref{tailbound-maxwindow} still applies because the w.h.p.\ bounds still apply.  For the case of big $\mathcal{A}$, we want to charge things to the first term. E.g., since the first term of $\Phi$ increases by $\mathcal{A}$, in order to offset that, the third term would need to increase, not to $k\cdot \tau$, but rather to $(k+ \mathcal{A}/\tau) \cdot \tau$. That's how we get the tail bound to work in the case when there are lots of arrivals during the interval.}

Consider an interval $\mathcal{I}$.  Because $|\mathcal{I}| \geq (1/\cint) \wmax / \ln^2 \wmax$, note that $|\mathcal{I}| \geq \wmin/\ln^2\wmin$, where $\wmin$ is a constant of our choice that specifies the smallest window size.

We next provide bounds on how the window size of each packet changes during an interval. Proving these bounds amounts to bounding the number of times a packet listens; the actual traffic on the channel is irrelevant as the proofs are pessimistic, and hence the number of jammed slots does not appear in any of the lemma statements. Naturally, the size $\tau$ of the interval affects the number of times that a packet listens during the interval. We thus reason about a packet's window size relative to a baseline $Z$ that satisfies $Z/\ln^2 Z = \tau$.  In some sense, $Z$ is the window size that ``matches'' the interval size because a packet with window $Z$ should listen $\Theta(c\ln(Z))$ times during the interval, which is the necessary number of times to change the window size by a constant factor. 

The parameter $\cint$ will allow us to later tune how $Z$ relates to $\wmax$---in the case that $\wmax \geq \first(t)^{1/2}$, $Z$ is roughly $\cint$ times smaller than $\wmax$. Packets with bigger windows are less likely to listen and hence less likely to change window size. 

The following two lemmas capture this intuition in two ways. First, packets that start with windows smaller than $Z$ are unlikely to grow to windows that are much larger than $Z$.  Second, packets with large windows (relative to $Z$) are unlikely to have their window sizes change by more than a small constant factor. 

% ========================================
% RESTATABLE USED HERE
% Leaving the old lemma statement to make it easier for us to read while editing the proof
% =========================================
%
\begin{lemma}[\textbf{Tail bound on growth of window size}] Consider any packet during an interval $\mathcal{I}$ with $\tau = |\mathcal{I}|$.  
    Let $Z$ satisfy $Z / \ln^2(Z) = \tau$.  And let $W\geq \wmin$ be the initial size of the packet's window. Let $W^+$ be the biggest window size the packet achieves during the interval. Then for large enough choice of constants $\wmin$ and $c$ and any $k\geq 2$:  if $W \leq kZ$, then 
         $$\Pr\left[W^+ > kZ\right] \leq 2^{-\Theta(c(\ln Z \cdot \ln k + \ln^2 k))}$$
or stated differently
     $$\Pr\left[\frac{W^+}{\ln^2(W^+)} > k\tau\right] \leq 2^{-\Theta(c(\ln \tau \cdot \ln k + \ln^2 k))} \,. $$
\label{lem:tailbound-maxwindow}\end{lemma} 
\begin{proof}
    In order for the packet to reach window size $kZ$, it must first reach window size $kZ/2$ and then listen at least $\Theta(c \ln(kZ))$ times. So we next upper-bound the expected number of times a packet with  window size at least $kZ/2$ listens. 
The expected number of times such a packet listens during the interval is 
    \begin{align*} \mu = O\left(c\cdot \frac{\ln^3(kZ)}{kZ}\cdot \tau\right) &= O\left(c\cdot \frac{\ln^3(kZ)}{kZ} \cdot \frac{Z}{\ln^2(Z)}\right)\\
    &= O\left(\frac{\ln^2(k) + \ln^2(Z)}{k\ln^2(Z)} \cdot c\ln(kZ)\right)\\
    &= O\left(\frac{\ln^2(k)}{k} \cdot c \ln(kZ)\right)\\
    &= O\left(\frac{1}{\sqrt{k}} \cdot c\ln(kZ)\right)
    \end{align*} 
    We can now Chernoff bound the number of times that the packet listens, while having window size at least $kZ/2$. Specifically, let $L_i$ be an indicator variable for the packet's $i$-th listen attempt during the interval with window size at least $kZ/2$, and let $L = \sum_i L_i$ be the total number of such listens. Above, we have a bound on $\E[L]$. Note that we can apply a standard Chernoff bound in this instance---indeed, we are concerned only with bounding the number of listens. And the adversary cannot force\footnote{This is in contrast with the bounds used in the proof of \lemref{AB_low_mid}, where Theorems \ref{thm:adversarial_azumas_lower} and \ref{thm:adversarial_azumas_upper} are required in place of a standard Chernoff bound due to the adversary's influence.} the listening probability of this particular packet to be larger than $c\ln^3(kZ/2) / (kZ/2)$ without also making its window smaller, removing it from consideration). Therefore $L$ is stochastically dominated by the sum of indicator variables with probability at most $O(c \ln^3(kZ)/kZ)$. This is precisely the expected value of which we have bounded above. Thus, by a standard Chernoff bound on $L$, the probability that the packet listens $\Omega(c \ln(kZ)) = \Omega(\sqrt{k} \mu)$ times is upper bounded by $(1/\sqrt{k})^{\Theta(c\ln(kZ))} = (1/2)^{\Theta(c\ln(k)\cdot\ln(kZ))}$.  That is,
    \[
        \Pr[W^+ > \Theta(kZ)] \leq 2^{-\Theta(c(\ln Z \cdot \ln k + \ln^2 k)} \ .
    \]
    Observe that $\Theta(kZ) = \Theta(k\tau\ln^2(Z))$ by definition of $Z$.  Moreover, $\ln(Z) = \Theta(\ln \tau)$. Thus, the second statement of the bound is just multiplying both sides of the inequality $W^+ > \Theta(k\tau\ln^2(Z))$ by $1/\ln^2(W^+) = \Theta(1/\ln^2(kZ))$ when talking about reaching window size $W^+ = \Theta(kZ)$, which cancels the $\ln^2(Z)$ on the right-hand side. 
 \end{proof}

% ========================================
% RESTATABLE USED HERE
% Leaving the old lemma statement to make it easier for us to read while editing the proof
% =========================================
%
The following corollary is typically used in our proofs with $S$ being the set of all packets, but sometimes it will be applied on a proper subset.
\begin{corollary}[\textbf{Tail bound on growth of window sizes}]
\label{cor:window-growth-union-bound}
Consider an interval $\mathcal{I}$ with $\tau=|\mathcal{I}|$. Let $S$ be a subset of the packets. Let $n$ be the number of distinct packets that are active at any point of the interval, 
and let $W^+$ denote the maximum window size reached by any packet in $S$ during the interval (and hence also an upper bound on the maximum window size at the end of the interval). 
Then for any $k\geq 2$: 
 $$\Pr\left[\frac{W^+}{\ln^2(W^+)} > k\tau\right] \leq n\cdot 2^{-\Theta(c(\ln \tau \cdot \ln k + \ln^2 k))}$$
\end{corollary}
\begin{proof}
This follows directly from Lemma~\ref{lem:tailbound-maxwindow} and a union bound over the at most $n$ packets in $S$.    
\end{proof}

% ========================================
% RESTATABLE USED HERE
% Leaving the old lemma statement to make it easier for us to read while editing the proof
% =========================================
%
% \begin{lemma}[\textbf{Bounds on the factor that a large window can grow/shrink}]\label{lem:window-change}
%     Consider any packet during an interval $\mathcal{I}$ with $\tau = |\mathcal{I}|$.  
%     Let $Z$ satisfy $Z / \ln^2(Z) = \tau$.  And let $W\geq \wmin$ be the initial size of the packet's window. Let $W^-$ be the smallest window size the packet has while still active in the interval, and let $W^+$ be the biggest window size the packet achieves during the interval. Then for large enough choice of constants $\wmin$ and $c$ and any constant parameter $\gamma>0$ and $k\geq 2$:\\
%     If $W = \Theta(kZ)$, then 
%     $$\Pr\left[W^+ \geq e^{\gamma} W\,\mbox{ or }\, W^- < W/e^{\gamma}\right] \leq 1/\tau^{\Theta(c \gamma \ln (\gamma k))}
%     $$    
% \end{lemma} 
\Wchange
\begin{proof}
    Consider a packet with window size $W = \Theta(kZ)$.  In order for the packet to grow or shrink its window by a factor of $e^{\gamma}$, it must listen to the channel $\Theta(\gamma \cdot c\ln(kZ))$ times.  Following the same argument as \lemref{tailbound-maxwindow},     
    the expected number of times the packet listens during the interval is $\mu = O(c\ln(kZ)/\sqrt{k})$.  Thus, the necessary number of listens to vary the window size by a factor of  $e^\gamma$  is $\Omega(\sqrt{k} \gamma \mu)$, which is $\Omega(\gamma \sqrt{k})$ times the expectation. Thus, from a standard Chernoff bound (made explicit below), we get that the probability of this event occurring is at most $(1/(\gamma \sqrt{k}))^{\Theta(c\gamma\ln(kZ))} \leq  (1/(\gamma \sqrt{k}))^{\Theta(c\gamma\ln \tau)} \leq 
    (1/\tau)^{c\gamma \ln (\gamma k)}$, where the first step follows because $\ln \tau = \Theta(\ln Z)$. Finally, the second centered equation of the lemma statement follows because a packet with window size $O(kZ)$ must first achieve window size $\Theta(kZ)$, at which point we can apply the first bound for constant $k$.
    
    In more detail, to see why a standard Chernoff bound applies, we let $L_i$ be an indicator variable for the packet's $i$-th listen attempt during the interval with window size at least $kZ/2$, and let $L = \sum_i L_i$ be the total number of such listens. Above, we have a bound on $\E[L]$. Note that we can apply a standard Chernoff bound in this instance;  we are concerned only with bounding the number of listens. Again, this is in contrast with the bounds used in the proof of \lemref{AB_low_mid}, where Theorems \ref{thm:adversarial_azumas_lower} and \ref{thm:adversarial_azumas_upper} are required in place of a standard Chernoff bound due to the adversary's influence. %And the adversary cannot force the listening probability of this particular packet to be larger than $c\ln^3(kZ/2) / (kZ/2)$ without also making its window smaller, removing it from consideration. 
   $L$ is stochastically dominated above and below by the sum of indicator variables with probability $\Theta(c \ln^3(kZ)/kZ)$. This is precisely the expected value of which we have bounded above.  
\end{proof}

% ========================================
% RESTATABLE USED HERE
% Leaving the old lemma statement to make it easier for us to read while editing the proof
% =========================================
%
% \begin{lemma}[\textbf{Tail bound on increase in \boldmath $\third(t)$}]
% \label{lem:tail-bound-third}
% Consider an interval $\mathcal{I}$ with length $\tau = |\mathcal{I}|$ starting from time~$t$ and ending at time $t' = t+\tau$.  Let $\mathcal{A}$ be the number of arrivals during the interval.
% Then for large-enough constant $c$ in the algorithm and any $k \geq 2$:
% \[
%     \Pr\left[\third(t') \geq \Theta(\mathcal{A} + k\tau)\right] \leq 2^{-\Theta(c(\ln \tau \cdot \ln k + \ln^2 k))}
% \]
% \end{lemma}
\Ltail
\begin{proof}
    Let $n_0=\first(t)$ be the number of packets that are alive at the start of the interval.  There are two cases: 
    
    {\bf Case 1: {\boldmath{$\tau^2 \geq \mathcal{A}$}}.} By definition of the interval size, $\tau > (1/\cint) \sqrt{n_0}$, and the total number of distinct active packets is $n = n_0 +\mathcal{A} = O( \tau^2)$.  Thus, by \corref{window-growth-union-bound}, the probability that any packet's window grows to $k\tau$ is $O(\tau^2)\cdot 2^{-\Theta(c(\ln \tau\cdot \ln(k) + \ln^2 k))} = 2^{-\Theta(c(\ln\tau\cdot \ln k+\ln^2k))}$. 
    
    {\bf Case 2: {\boldmath{$\mathcal{A} > \tau^2$ and hence also $n = \mathcal{A} + n_0 = O(\mathcal{A})$}}.}  Consider $k' = \mathcal{A}/\tau+k$, and note that $\ln(k') = \Theta(\ln(\mathcal{A}+k))$.  Again, applying \corref{window-growth-union-bound}, the probability that any packet's window grows to $k'\tau = \Theta((\mathcal{A}/\tau +k)\cdot \tau) = \Theta(\mathcal{A}+k\tau)$ is at most $O(\mathcal{A})\cdot 2^{-\Theta(c(\ln(k')\ln\tau + \ln^2(k'))))} = 2^{\Theta(\ln \mathcal{A}) - \Theta(c\ln(\mathcal{A}+k)\ln \tau + \ln^2(\mathcal{A}+k))} \leq 2^{-\Theta(c\ln(k)\ln \tau + \ln^2(k))}$.
\end{proof}

% ========================================
% RESTATABLE USED HERE
% Leaving the old lemma statement to make it easier for us to read while editing the proof
% =========================================
%
% \begin{lemma}[\textbf{Mostly low contention implies decrease in
% \boldmath $\third(t)$}]
% \label{lem:C-term}
% Consider an interval $\mathcal{I}$ starting at $t$ of length $\tau$, where $\tau = (1/\cint)\max\left\{L(t), \sqrt{\first(t)}\right\}$. 
% Let $t_1 = t + \tau$, and let $\mathcal{A}$ and $\mathcal{J}$ denote the number of packet arrivals and jammed slots, respectively, over $\mathcal{I}$. Then, with high probability
% in $\tau$, either 
% \begin{itemize}
%     \item $\third(t_1) \leq \third(t) / d + O(\mathcal{A})$,
%     % \mab{I think we do not need to include jammed slots in the additive term}
%     where $d>1$ is a positive constant, 
%         or
%     \item At least a $1/10$-fraction of the slots $t'$ in the interval $\mathcal{I}$ are either jammed or have contention $C(t') \geq \clow$. 
%     \end{itemize}  
%     Incorporating the fact that $\tau \geq \third(t)/\cint$, 
%    it follows that if at least a $9/10$ fraction of slots in the interval have contention at most $\clow$, then $(L(t_1)-L(t)) \leq O(\mathcal{A}+\mathcal{J})-\Omega(\tau)$.  
% \end{lemma}
\Cterm
\begin{proof}
We prove this using a case analysis.

{\bf Case 1: {\boldmath{$\Omega(\tau)$}} packet injections.} 
In this case, we argue that $\third(t_1) = O(\mathcal{A})$, with high probability in $\tau$, which satisfies the first bullet. In fact, by Lemma~\ref{lem:tail-bound-third}, the probability that $\third(t_1) = \Omega(\mathcal{A} + 2\tau) = \Omega(\mathcal{A})$ is no more than $2^{-\Theta(c\ln\tau+1)} = (1/(2\tau))^{\Theta(c)}$, where $c$ is a (large) constant specified in our algorithm.

{\bf Case 2: {\boldmath{$\sqrt{\first(t)} > L(t)$}}, i.e., \boldmath $\tau = (1/\cint) \sqrt{\first(t)})$. }
We will show that with high probability, contention is greater than $\clow$ in \emph{all} slots of the interval, and hence the second bullet is satisfied.  (In fact, contention is a fast-growing function of $\first(t)$, so it is also greater than $\chigh$ as long as $\first(t)$ is large enough, which can be controlled by tuning $\wmin$.) We start by noting that there are more than $\third(t)^2$ active packets (by case assumption that $\sqrt{\first(t)} > \third(t)$). Choose any $\third(t)^2$ packets that are active at the start of the interval.

By \corref{window-growth-union-bound}, with high probability in $\tau$, none of these packets ever reach window size above $O(\tau \ln^2(\tau)) = O(\first(t)^{1/2}\ln^2(\first(t)))$.  Moreover, at most $\tau = O(\sqrt{\first(t)})$ of the packets can succeed during the interval. Thus, the contention in every slot is at least $\Omega(\first(t)^{1/2}/\ln^2(\first(t))) \gg \clow$.

{\bf\boldmath Case 3: $\mathcal{A} = O(\tau)$ and $\first(t)^{1/2} \leq L(t)$, i.e., 
$\tau= (1/\cint) \cdot {\frac{\wmax}{\ln^2( \wmax)}}$.} 
We will argue that either the first or second bullet holds.  

We begin by arguing that packets with small windows initially are irrelevant. Observe that the total number of packets active during the interval is at most $O(\first(t)+\mathcal{A}) = O(\tau^2)$.  Thus, for $\cint$ sufficiently large%\footnote{Choosing large $\cint$ means that ``$Z$'' in the preceding lemmas is smaller relative to $\wmax$. Specifically, $\wmax$ is roughly $\cint$ times larger than $Z$. In other words, a window of size close to $\wmax$ corresponds to using $k$ close to $\cint$ in the lemmas.}
, we will see below how we can apply \lemref{window-change}, first for packets with small windows, and  then for packets with large windows, in order to ultimately get high probability bounds.

We apply the second bound of \lemref{window-change} to  the subset of packets $S$ with window at most $O(\wmax)$. In particular, with failure probability at most $1/\tau^{\Theta(1)}$, no packet with window smaller than say $(3/4)\wmax$ grows to window size larger than say $(7/8)\wmax$; here, we use our case assumption that $N(t)^{1/2} \leq \third(t)$.  (The particular choice of constants (3/4) and (7/8) is not important---any constant can be achieved by tuning $\gamma$ in \lemref{window-change}; e.g., $\gamma=1/7$ for the stated constants.)%\footnote{Tuning $\gamma$ also places restrictions on $c$ and $\cint$ as we would like $c,\cint \gg 1/\gamma$ and hence $c,k\gg 1/\gamma$ for the failure probability in \lemref{window-change} to be polynomially small in~$\tau$.}) \john{Probably should fix constants in Theta of cited Chernoff bound and make sure these match.}

For the remainder, we focus on any packet that starts the interval with window size at least $(3/4)\wmax$.  By \lemref{window-change}, with high probability such a packet maintains a windows size between say $(1/2)\wmax$ and $(6/5)\wmax$, meaning that the probability of listening does not vary by more than a constant factor. In particular, the packet chooses to listen to each slot with probability at least $(1/2) c\ln^3(\wmax)/\wmax$ and at most $(6/5)c\ln^3(\wmax)/\wmax$ . 

Now suppose that at least a (9/10)-fraction of the slots in $\mathcal{I}$ have low contention and are not jammed.  We next count two things: $x_s$ is the number of times the packet listens and hears silence, and $x_n$ is the number of times the packet listens and hears noise.  As long as $x_s - x_n \geq \delta(c\ln(\wmax))$, for constant $\delta$, then the packet's window shrinks by a factor of roughly $e^\delta$, which corresponds to the constant $d$ in the first bullet of the lemma statement.  We thus need only argue that $x_s - x_n = \Omega(c\ln(\wmax))$ with high probability (and then take a union bound across packets) to justify that such a constant exists. 

There is one subtle detail here: the packet's window has to be large enough that the lower bound we derive on $x_s$ below  would not result in the packet's window dropping below the minimum window size $\wmin$; it would be sufficient for the following if $\wmax \geq 2\wmin$, for example. We can enforce this constraint by simply setting $\clow \leq 1/(2\wmin)$, as even a single packet with window at most $2\wmin$ would contribute at least $\clow$ to the contention. 

We first argue that the packet is likely to observe many empty slots.  
By \lemref{p-emp} and the fact that $\clow\leq 1/\wmin \leq 1/20$, each of the unjammed low-contention slots is empty with probability at least $9/10$. Thus, the expected value of $x_s $ is at least:
\begin{align*}
& \mbox{~~~~~}\left(\frac{9}{10}\right) \left(\frac{1}{2}\right) \left(\frac{c\ln^3(\wmax)}{\wmax}\right)  \left(\frac{9}{10}\right)\left(\frac{1}{\cint}\right)\left(\frac{\wmax}{\ln^2(\wmax)}\right)\\
& = \left(\frac{81}{200}\right)\left(\frac{c}{\cint}\right) \ln(\wmax).
\end{align*}

Now we can apply an adaptive Chernoff bound to conclude that $x_s \geq (1/3)\cdot(c/\cint)\ln(\wmax)$ with high probability in $\wmax \geq \tau$.

We next consider the number of noisy slots, which has two components.  First, for slots where the contention exceeds $\clow$ or the slot is jammed, the worst case is that all such slots are noisy.  The expected number of listens to such slots is at most:
\begin{align*}
& \mbox{~~~~~}\left(\frac{6}{5}\right)\left(\frac{c\ln^3(\wmax)}{\wmax}\right) \left(\frac{1}{10}\right)\left(\frac{1}{c_{\tau}}\right)\left(\frac{\wmax}{\ln^2(\wmax)}\right)\\
    &= \left(\frac{6}{50}\right) \left(\frac{c}{c_{\tau}}\right) \ln(\wmax),
\end{align*}
\noindent which is at most $(1/8)(c/\cint)\ln(\wmax)$ with high probability.  There are also the slots where the contention is low but that happen to be noisy.  For low contention unjammed slots, there may also be noise. But the probability of noise is at most $(1/10)$.  The expected number of noisy unjammed low-contention slots that the packet listens to is at most:
\begin{align*}
& \mbox{~~~~~}\left(\frac{1}{10}\right)\left(\frac{6}{5}\right)\left(\frac{c\ln^3 (\wmax)}{\wmax}\right) \left(\frac{1}{c_{\tau}}\right)\left(\frac{\wmax}{\ln^2(\wmax)}\right)\\
   &= \left(\frac{3}{25}\right)\left(\frac{c}{c_{\tau}}\right)\ln(\wmax).
\end{align*}
%$(1/10)\cdot (6/5)c\ln^3(\wmax)/\wmax \cdot (1/\cint)\wmax/\ln^2(\wmax) = (3/25)\cdot(c/\cint)\ln(\wmax)$.  

Thus, by a Chernoff bound, with high probability, the number of such slots is at most $(1/8)(c/\cint)\ln(\wmax)$.  Adding these two together, we get that $x_n \leq (1/4)\cdot(c/\cint)\ln(\wmax)$, with high probability.  

In conclusion, with high probability, $x_s -x_n \geq (1/3 - 1/4) \cdot (c/\cint)\ln(\wmax) = (1/(12\cint)) \cdot c\ln(\wmax)$. 
\end{proof}

\subsection{Bounding the decrease in  \texorpdfstring{\boldmath{$\Phi(t)$}}{}  over intervals by combining   \texorpdfstring{$\first(t)$}{}, \texorpdfstring{\boldmath{$\second(t)$}}{}, and \texorpdfstring{\boldmath{$\third(t)$}}{}}
\label{sec:combining-analysis}

% \mab{new title: e.g., ``Combining analyses of  $\first(t)$, $\second(t)$, and $\third(t)$, to bound the decrease in  $\Phi(t)$  over intervals''. 
% Important piece previously missing from title: ``intervals''}

% \john{This entire subsection needs review.}

% \mab{That's right. All we really want in this section is a lemma saying the following: 
% \color{olive}{
% \begin{lemma}
% Consider an interval $I$, where $\tau=|I|$. \john{(Where $\tau$ is as specified above).}
% Say that there are $\mathcal{A}$ arrivals and $\mathcal{J}$ jammed slots 
% Then we have: \\
% (1) W.h.p.\ in $\tau$ $\Phi$ decreases by $\Omega(|I|)-  O(\mathcal{A}) - O(\mathcal{J})$.\\
% (2) There's a proper choice of constants $c'$ {\color{blue}{[or whatever]}}, so that the probability that $\Phi$ increases by $k\cdot \tau^2$  {\color{blue}{[is this a 2 or arbitrary poly]}}
% is less that $1/\tau^{c'} \cdot 2^{\Theta(\ln^2k)}$.   
% \end{lemma}
% }}

We now consider the change to $\Phi$ over an interval by combining the bounds for $\first$, $\second$, and $\third$.  The first lemma here states that with high probability, if there are not too many packet arrivals or jammed slots, then $\Phi$ is very likely to decrease by an amount that is proportional to the size of the interval. 
There is a low probability failure event in which the desired decrease to potential does not occur.  The second lemma focuses on such failures, providing a tail bound for the likelihood that $\Phi$ increases by a large amount. Together, these lemmas match the conditions necessary for the betting game. 

\phimain
\begin{proof}
    First, consider the high probability bounds on how much the terms in the potential increases during $\mathcal{I}$. 
    \begin{itemize}[leftmargin=13pt]
    \item In the worst case, $\first(t)$ increases by at most $\mathcal{A}$.
    \item By \lemref{AB_not_high_and_decrease},  $\alpha_1\first(t) + \alpha_2\second(t)$ increases by at most $O(\ln^3 \tau + \mathcal{A} + \mathcal{J})$, with high probability in~$\tau$.  
    \item By \lemref{tailbound-maxwindow}, $\third(t)$ increases by at most $O(\mathcal{A} + \tau)$, with high probability in~$\tau$. 
\end{itemize}   

We next consider two cases to conclude that the potential decreases overall. \\
{\bf Case 1:} at least $(9/10)$ fraction of the slots have low contention. Then, \lemref{C-term} states that with high probability in $\tau$ $\third(t)$ decreases by $\Omega(\tau) - O(\mathcal{A} + \mathcal{J})$.  Adding the increases summarized above to the other terms, we get a net decrease of $\Omega(\tau) - O(\mathcal{A} + \mathcal{J} + \ln^3\tau) = \Omega(\tau) - O(\mathcal{A} + \mathcal{J})$.\\
{\bf Case 2:} at most a $(9/10)$ fraction of the slots have low contention. Then, \lemref{AB_not_high_and_decrease} states that with high probability in $\tau$$\alpha_1\first(t) + \alpha_2\second(t)$ decreases by at least $\Omega(\tau) - O(\mathcal{A} + \mathcal{J})$.  Choosing $\alpha_1$ and $\alpha_2$ much larger than $\alpha_3$
the $\Omega(\tau)$ decrease to the first two terms dominates the $O(\tau)$ increase to the third term.  
\end{proof}

 Our next lemma provides a tail bound on the amount by which $\Phi(t)$ increases, which is relevant if the high probability bound in the preceding lemma fails.  At first glance, the reader may be surprised to see that the amount of jamming does not occur in the lemma statement. This omission is in part due to the fact that any jamming is subsumed by the $\Theta(k\tau^2)$ term because there can be at most $\tau$ jammed slots. 

% ========================================
% RESTATABLE USED HERE
% Leaving the old lemma statement to make it easier for us to read while editing the proof
% =========================================
%
% \begin{theorem}[\textbf{\boldmath Tail bound on increase in $\Phi(t)$ over interval $\mathcal{I}$}]
% % [\textbf{Tail bound on increase in \boldmath $\Phi(t)$ over  interval ]
%  \thmlabel{tail-bound}
% Consider an interval $\mathcal{I}$ of length $|\mathcal{I}|=\tau$ starting at time $t$ and ending at time $t'$. Let $\mathcal{A}$ be the number of packet arrivals in $\mathcal{I}$. Then 
% the probability that $\Phi$ increases by at least $\Theta(\mathcal{A}) + \Theta(k\tau^2)$ 
% is at most $2^{-\Theta(c(\ln \tau \cdot \ln k + \ln^2 k))} \leq  (1/\tau^{\Theta(c)})\cdot 2^{-\Theta(\ln^2 k)}$, where $c$ is the constant parameter of the algorithm.  That is, 
% \[ \Pr\left[(\Phi(t')-\Phi(t)) \geq  \Theta(\mathcal{A}) + \Theta(k\tau^2)\right] \leq \left(\frac{1}{\tau^{\Theta(1)}}\right) \cdot 2^{-\Theta(\ln^2 k)} \ .
% \]
% \end{theorem}
\tailbound
\begin{proof}
    In fact, we will upper-bound $\Phi(t')$, which implies an upper bound on $\Phi(t')-\Phi(t)$ because $\Phi$ is always nonnegative.  We begin by providing worst-case upper bounds on $\first(t')$ and $\second(t')$.  We complete the proof by incorporating the tail bound on $\third(t')$.  

    Observe that $\first(t') \leq \first(t) + \mathcal{A}$.  Because $\tau \geq (1/\cint)\sqrt{\first(t)}$, we have $\first(t') \leq (c_{\tau}\tau)^2 + \mathcal{A} = O(\tau^2) + \mathcal{A}$.   
    To bound $\second$, notice that each packet active at time $t'$ contributes at most $1/\ln(\wmin) < 1$ to $\second(t')$. Thus, $\second(t') < \first(t') = O(\tau^2) + \mathcal{A}$.    

    Finally, by \lemref{tail-bound-third},  the probability that $\third(t')$ exceeds $\Theta(\mathcal{A} + k\tau)$ is upper bounded by $2^{-\Theta(c(\ln \tau \cdot \ln k + \ln^2k))}$.   And when $\third(t')$ falls below the stated threshold, we have $\Phi(t') \leq \alpha_1 \cdot O(\mathcal{A}+\tau^2) + \alpha_2 \cdot O(\mathcal{A}+\tau^2) + \alpha_3\cdot O(\mathcal{A} + k\tau) = O(\mathcal{A} + \tau^2 + k\tau) = O(\mathcal{A} + k\tau^2)$.
\end{proof}

\subsection{Using the analysis of \texorpdfstring{\boldmath{$\Phi(t)$}}{} to prove throughput via a betting-game argument}
\seclabel{bettinggame}

The preceding section establishes progress guarantees over sufficiently large intervals in the form of \thmreftwo{phi-main}{tail-bound}. 
Specifically, consider an interval $\mathcal{I}$ starting at slot $t$ with length  
$|\mathcal{I}|= \tau = (1/\cint) \cdot \max\big\{\frac{\wmax}{\ln^2( \wmax)}, \first(t)^{1/2}\big\}$.
%\mab{Pls fix.}
%\mab{We've been saying that the length of the interval is $\tau$. }
Then, with high probability in $|\mathcal{I}|$,  $\Phi(t)$ decreases over $\mathcal{I}$ by at least $\Omega(\tau) - O(\mathcal{A}+\mathcal{J})$, where $\mathcal{A}$ is the 
the number of packet arrivals and $\mathcal{J}$ is the number of jammed slots in $\mathcal{I}$.
Critically, these bounds hold with high probability in  $\tau$.  In this section we show how to apply these lemmas to achieve  results for the execution with high probability in $N$, the total number of packets.

Since the adversary is adaptive, we have to be careful in
combining bounds across intervals. 
The adversary can use the results of earlier intervals in choosing new arrivals and jamming, which affects the size of later intervals.  This may lead to a more advantageous interval length for the adversary, which might increase its probability of success.

% \iffalse
% In the hands of an adaptive adversary, the limited nature of these bounds poses a challenge as we try to show a decrease in potential over a sequence of intervals. Specifically,  the adaptive adversary exerts {\it some} control over the size of each interval (since the size is proportional to $\third(t)$, which is impacted by packet injections and jamming), complicating the analysis of $\Phi(t)$. For instance, even if \lemref{phi-main} holds for some ``large'' interval,  the adversary may attempt to offset the resulting decrease in potential by arranging several subsequent ``small'' intervals, each of which has a reasonable chance of increasing the potential.   
% %While we may suspect that $\Phi(t)$ still decreases over the sequence of intervals, analyzing this process requires care. 
% \fi

To reason about this process, we reframe it in a setting that resembles a random walk, which we describe below in a \defn{betting game}. Our analysis of this game then allows us to analyze the implicit throughput (recall Section~\ref{sec:model}). We first summarize the betting game and then show how it corresponds to the backoff process. 

\paragraph{The Betting Game.} There is a \defn{bettor} who makes a series of bets.  Each bet has a size, chosen arbitrarily by the bettor to any value $\geq \wmin$. At any given time, the bettor has some amount of money, which is initially 0 dollars. When the bettor loses a bet, the bettor loses some money, and when the bettor wins, the bettor wins some money. (The amounts won or lost are specified below as a function of the size of the bet.)  

Additionally, at any time, the bettor may choose to receive a \defn{passive income}---and in fact must do this at time zero (in order to have some money to bet).  The passive income is (immediately) added to the bettor's wealth. The total amount of passive income taken, however, means that the bettor must play the game longer.  The game begins when the bettor first takes some passive income, and the game does not end until either: (i) the bettor goes broke, or (ii) the bettor has resolved bets totaling $\Omega(\passive)$ size, where $\passive$ is the passive income received, whichever comes first.  The bettor's goal is to complete the game without going broke. Importantly, although the bettor can always choose to take more passive income, doing so increases the total play time. 

The bettor loses a size-$s$ bet with probability at least $1-\frac{1}{\poly(s)}$, and hence wins a size $s$ bet with probability $O(1/\poly(s))$.  If the bettor loses the size-$s$ bet, it loses $\Theta(s)$ dollars. If the bettor wins the bet, it gets $\Theta(s^2)$ dollars, plus $Y$ \defn{bonus dollars}, where $Y$ is a random variable such that $\Pr[Y \geq k s^2] \leq  \frac{1}{\poly(s)}\cdot 2^{-\Theta(\ln^2k)}$.

%This effect is analogous to implicit throughput. 
%\mab{Specifically, if the bettor cannot help going broke before the end of the game, then it cannot force the implicit below some target value. In contrast, the if bettor wins the game, this corresponds to pushing the implicit throughput below the guarantee that we want. }

\paragraph{ Correspondence to the Backoff Process.} The betting game mirrors the backoff process. The adversary corresponds to the bettor. Each bet corresponds to an interval, i.e., each bet is of size $\tau$ corresponding to an interval's length.  (In the actual backoff process, the interval sizes are dictated by the current state of the system, and not entirely under the control of the adversary.). Money corresponds to potential.  Passive income during a bet corresponds to the potential caused by arrivals and jammed slots during the interval.   

%  Each bet has a size equal to the duration $\tau$. The bettor also has some amount of money, which is initially 0 dollars. When the bettor loses a bet, the bettor loses some money, and when the bettor wins, the bettor wins some money. (The amounts won or lost are specified below as a function of the size of the bet.)  Additionally, at any time, the bettor may choose to receive a \defn{passive income}.  The passive income is added to the bettor's wealth. The total amount of passive income taken, however, means that the bettor must play the game longer.  The game begins when the bettor first takes some passive income, and the game ends when either the bettor goes broke or the bettor has resolved bets totaling $\Omega(\passive)$ size, where $\passive$ is the passive income received, whichever comes first.  The bettor's goal is to complete the game without going broke. Importantly, although the bettor can always choose to take more passive income, doing so increases the total play time. 

The parameters set above for the betting game correspond exactly to the bounds we have set on the potential function in our contention resolution setting.  Recall that an interval of length $\tau$ is successful with probability $1 - \frac{1}{\poly(\tau)}$, corresponding to the bettor losing (see \thmref{phi-main}).  When an interval of length $\tau$ is successful (i.e., the bettor loses), the potential (i.e., money) decreases by $\Omega(\tau) - O(\mathcal{A} + \mathcal{J})$, with the latter term corresponding to passive income taken during the interval.  When an interval of length $\tau$ is unsuccessful (i.e., the bettor wins), the potential (i.e., money) increases by $\Theta(\tau^2)$, plus $Y$ \defn{bonus dollars}, where $Y$ is a random variable such that $\Pr[Y \geq k\tau^2] \leq  \frac{1}{\poly(\tau)}\cdot 2^{-\Theta(\ln^2k)}$; these bonus dollars correspond exactly to the tail bound of \thmref{tail-bound}.

%a size-$\tau$ bet with probability at least $1-\frac{1}{\poly(\tau)}$.
%If the bettor  loses the size-$\tau$ bet, it loses $\Theta(\tau)$ dollars.
%This loss corresponds to the high-probability event (in $\tau$) of \thmref{phi-main}. 
%The bettor wins a size-$\tau$ bet  with probability %$O(1/\poly(\tau))$.
%If the bettor wins the bet, it gets $\Theta(\tau^2)$ dollars, plus $Y$ \defn{bonus dollars}, where $Y$ is a random variable such that $\Pr[Y \geq k\tau^2] \leq  \frac{1}{\poly(\tau)}\cdot 2^{-\Theta(\ln^2k)}$; these winnings correspond to the tail bound of \thmref{tail-bound}.
%(Of course, during each bet, the bettor may also gain passive income for arrivals and jammed slots.)
%\mab{To determine: can you mount that you win or lose in a bit? Just be in order notation or do we need to talk about constants?}

%\paragraph{Discussion.} We give the bettor the power to choose arbitrary bet sizes (subject to a minimum interval size, which itself is determined by $\wmin$), and the bettor is even allowed to place bets whose loss would cause the bettor to end with negative money.  

%\mab{The \defn{implicit throughput} at time $t$ is defined as $N_t/S_t$, where $N_t$ is the total number of packets that arrive at or before time~$t$ and $S_t$ is the total number of active slots so far. }

The rules of betting game are set in favor of the bettor such that when the bettor wins, the potential $\Phi(t)$ increases more slowly than the bettor's wealth increases, and when the bettor loses, $\Phi(t)$ decreases at least as fast as the bettor's wealth decreases. Therefore, this betting game stochastically dominates the potential function. 

\emph{The takeaway is that at any point $t$, the bettor's wealth is an upper bound on $\Phi(t)$.} 
Because $\Phi(t)$ is an upper bound on the number of packets in the system, the bettor going broke corresponds to all packets succeeding. We thus obtain good implicit throughput, because there must either be many jammed slots or packet arrivals, or there must be many packets succeeding, leading to inactive slots. 

%\paragraph{Upper bounding the bettor's maximum wealth/potential and 
\paragraph{Showing $\Omega(1)$ implicit throughput.}  We will use bounds on the bettor's maximum wealth to imply good implicit throughput.  In \lemref{passive}, we provide a high-probability upper bound on the bettor's maximum wealth and the amount of time until they go broke; the time at which they go broke corresponds to there being no packets in the system.  We use these conclusions about the betting game to establish Corollary \ref{cor:implicit-throughput}, which proves $\Omega(1)$ implicit throughput. Finally, Corollary \ref{cor:phi-bound-corollary} bounds the potential function in terms of the number of packet arrivals and jammed slots (which, in terms of the betting game, correspond to passive income). 

% ========================================
% RESTATABLE USED HERE
% Leaving the old lemma statement to make it easier for us to read while editing the proof
% =========================================
%
 \begin{lemma}\label{lem:passive}
 ({\bf The bettor loses the betting game}). Suppose the bettor receives $\passive$ dollars of passive income.   Then with high probability in $\passive$, the bettor never has more than $O(\passive)$ dollars during the game.  Moreover, the bettor goes broke within $O(\passive)$ time, with high probability in $\passive$. 
 \end{lemma}
%\BGpassive 
\begin{proof}
Group the bets into \defn{bet classes}, where bet class $j$ contains all bets with size in $[2^{j}, 2^{j+1})$. 
% for $2^{j}\geq \minBet$. \mab{$\tau_{\rm min}$? May need to define this. } 
%(W.l.o.g, we can make $\minBet$ a power of $2$.) 

\paragraph{Passive income.} For the purpose of determining how much passive income the bettor takes (and hence how long to play the game), we will let the adversary have a limited view of the future: for each bet class, the bettor can see the sequence of outcomes for bets in that bet class; with this information the bettor can plan how many bets to allocate to each class and how much passive income to take.

Let $\passive$ be the total passive income that the adversary chooses to take as a result of this information; we will generously allow the adversary to take that passive income at the very beginning of the game.

\iffalse
We will let the the bettor see the outcome of all of the bets of each size that it could make going into the future. After seeing this, we allow the adversary to 
\mab{because the adversary gets to see the outcome of all potential events going infinitely far into the future, without loss of generality, it can decide how it wants to adapt at the outset. Then w.l.o.g. we can be generous to the adversary, and let it get all its money up front. }
select the amount of money it wishes to have up front.
\fi

We will show w.h.p.\ in $\passive$, that the bettor goes broke over the first $O(\passive)$ slots, and never has more than $O(\passive)$ dollars.

%Looking ahead, the big bet classes will be easy to deal with. 

\paragraph{Large bets.} Before proceeding further, we will dispose of the large bets.  A bet of size $\passive^{1/4}$ or larger is \defn{big}, and  otherwise the bet is \defn{small}.  With high probability in $\passive$, the bettor loses \emph{every} big bet: If the bet size is at least $\passive^{1/4}$, then the probability of winning the bet is $O(1/\poly(\passive^{1/4})) = O(1/\poly(\passive))$ (by the rules of the betting game, where the degree in the $\poly$ is our choice.) Since there can only be  $O(P^{3/4})$ lost bets of this size before going broke, by a  union bound, none of them are successes with high probability.

\paragraph{Analyzing a bet class.} We now consider the bets in a single bet class.  We will imagine that the bets in that bet class are an infinite stream of possible bets, and that the adversary can see the outcomes in advance when choosing whether or not to make the next bet.  (This is to account for the adaptive nature of the adversary that it might choose to make bets of a different size based on past history.)  We will show that every prefix of this stream of bets is ``good'' with high probability.

\paragraph{Bonus dollars.}  First, we will examine the bonus dollars.  We will show that, with high probability, for any prefix of the stream of bets: if the bettor wins $X$ regular dollars, then it also wins $O(X)$ bonus dollars, with high probability in $\passive$.  Recall that for any single bet in bet class $[2^{j}, 2^{j+1})$, the bonus dollars $Y$ are drawn from a distribution where $\Pr[Y > k \cdot 2^{j+1}] \leq 1/2^{\Theta(j+\ln^2 k)}$, by the rules of the game.

We consider the first $O(\ln(\passive))$ successful bets in each bet class of size at most $\passive^{1/4}$.  The probability of winning at least $\Theta(\passive^{1/2})$ bonus dollars is polynomially small in $\passive$, and there are at most $O(\ln(\passive))$ such winning bets.  So we conclude that the bonus dollars from all such bets is at most $O(\passive)$ with high probability in $\passive$.

Next, we fix a specific bet class $j$ and consider some fixed prefix of the bet sequence for this bet class.  Let $B$ be the set of bets that the bettor wins; we assume $|B| > \Theta(\ln{\passive})$, as we have already accounted above for the first few successful bets.  We want to examine each possible value of $k$ for the bettor's ``bonus multiplier,'' hence for each $k \in \mathbb{Z}^+$ define the following indicator random variables: for the $i$th bet in $B$, let $x_i^{(k)} = 1$ iff the $i$th bet in $B$ yields at least $k 2^{j+1}$ bonus dollars. (Otherwise, $x_i^{(k)} = 0$.) 

The total number of bonus dollars won from bets within this bet class is at most
$\sum_{i=1}^{|B|} \sum_{k=1}^\infty  2^{j+1} x_i^{(k)}$ (note that the inner summation acts as a $k$ multiplier). 
As we only want to consider values of $k$ that are a power of 2, we can round up the bonus dollar winnings and conclude that the bonus dollars are bounded by:
$2^{j+2} \sum_{k \in \{2^i \mid i \in \mathbb{Z}_{\geq 0}\}} \sum_{i=1}^{|B|} kx_i^{(k)}$.
%Let $X^{(k)} = \sum_{i=1}^{|B|} x_i^{(k)}$ denote the sum of all $x_i^{(k)}$.  

%= 2^{j+1} \sum_{k=1}^\infty \sum_{i=1}^{|B|} x_i^{(k)}. \]

%By rounding all $k$ up to the next power of $2$, we get that 
%\[2^{j+1} \sum_{k=1}^\infty \sum_{i=1}^L x_i^{(k)} \leq 2^{j+2} \sum_{k \in \{2^i \mid i \in %\mathbb{Z}_{\geq 0}\}} \sum_{i=1}^L k x_i^{(k)} \]

Given the probability of winning bonus dollars, we know that:
$$\expect{x_i^{(k)}} \leq 2^{-\Theta(j+\ln^2{k})} \leq 2^{-\Theta(\ln^2{k})}.$$   

We conclude that: 
\begin{align*}
   \expect{ \sum_{i=1}^{|B|} k x_i^{(k)}} & \leq k|B|2^{-\Theta(\ln^2{k})}\\ & \leq |B|2^{-\Theta(\ln^2{k})}.  
\end{align*}
\noindent where the last line follows by wrapping up $k$ with the  $\Theta(\ln^2 k)$ term in the exponent. Thus, the expected total bonus dollars for the bet class is: 
$$2^{j+2} \sum_{k \in \{2^i \mid i \in \mathbb{Z}_{\geq 0}\}} |B|2^{-\Theta(\ln^2{k})} = O(2^{j+2}|B|).$$

By the second bound given in Theorem~\ref{thm:adversarial_azumas_upper}, we conclude that the total bonus dollars for the bet class is $O(2^{j+2}|B|)$ with  high probability in $\passive$ (as $|B| = \Omega(\ln \passive)$). 
%at least probability $1 - e^{-\Theta(|B|)}$, where $|B| > \Theta(\ln\passive)$.  
Since this holds with high probability in $\passive = \Omega(|B|)$, it also holds with high probability in $\passive$ for all prefixes of the bet sequence via a union bound.

Since $|B|$ bets in bet class $j$ make at least $\Theta(2^j |B|)$, and with high probability in $\passive$ there are at most $O(2^j |B|)$ bonus dollars, we conclude that with high probability we can bound the total bonus dollars in terms of: (i) $O(\passive)$ bonus dollars (for the first few bets), plus (ii) $O(X)$ bonus dollars for a bet class that that wins $X$ dollars.

\begin{comment}
\iffalse
\john{double check this last statement about ev is true}. We know that:
\begin{align*}
    E[X^{(k)}] & \leq k \sum_{i=1}^L E[x_i^{(k)}]\\
    & \leq L\,k\,(1/2)^{d\ln^3 k}\\
    & \leq L\,k\,(1/k)^{d\ln^2 k}\\
      & \leq L\,(1/k)^{d\ln^2 (k) - 1}\\
      &= L 2^{-\Theta( \ln^3 k)}.
\end{align*}
\noindent where $d$ is a positive constant.

\john{W.h.p., no bet will yield more bonus dollars than $e^{{(c \ln I)}^{1/3}}$.  We define $x_i^{(k)}$ as an indicator random variable that is $1$ iff the $i$th bet wins at least $k 2^j$ dollars. Otherwise, $x_i^{(k)}$ is $0$. }
\fi
\end{comment}

\paragraph{Wins versus losses.}  Now we analyze the total winnings over any prefix of the bets in a given bet class.  

The first observation is that the first $O(\ln(\passive))$ bets result in at most $O(\ln(\passive))$ wins, and hence at most $O(\passive^{1/4}\ln(\passive)) = O(\passive)$ income.  We will see that the bettor never makes more money than $O(\passive)$ income with high probability.

Next, we will show that once $\Omega(\ln(\passive))$ bets have been made, w.h.p.\ in $\passive$: at all subsequent points in the betting sequence (for that bet class), all but at most an $\alpha$-fraction of bets are losses, where $\alpha$ is a constant that we can choose by changing the polynomial. 

More precisely, recall that the bettor wins a bet in this bet class with probability $1/\poly(2^j) \leq \alpha/(2(1 +\delta))$ (where $\delta > 0$ is a parameter used in \thmref{adversarial_azumas_upper}).  Thus in the prefix containing $|B|$ bets, the expected number of wins by the bettor is at most $|B|\alpha/(2(1 + \delta))$.  We conclude by the second bound of \thmref{adversarial_azumas_upper} that the bettor wins at most $\alpha |B|$ bets with %probability at most $e^{-\Theta(|B|)}$. 
high probability in $\passive$.
Thus with high probability in $\passive$, it holds that the bettor wins at most $\alpha |B|$ bets for every prefix of bets (of size at most $\passive$) in the bet class that is longer than $\Omega(\ln(\passive))$. 

Imagine, for example, that the bettor makes $|B|$ bets in bet class of sizes $[2^{j}, 2^{j+1})$.  If $|B| = \Omega(\ln(\passive))$, then we conclude that (with high probability in $\passive$) the bettor loses $\Theta(2^j)(1-\alpha)|B| - \Theta(2^{j+1})\alpha |B|$---remembering that we have already shown that if the bettor wins $X$ regular dollars in this prefix, then it also wins only $O(X)$ bonus dollars.  Thus for a proper choice of $\alpha$, the bettor loses money for every choice of $|B|$.  

Specifically, if the bettor chooses to bet $X$ in a small bet class and those bets consists of more than $O(\ln(\passive))$ bets, then with high probability the bettor loses $\Theta(X)$ dollars.

\paragraph{Wrapping up the proof.}  Thus we conclude that after bets totaling $\Theta(\passive)$ either $\Theta(\passive)$ of the bets are large---and the bettor goes bankrupt---or $\Theta(\passive)$ of the bets are small.  In the latter case, at least $\Theta(\passive)$ of those bets must not consist of ``the first $O(\ln(\passive))$ bets in a bet class'', since there are only $\ln^2(\passive)$ such bets.  Thus with high probability the bettor loses $\Theta(\passive)$ dollars and goes bankrupt.
\end{proof}

\begin{corollary}[{\bf Constant implicit throughput}]\label{cor:implicit-throughput}
At the $t$th active slot (in the execution of contention resolution), the implicit throughput is $\Omega(1)$ with high probability in~$t$.
\end{corollary}
\begin{proof}
We explain here how \lemref{passive} implies implicit throughput.  At time $t$ in the betting game, with high probability in $t$, the bettor must have taken a passive income of  $\Omega(t)$ dollars---otherwise the bettor would have gone broke, as per \lemref{passive}.

Consider a time $t$, and suppose that the bettor has received $\passive$ dollars from passive income. %, for constant $c$ matching the \mbox{big-$O$} of the lemma. 
Then from \lemref{passive}, with high probability in $P$, the bettor goes broke within $\eta \cdot P$ time, for a constant $\eta$ matching the \mbox{big-$O$} of \lemref{passive}.  This corresponds to there being no active packets at time $\eta P$. Setting $P = t/\eta$, we see that at least $t/\eta$ packet arrivals/jammed slots are necessary to obtain $t$ active slots. We thus obtain the $\Omega(1)$ implicit throughput result.
%Suppose that the bettor has received $t$ dollars of passive income.  This results in at most $c\cdot t$ active slots for some constant $c$ w.h.p.\ in $t$ by \lemref{passive}. Similarly, if the bettor has received $t/c$ dollars from  passive income, then there are at most $t$ active slots w.h.p.\ in $t/c$ (which is w.h.p.\ in $t$).  So, in the $t$-th active slot, there need to be at least $\Omega(t)$ packets injected w.h.p.\ in $t$, and thus we obtain the corollary.
\end{proof}
 
\noindent We conclude by noting that the above corollary yields \thmref{infinite_arrival_throughput}.  Moreover, Corollary~\ref{cor:finite_arrival_throughput} is a special case of Theorem \ref{thm:infinite_arrival_throughput}. 

Our next corollary just restates part of \lemref{passive} with respect to the potential.  Due to \thmreftwo{phi-main}{tail-bound}, the potential increases less than the bettor's net worth does, and it decreases at least as fast.  That is, the betting game stochastically dominates the  potential.  
\begin{corollary}[{\bf Bound on potential in terms of packet arrivals and jammed slots}]
Suppose that we have had $N_t$ packet arrivals and $\mathcal{J}_t$ jammed slots up until slot $t$. Then  w.h.p.\ in 
$N_t+\mathcal{J}_t$, for all slots $t'=1,\ldots, t$, the potential $\Phi(t')=O(
N_t+\mathcal{J}_t)$.
\label{cor:phi-bound-corollary}
\end{corollary}

\paragraph{Bounds on throughput for infinite and finite streams.} We next address the case of adversarial queuing arrivals with granularity $S$, and where the arrival rate $\lambda$ is a sufficiently small constant. Recall that the ``arrival rate'' $\lambda$ limits the number of packet arrivals and jammed slots to at most $\lambda S$ in every interval of $S$ consecutive time slots.

We start with \lemref{queuing_nearby_inactive}, showing that, in the case of adversarial queuing arrivals with a sufficiently small arrival rate, there is an inactive slot within any interval of length $S$. This is used to prove Corollary \ref{cor:backlog}, a high-probability bound on the number of packets in the system under an adversarial queuing arrival model, which is the same as \corref{backlog-intro}. 

%a lemma that translates implicit throughput in adversarial queuing into a bound on how far away the nearest inactive slot can be.

%We next prove the second bullet of \thmref{infinite_arrival_throughput}, which pertains to 

%We start with a lemma that translates implicit throughput in adversarial queuing into a bound on how far away the nearest inactive slot can be.

\begin{lemma}[{\bf Existence of inactive slot}]\lemlabel{queuing_nearby_inactive}
    Consider an input stream of adversarial queuing arrivals with granularity $S$ and sufficiently small arrival rate $\lambda=O(1)$.  Consider also any particular time slot~$t$. Then with high probability in $S$:
    \begin{itemize}[noitemsep]
        \item 
        there is an inactive slot at some time $t'$ with $t-\Theta(S) \leq t' < t$, and 
        \item 
    there is an inactive slot at some time $t'$ with $t < t' \leq t+\Theta(S)$.
    \end{itemize} 
\end{lemma}

\begin{proof}
    Consider a time $t$, and let $t' < t$ be the last inactive slot before time $t$.  Our goal is to prove that $t - t' = O(S)$, with high probability in~$S$

    Let us first consider the implication of large $t-t' \geq S$ on implicit throughput.  Then we have a time period where at most $\lambda(t-t'+S) \leq 2\lambda(t-t')$ packet arrivals or jamming occur. Since there are packets active throughout this entire interval from $t'$ to $t$ by construction, the implicit throughput is at most $2\lambda (t-t') / (t-t') = O(\lambda)$ at time $t$.  If $\lambda$ is small enough, this level of throughput does not meet the $\Omega(1)$ implicit throughput promised by \corref{implicit-throughput}, which contradicts the assumption that $t-t'$ is large.  We can thus apply \corref{implicit-throughput} at time $t$ to conclude that with high probability in $S$, the $\Omega(1)$ implicit throughput is met at $t$, and hence $t-t' = O(S)$.  

    To achieve the second claim, we apply the same argument at time $t+\Theta(S)$, which implies an inactive slot between $t$ and $t+\Theta(S)$, w.h.p.
\end{proof}

We can leverage  the first bullet point of \lemref{queuing_nearby_inactive} to obtain the following bound on an adversary with granularity $S$. (And we use the second bullet later in Theorem~\ref{thm:queuing_adversarial_accesses}; although it immediately follows from the first bullet point, it is still useful to note explicitly.)

\begin{corollary}[\bf Bounded backlog for adversarial-queuing arrivals]
  Consider an input stream of adversarial queuing arrivals with granularity $S$ and sufficiently small arrival rate $\lambda=O(1)$.  Consider any particular slot $t$.  Then with high probability in $S$, the number of packets in the system at time $t$ is $O(S)$.
  \label{cor:backlog}
\end{corollary}

\begin{proof}
    By \lemref{queuing_nearby_inactive}, there is an inactive slot at time $t-O(S)$ with high probability in $S$.  If that is the case, there can be at most $O(S)$ packet arrivals plus jammed slots since the last inactive slot. 
\end{proof}

This completes the proof Corollary~\ref{cor:backlog-intro}. Again, since Corollary \ref{cor:finite_arrival_throughput} is a special case of Theorem \ref{thm:infinite_arrival_throughput}, we have proved our central results on throughput.

\subsection{Channel access/energy bounds}\seclabel{energy}

\paragraph{Energy bounds for finite streams.}
Our first theorems apply for finite instances.

% ========================================
% RESTATABLE USED HERE
% Leaving the old lemma statement to make it easier for us to read while editing the proof
% =========================================
%
% \begin{theorem}\label{thm:finite_adaptive_accesses}
% Consider an input stream with $N$ packets and $\jams$ jammed slots. Assume that the adversary is adaptive but not reactive. Then w.h.p.\ in $N+\jams$, a packet accesses the channel at most $O(\ln^4(N+\jams))$ times.
% \end{theorem}
\finiteenergy

\begin{proof}
From Corollary \ref{cor:finite_arrival_throughput}, the number of active slots is $O(N+\jams)$. 
We also know from Corollary~\ref{cor:phi-bound-corollary} that the potential function is $O(N+\jams)$ throughout the entire execution w.h.p.\ in $N+\jams$.

One consequence is that for all $t$, we can bound the largest window size ever reached, denoted by $W^*$. In particular, at each time $\wmax = O(\poly(N+\jams))$ w.h.p., and hence $W^* = O(\poly(N+\jams))$ w.h.p.\ (In fact, for all $t$, $\wmax = O( (N+\jams) \ln^2(N+\jams))$ w.h.p., but we do not need to be that exact.)  

First, let us count the number of times a packet can back off before it reaches a window size of $W^*$, assuming no backons.  Suppose a packet has window size $w$.  Then each time it backs off, its window increases in size by factor of $(1+1/\Theta(\ln w))$.  Thus, there are $\Theta(\ln w)$ backoffs before the window size doubles. Rounding $w$ up to the maximum $W^*$, we get $O(\ln(W^*))$ backoffs to double the window size and $O(\ln^2(W^*))$ backoffs to reach maximum window size. 

In order for a packet to perform additional backoffs, it must also back on as otherwise its window would increase beyond $W^*$.  In particular, for each additional backoff, the packet must also perform $\Omega(1)$ backons.

Next, we shall bound $k$, the number of backons that the packet performs, with high probability. Since each backon allows for $O(1)$ additional backoffs, the total number of backoffs is now given by $O(\ln^2(W^*)+k)$.  And the number of listening attempts is the number of backoffs plus backons, and also $O(\ln^2(W^*) +k)$. 

The key point is that each time a backon occurs, the slot is empty. So the only question is how many times the packet can listen to empty slots before it chooses to also sends.  Every time a packet listens, there is an $\Omega(1/\ln^3 W^*)$ probability that it also makes a sending attempt.  The probability that the packet never sends on these empty slots is $(1-\Omega(1/\ln^3W^*))^k$, which is polynomially small in $W^*$ for $k=\Theta(\ln^4 W^*)$.  Thus with high probability in $W^*$, we have $k = O(\ln^4 W^*$).

To conclude the proof, we substitute in that $W^*=O(\poly(N+\jams))$ w.h.p.
\end{proof}

\begin{theorem}[{\bf Energy bounds for the finite case against an adaptive and reactive adversary}]\label{thm:finite_reactive_accesses}
Consider an input stream with $N$ packets and $\mathcal{J}$ jammed slots. Assume that the adversary is adaptive and reactive. Then, w.h.p.\ in $N+\jams$
\begin{itemize}
\item  a packet accesses the channel at most $O( (\jams+1) \ln^3(N+\jams) + \ln^4(N+\jams))$ times. 
\item the average number of channel accesses is $O((\jams/N + 1) \ln^4(N+\jams))$ times.
\end{itemize}
\end{theorem}

\begin{proof}
We first explain the worst-case bound. Recall that with an adaptive adversary, the jammer commits at the start of a slot whether that slot will be jammed. In contrast, the reactive adversary does not decide to jam a slot until after it sees which packets have chosen to send during that slot.  That is, the reactive adversary may reserve all of its jamming for slots when a packet would have otherwise succeeded.  Recall also that the reactive adversary cannot react to listening---it can only react to sending. 

We revisit the proof of \thmref{finite_adaptive_accesses} at the point that reactivity has an impact---bounding $k$, the number of backons, before the packet succeeds.  Each time a packet observes an empty slot, it also chooses to send with probability $\Omega(1/\ln^3 W^*)$. Notice that if the packet observes $\rho\ln^3 W^*$ empty slots, then it would choose to send in $\Omega(\rho)$ of them in expectation. But the adversary can only block $\jams$ of these, so choosing $\rho=\Omega(\jams)$ is large enough to get more than $\jams$ sends, and hence a send that is not jammed, in expectation. 
To achieve a high probability result, we apply a Chernoff bound: with failure probability exponentially small in $\rho$, the packet chooses to send $\Omega(\rho)$ times.
Choosing $\rho = \Theta(\jams + \ln W^*)$ for some sufficiently large constant, we get that with high probability $W^*$ the number of sends is more than $\jams$.  Thus, with high probability in $W^*$, $k \leq \rho\ln^3 W^* = O((\jams+\ln W^*) \ln^3W^*)$. 

It is an additional step to compute the average bound. Now, collectively, there must be $N$ successes. But since the jammer can block up to $\jams$ transmissions that otherwise would succeed, there need to be enough listening attempts so that there would be $N+\jams$ successes without the jammer. The bounds follow immediately by a similar argument.
\end{proof}

\paragraph{Infinite streams: with adversarial queuing and in general.}
We first turn to infinite streams with adversarial-queuing arrivals.  
The main tool is \lemref{queuing_nearby_inactive}, which allows us to transform the adversarial queuing case into finite instances that are not very large.

\begin{theorem}[{\bf Energy bound for adversarial-queuing arrivals against an adaptive adversary}]\label{thm:queuing_adversarial_accesses}
    Consider an input stream of adversarial queuing arrivals with granularity $S$ and sufficiently small arrival rate $\lambda=O(1)$. Consider an adaptive but not reactive adversary. Then w.h.p.\ in $S$, a particular packet accesses the channel at most $O(\ln^4 S)$ times.
\end{theorem}

\begin{proof}
    Consider a particular packet $p$ injected at time $t$. Now let us consider the latest time slot $t_0<t$ with no active packets and the earliest time slot $t_1>t$ with no active packets (and hence packet $p$ has also finished).  We thus have a finite instance with $N + \jams \leq t_1-t_0$.  By \lemref{queuing_nearby_inactive}, $t_1-t_0 = O(S)$, with high probability in $S$.  The current claim then follows from \thmref{finite_adaptive_accesses}.
\end{proof}

We next consider an adversary that is both adaptive and reactive for adversarial queuing.  The argument is similar, except that now we apply the finite case of \thmref{finite_reactive_accesses}.

\begin{theorem}[{\bf Energy bounds for adversarial-queuing arrivals against an adaptive and reactive adversary}]\label{thm:queuing_reactive_accesses}
Consider an input stream of adversarial-queuing arrivals with granularity $S$ and sufficiently small arrival rate $\lambda=O(1)$. Consider a reactive and adaptive adversary. 
\begin{itemize}[noitemsep]
\item W.h.p.\ in $S$ a packet accesses the channel at most $O(S)$ times. 
\item W.h.p.\ in $S$, the average number of channel accesses per slot is $O(\ln^4 S)$, averaging over all the slots. 
\end{itemize}
\end{theorem}

\begin{proof}
For the first claim, we leverage the second part of \lemref{queuing_nearby_inactive}.  That is, consider a packet injected at time~$t$.  Then, with high probability in $S$, there is an inactive slot before time $t+\Theta(S)$, and hence the packet finishes by time $t+O(S)$.  Even if it accesses the channel in every slot during which it is active, the packet only makes $O(S)$ accesses.

For the second claim, we reduce to the finite case, but the way we do this is slightly different as we want to be able to average the total number of accesses across a reasonably large interval. Consider times $t_i = i\cdot\Theta(S)$, and consider a particular interval from time $t_i$ to time $t_{i+1}$. Let $t_i'<t_i$ be the latest inactive slot before $t_i$ and $t_{i+1}'>t_{i+1}$ be the earliest inactive slot after $t_{i+1}$.  Then we consider the interval from $t_i'$ to $t_{i+1}'$ as a finite instance. By \lemref{queuing_nearby_inactive}, the total size of this interval is $\Theta(S)$ with high probability in $S$.   There are thus $O(S)$ packets and $O(S)$ jamming.  We can then apply \thmref{finite_reactive_accesses} to this finite instance to conclude that the total number of channel accesses is thus at most $O(S \ln^4(S))$.  Even if all of these accesses fall in the subinterval from $t_i$ to $t_{i+1}$, each slot in that subinterval receives $O(\ln^4(S))$ accesses on average.
\end{proof}

We now turn to general infinite streams.  
% ========================================
% RESTATABLE USED HERE
% Leaving the old lemma statement to make it easier for us to read while editing the proof
% =========================================
%
% \begin{theorem}\label{thm:infinite_accesses}
%     Suppose that up until time $t$ there have been $N_t$ packet arrivals and $\jams_t$ jammed slots.  
%     \begin{itemize}[noitemsep]
%         \item  Consider an adaptive adversary that is not reactive. Then w.h.p.\ in $\jams_t+N_t$, each packet  makes $O(\ln^4(\jams_t+N_t))$ channel accesses before time $t$. 
%         \item Consider and adaptive adversary that is reactive.  Then w.h.p.\ in $\jams_t+N_t$, a particular packet accesses the channel at most $O( (\jams_t+1)\ln^3(N_t+\jams_t) + \ln^4(N_t+\jams_t))$ times. Moreover, the average number of channel accesses is $O((\jams_t/N_t + 1)\ln^4(N_t+\jams_t))$.
%     \end{itemize}   
% \end{theorem}
\infiniteenergy

\begin{proof}
    Since we are only looking at events occurring before time $t$, we essentially have a finite instance. That is, consider an extension to the execution after time $t$ where the adversary performs no further jams or injections. Then \thmref{finite_adaptive_accesses} and \thmref{finite_reactive_accesses} apply, and they achieve the bounds claimed here except across the entire extended execution. Truncating the execution at time $t$ only reduces the number of channel accesses.
\end{proof}
 % ending conference omission of infinite streams

\section{Conclusion and Future Work}

We have provided a simple contention-resolution algorithm that achieves constant throughput with full energy efficiency (i.e., low sending {\it and} listening complexity), despite a powerful jamming adversary. This resolves in the affirmative two open questions about whether full-energy efficiency is possible {\it at all} in the popular ternary-feedback model, and whether it remains possible in the presence of jamming.

There are some natural directions for future work. The ability to prioritize network traffic is increasingly important to modern systems. For example, some real-time applications must guarantee message delivery within a bounded amount of time~\cite{WirelessHART,RoweMaRa06}; additionally, assigning levels of priority to packets enhances fairness by preventing starvation~\cite{4411881,4509751,1286161}. Agrawal et al.~\cite{DBLP:conf/spaa/AgrawalBFGY20} examine contention resolution when each packet has a deadline by which it must succeed. However, their result is not listening efficient and  handles only a form of random jamming. It may be interesting to explore whether jamming by a stronger adversary can be tolerated in a fully energy-efficient manner, where packets may be late, but only as a (slow-growing) function of the amount of jamming.

Another open problem is achieving similar guarantees in a multi-hop setting. Here, a major challenge is that packets listening in the same slot may receive different feedback. This aspect corresponds to the properties of signal propagation; for example, network participants may perceive different channel states depending on their distance from a sender. Richa et al.~\cite{richa:jamming2,richa:competitive-j} derive results for contention resolution in a multi-hop setting with jamming; however, their approach is not fully energy-efficient. The methods we employ here to achieve full energy efficiency would need revision in order to cope with the issue of packets receiving different channel feedback in the same slot.

We note that \OurAlg is not guaranteed to be fair; that is, it is possible for some packets to succeed quickly, while others linger in the system for longer; this is especially pertinent for the infinite case.  Can \OurAlg be revised to provide a guarantee on fairness?

Finally, it may be interesting to consider variable transmission strengths. For example, in certain cases where device hardware permits, jamming may be overcome by transmitting with sufficient energy. This would lead to a more complex cost model, where sending at a higher energy imposes a higher cost on the sender. 

%%%%%%%%%%%%%%%%%%%%%%%%%%%%%%%%%%%%%

%\bibliographystyle{siamplain}
%\bibliography{backoff,bender}

\end{document}